\documentclass[]{article}
\usepackage{amsmath} 
\usepackage{amsthm} 
\usepackage{latexsym} 
\usepackage{amssymb}

\usepackage[usenames]{xcolor}
\usepackage{latexsym}
\usepackage{amssymb}
\usepackage{amsmath}
\usepackage{graphicx}
\usepackage{tikz}
\usetikzlibrary{calc}
\usetikzlibrary{arrows}
\usetikzlibrary{decorations.pathmorphing}
\usetikzlibrary{patterns}
\usepackage{enumerate}
\newtheorem{theorem}{Theorem}
\newtheorem{lemma}[theorem]{Lemma}
\newtheorem{corollary}[theorem]{Corollary}

\allowdisplaybreaks

\renewcommand{\phi}{\varphi}

\newcommand{\SECT}{Sec.~}

\newcommand{\tseq}[1]{\mathfrak{#1}}
\newcommand{\intermediate}[1]{\dot{#1}}
\newcommand{\inner}[1]{\ddot{#1}}

\newcommand{\pcpW}{\mathsf{W}} 
\newcommand{\pcpw}{\mathsf{w}} 

\newcommand{\partition}{\mathsf{part}}  
\newcommand{\colourComp}{\mathsf{colour}} 
\newcommand{\cpartition}{\mathsf{sc\text{-}part}}
\newcommand{\noncontact}{\mathsf{notC}}
\newcommand{\noncontacti}{\mathsf{K5m}}
\newcommand{\stack}{\mathsf{stack}}
\newcommand{\frameFla}{\mathsf{frame}}
\newcommand{\stacki}{\ti{\mathsf{stack}}}
\newcommand{\frameFlai}{\ti{\mathsf{frame}}}

\newcommand{\Sat}{\textit{Sat}}
\newcommand{\ti}[1]{{#1}^\circ}
\newcommand{\tc}[2][]{{#2}^{-_{#1}}}
\newcommand{\ic}{c^\circ}

\newcommand{\md}[2][] {{\lfloor#2\rfloor_{#1}}}

\newcommand{\ConRC}{{\sf ConRC}}
\newcommand{\RegC}{{\sf RC}}
\newcommand{\RC}{{\sf RC}}
\newcommand{\RCP}{{\sf RCP}}

\newcommand{\R}{\mathbb{R}}
\newcommand{\bS}{\mathcal{S}}
\newcommand{\cK}{\mathcal{K}}%
\newcommand{\cL}{\mathcal{L}}%

\newcommand{\ExpTime}{\textsc{ExpTime}}
\newcommand{\NP}{\textsc{NP}}
\newcommand{\PSpace}{\textsc{PSpace}}

\newcommand{\cB}{\ensuremath{\mathcal{B}}}%
\newcommand{\cBc}{\ensuremath{\mathcal{B}c}}%
\newcommand{\cBci}{\ensuremath{\mathcal{B}{\lowercase{c}}^\circ}}%
\newcommand{\cBC}{\ensuremath{\mathcal{C}}}%
\newcommand{\cBCc}{\ensuremath{\mathcal{C}c}}%
\newcommand{\cBCci}{\ensuremath{\mathcal{C}c^\circ}}%
\newcommand{\RCCE}{\ensuremath{\mathcal{RCC}8}}%
\newcommand{\RCCF}{\ensuremath{\mathcal{RCC}5}}%
\newcommand{\BRCCE}{\ensuremath{\mathcal{BRCC}8}}%
\newcommand{\SF}{\ensuremath{\mathcal{S}4}}%
\newcommand{\SFU}{\ensuremath{\mathcal{S}4_u}}%

\begin{document}
\title{Topological Logics with Connectedness over Euclidean Spaces}
\author{\and Roman Kontchakov
\and Yavor Nenov
\and Ian Pratt-Hartmann
\and Michael Zakharyaschev}
\date{}

\def\pkg#1{{\texttt{#1}}}
\let\cn=\pkg
\def\latex{\LaTeX}
\def\amslatex{{\protect\AmS-\protect\LaTeX}}
\def\latextwoe{\LaTeX2\raisebox{-1pt}{$\epsilon$}}
\def\cmit{\fontfamily{cmr}\fontshape{it}\selectfont}

\maketitle

\begin{abstract}
We consider the quantifier-free languages, $\cBc$ and $\cBci$,
obtained by augmenting the signature of Boolean algebras with a unary
predicate representing, respectively, the property of {\em being
  connected}, and the property of {\em having a connected
  interior}. These languages are interpreted over the regular closed
sets of $\R^n$ ($n \geq 2$) and, additionally, over the regular closed
{\em polyhedral} sets of $\R^n$.  The resulting logics are examples of
formalisms that have recently been proposed in the Artificial
Intelligence literature under the rubric {\em Qualitative Spatial
  Reasoning}.  We prove that the satisfiability problem for $\cBc$ is
undecidable over the regular closed polyhedra in all dimensions
greater than 1, and that the satisfiability problem for both languages
is undecidable over both the regular closed sets and the regular
closed polyhedra in the Euclidean plane. However, we also prove that
the satisfiability problem for $\cBci$ is \NP-complete over the
regular closed sets in all dimensions greater than 2, while the
corresponding problem for the regular closed polyhedra is
\ExpTime-complete.  Our results show, in particular, that spatial
reasoning over Euclidean spaces is much harder than reasoning over
arbitrary topological spaces.
\end{abstract}

\section{Introduction}

Let \cBc{} be the quantifier-free fragment of first-order logic in the
signature\linebreak
$(+,\cdot,-,0,1,c)$, where $c$ is a unary predicate; and let
$\RCP(\R^n)$ be the collection of regular closed, polyhedral sets in
$n$-dimensional Euclidean space.  (A set is {\em regular closed} if it
is the closure of an open set, and {\em polyhedral} if it is a finite
union of finite intersections of closed half-spaces.)  The collection
$\RCP(\R^n)$ forms a Boolean algebra under the subset ordering; and we
interpret \cBc{} over $\RCP(\R^n)$ by taking the symbols
$+,\cdot,-,0,1$ to have their natural meanings in this Boolean
algebra, and by taking $c$ to denote the property of \emph{being
  connected}.  Intuitively, we think of elements of $\RCP(\R^n)$ as
regions of space, and of formulas of \cBc{} as descriptions of
arrangements of these regions. Our primary concern is the {\em
  satisfiability problem}: given a \cBc-formula, is there an
assignment of elements of $\RCP(\R^n)$ to its variables making it
true?

The motivation for studying this problem comes from the field of {\em
  Qualitative Spatial Reasoning} in Artificial Intelligence, the aim
of which is to develop formal languages for representing and
processing qualitative spatial information.  In this context, \cBc{}
constitutes a parsimonious language: it has no quantifiers, and its
non-logical primitives express only the basic region-combining
operations and the property of connectedness. At the same time, the structures $\RCP(\R^n)$---particularly
in the cases $n=2$ and $n=3$---constitute its most natural domains of
interpretation, given current practice in the fields of
Qualitative Spatial Reasoning, Geographic Information Systems and
Spatial Databases.

For reasons discussed below, we broaden the subject of enquiry
slightly.  Let \cBci{} denote the quantifier-free fragment of
first-order logic in the signature $(+,\cdot,-,0,1,\ic)$, where
$+,\cdot,-,0,1$ are as before, and $\ic$ is a unary predicate
interpreted as the property of {\em having a connected
  interior}. Further, let $\RC(\R^n)$ denote the collection of regular
closed sets in $n$-dimensional Euclidean space. Again, $\RC(\R^n)$
forms a Boolean algebra under the subset ordering, having $\RCP(\R^n)$
as a sub-algebra. Intuitively, we think of $\RC(\R^n)$ as a more
liberal model of spatial regions than $\RCP(\R^n)$. In the sequel, we
consider the satisfiability problem for \cBc{} and \cBci{} over the
structures $\RC(\R^n)$ and $\RCP(\R^n)$. The results of this paper are
as follows: ({\em i}) the satisfiability problem for \cBc{} over
$\RCP(\R^n)$ is undecidable for all $n \geq 2$; ({\em ii}) the
satisfiability problem for \cBc{} over $\RC(\R^2)$ is undecidable, as
are the satisfiability problems for \cBci{} over both $\RC(\R^2)$ and
$\RCP(\R^2)$; ({\em iii}) the satisfiability problem for \cBci{} over
$\RC(\R^n)$ is \NP-complete for all $n \geq 3$, while over
$\RCP(\R^n)$ the corresponding problem is \ExpTime-complete. (It may
be of interest to note that, over $\RC(\R)$ and $\RCP(\R)$, the
satisfiability problem for \cBc{} and \cBci{} is \NP-complete.)  The
decidability of the satisfiability problems for \cBc{} over
$\RC(\R^n)$, for $n \geq 3$, is left open.  Results ({\em ii}) and
({\em iii}) were announced, without proofs, in~\cite{KPZ10kr, KNPZ11ijcai}.

Mathematically, it is also meaningful to consider the satisfiability
of \cBc- and \cBci-formulas over the regular closed subsets of {\em
  any} topological space. If $T$ is a topological space, we denote the
collection of regular closed subsets of $T$ by $\RC(T)$; again, this
collection always forms a Boolean algebra under the subset
ordering. The satisfiability problem for $\cBc$ over the class of
structures of the form $\RC(T)$ is known to be \ExpTime-complete,
while for \cBci{}, the corresponding problem is
\NP-complete~\cite{iscloes:kp-hwz10, KPZ10kr}.
However,
satisfiability over {\em arbitrary} topological spaces is of at most
marginal relevance to Qualitative Spatial Reasoning. Indeed, the
results reported here show that, for languages able to express the
property of connectedness, reasoning over {\em Euclidean} spaces is a
different kettle of fish altogether.  In the remainder of this
section, we discuss the significance of these results in the context
of recent developments in spatial, algebraic and modal logics.

\subsection{Spatial logic}
\label{sec:sl}

A spatial logic is a formal language interpreted over some class of
geometrical structures. Spatial logics, thus understood, have a long
history, tracing their origins back both to the axiomatic tradition in
geometry~\cite{hilbert09,Tarski59} and also the region-based theory of
space~\cite{Whitehead29,deLaguna22}, subsequently developed
in~\cite{KK:Clarke81,KK:Clarke85,BiacinoGerla91}. Such logics were
proposed as a formalism for Qualitative Spatial Reasoning in the
seminal paper~\cite{Randelletal92}. The basic idea is as follows:
numerical coordinate-based descriptions of the objects that surround
us are hard to acquire, inherently error-prone, and probably
unnecessary for everyday spatial reasoning tasks; therefore---so goes
the argument---we should employ a representation language whose
variables range over spatial regions (rather than points), and whose
non-logical primitives are interpreted as qualitative (rather than
quantitative) relations and operations.  On this view, formulas are to
be understood as expressing descriptions of (putative) configurations
of objects in space, with the satisfiability of a formula over the
space in question equating to the geometrical realizability of the
described arrangement. If we imagine an intelligent agent employing
such a language to represent spatial arrangements of objects, then the
problem of recognizing satisfiable formulas amounts to that of
eliciting the geometrical knowledge latent in that agent's operating
environment and cognitive design.

The best-known, and most intensively studied, qualitative spatial
representation language is
\RCCE~\cite{Egenhofer&Franzosa91,Randelletal92,Smith&Park92}. This
language features predicates for the six topological relations
$\mathsf{DC}$ (disconnection), $\mathsf{EC}$ (external connection),
$\mathsf{PO}$ (partial overlap), $\mathsf{EQ}$ (equality),
$\mathsf{TPP}$ (tangential proper part) and $\mathsf{NTPP}$
(non-tangential proper part) illustrated, for the case of closed
disc-homeomorphs, in Fig.~\ref{fig:ccircles}.  (The name \RCCE{}
becomes less puzzling when we observe that the relations
$\mathsf{TPP}$ and $\mathsf{NTPP}$ are asymmetric.)  Note that \RCCE{}
has no individual constants or function symbols, and no
quantifiers.
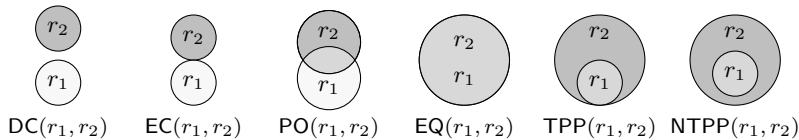
\begin{figure}[ht]
\begin{center}
\begin{tikzpicture}[scale=0.6]\small
\draw[fill=gray!50] (0,0.7) circle (0.5);
\draw[fill=gray!10,fill opacity=0.5] (0,-0.5) circle (0.5);
\node at(0,0.7) {$r_2$};
\node at(0,-0.5) {$r_1$};
\node at(0,-1.5) {\footnotesize ${\mathsf{DC}}(r_1,r_2)$};
%
\draw[fill=gray!50] (3,0.5) circle (0.5);
\draw[fill=gray!10,fill opacity=0.5] (3,-0.5) circle (0.5);
\node at(3,0.5) {$r_2$};
\node at(3,-0.5) {$r_1$};
\node at(3,-1.5) {\footnotesize ${\mathsf{EC}}(r_1,r_2)$};
%
\draw[fill=gray!50] (6,0.4) circle (0.7);
\draw[fill=gray!10,fill opacity=0.5] (6,-0.4) circle (0.7);
\draw (6,0.4) circle (0.7);
\node at(6,0.6) {$r_2$};
\node at(6,-0.6) {$r_1$};
\node at(6,-1.5) {\footnotesize ${\mathsf{PO}}(r_1,r_2)$};
%
\draw[fill=gray!50] (9,0) circle (1);
\draw[fill=gray!10,fill opacity=0.5] (9,0) circle (1);
\node at(9,0.4) {$r_2$};
\node at(9,-0.4) {$r_1$};
\node at(9,-1.5) {\footnotesize ${\mathsf{EQ}}(r_1,r_2)$};
%
\draw[fill=gray!50] (12,0) circle (1);
\draw[fill=gray!10,fill opacity=0.5] (12,-0.5) circle (0.5);
\node at(12,0.6) {$r_2$};
\node at(12,-0.5) {$r_1$};
\node at(12,-1.5) {\footnotesize ${\mathsf{TPP}}(r_1,r_2)$};
%
\draw[fill=gray!50] (15,0) circle (1);
\draw[fill=gray!10,fill opacity=0.5] (15,-0.3) circle (0.5);
\node at(15,0.6) {$r_2$};
\node at(15,-0.3) {$r_1$};
\node at(15,-1.5) {\footnotesize ${\mathsf{NTPP}}(r_1,r_2)$};
\end{tikzpicture}
\end{center}
\caption{\RCCE-relations over discs in $\R^2$.}
\label{fig:ccircles}
\end{figure}
Traditionally, \RCCE{} is interpreted over the regular closed sets in
some topological space, so as to finesse the awkward issue of whether
regions should be taken to include their boundary points. The
satisfiability problems for $\RCCE$ over the class of structures of
the form $\RC(T)$ (for $T$ a topological space) is easily seen to be
\NP-complete, though tractable fragments have been
explored~\cite{Renz&Nebel01,aledThesis}. Further, satisfiability of an
$\RCCE$ formula over any structure in this class implies
satisfiability over $\RCP(\R^n)$, for all $n \geq
1$~\cite{Renz&Nebel99}. Thus, the satisfiability problems for
$\RC(\R^n)$ and $\RCP(\R^n)$ ($n \geq 1$) coincide and are
\NP-complete---a fact which testifies to the restricted expressive
power of $\RCCE$.

A word of caution is in order at this point.  Satisfiability of an
\RCCE-formula over $\RC(\R^2)$ does not necessarily imply
satisfiability by {\em natural} or {\em familiar} regions---for
example, closed disc-homeomorphs.  The $\RCCE$-satisfiability problem
for such interpretations requires specialized, and highly
non-trivial, techniques. A landmark result~\cite{iscloes:sss03} in the
area shows, however, that the satisfiability problem for \RCCE{}
interpreted over the closed disc-homeomorphs in $\R^2$ is still in
\NP. The contribution of present paper, with its emphasis on Euclidean
spaces and the property of connectedness, imposes severe limits on
what further results of this kind we can hope for.

We mentioned above that, if $T$ is a topological space, the collection
$\RC(T)$ always forms a Boolean algebra under the subset
ordering. This enables us to extend \RCCE{} with the function symbols
$+$, $\cdot$, $-$ and constants $0$, $1$, interpreting these in the
natural way over any structure $\RC(T)$. Such an extended language was
originally introduced in~\cite{Wolter&Z00ecai} under the name \BRCCE{}
({\em Boolean} \RCCE).  Intuitively, if $a_1$ and $a_2$ are regular
closed sets, we may think of $a_1 + a_2$ as the agglomeration of $a_1$
and $a_2$, $a_1 \cdot a_2$ as the common part of $a_1$ and $a_2$,
$-a_1$ as the complement of $a_1$, $0$ as the empty region and $1$ as
the whole space. The satisfiability problem for \BRCCE{} over the
class of structures of the form $\RC(T)$ is still \NP-complete;
however, restricting attention to {\em connected} spaces $T$ yields
a \PSpace-complete satisfiability problem. Thus, \BRCCE{}, unlike
\RCCE{}, has sufficient expressive power to distinguish between
satisfiability over arbitrary spaces and satisfiability over connected
spaces. But that is about as far as this extra expressive power takes
us: satisfiability of a \BRCCE-formula over any structure $\RC(T)$,
for $T$ connected, implies satisfiability over $\RCP(\R^n)$ for all $n
\geq 1$. Hence, the satisfiability problems for $\RC(\R^n)$ and
$\RCP(\R^n)$ ($n \geq 1$) coincide, and are \PSpace-complete. Note
in particular that \BRCCE{} does not enable us to say that a given
{\em region} of space is connected.

We end this discussion of \RCCE{} and \BRCCE{} with a remark on the
absence of quantification from these languages. This restriction is
motivated by computability considerations: essentially all
region-based spatial logics with full first-order syntax have
undecidable satisfiability problems, and so are considered unsuitable
for Qualitative Spatial
Reasoning~\cite{KK:Grzegorczyk51,KK:Dornheim98kr,Davis06,Lutz&WolterLMCS}.
To be sure, first-order spatial logics are nevertheless of
considerable model-theoretic interest; see~\cite{PH:HSL} for a
survey. We note in particular that, if we can quantify over regions,
then the \RCCE{}-primitives easily enable us to define, over most
interesting classes of interpretations, all of the primitives $+$,
$\cdot$, $-$, $0$, $1$, $c$ and $\ic$. However, as computability
considerations are to the fore in this paper, we too confine
ourselves to quantifier-free formalisms in the sequel.



\subsection{Algebraic and modal logic}
\label{sec:aml}

The standard view of
topology takes a topological space to consist of a set of points on
which a collection of open subsets is defined. However, a dual view is
possible, in which one begins with a Boolean algebra, and then adds
algebraic structure defining distinctively topological relations
between its objects. There are two main approaches to developing this
second view.  On the first, we think of the underlying Boolean algebra
as a field of sets, and we augment this Boolean algebra with a pair of
unary operators, conceived of as representing the operations of {\em
  closure} and {\em interior}, and assumed to obey the standard
Kuratowski axioms~\cite{McKinsey&Tarski44}.  The striking similarity
between these axioms and the axioms for the propositional modal logic
\SF~\cite{Orlov28,Godel33a} led to the development of topological
semantics for modal logics. Under this semantics, the (propositional)
variables are taken to range over any collection of subsets of a
topological space (not just regular closed sets), and the logical
connectives are interpreted by the operations of union, intersection,
complement and topological interior (for necessity) and closure (for
possibility). The extension of this language with the universal
modality, denoted \SFU~\cite{Goranko&Passy92}, is known to be a
super-logic for \RCCE{} and
\BRCCE~\cite{Bennett94,Renz&Nebel97,Cohnetal97,Nutt99,Wolter&Z00ecai}. The
satisfiability problem for \SFU{} is the same over every connected,
separable, dense-in-itself metric space, and this problem is
\PSpace-complete~\cite{McKinsey&Tarski44,Shehtman99,Arecesetal00}.
We remark that, as for \RCCE{} and \BRCCE{}, \SFU{} is
unable to express the condition that a region is connected. For a
survey of the relationship between spatial and modal logics see
\cite{Benthem&Guram07hb,Gabelaiaetal05,KPWZ08aiml} and references
therein.

On the second approach, we instead think of the underlying Boolean
algebra as an algebra of {\em regular closed} sets, and we augment
this Boolean algebra with a binary predicate $C$, conceived of as
representing the relation of {\em contact}. (Two sets are said to be
{\em in contact} if they have a non-empty intersection). This binary
predicate is assumed to satisfy the axioms of {\em contact algebras},
a category which is known to be dual to the category of dense
sub-algebras of regular closed algebras of topological
spaces~\cite{DW05,DV1,DV2,Balbianietal07,Tinchev&Vakarelov10,Vakarelov07}.
The contact relation as a basis for topology actually has a venerable
career, having originally been introduced in~\cite{Whitehead29} under
the name `extensive connection'.  More relevantly for the present
paper, it is straightforward to show that all the \RCCE{} relations
can be expressed, in purely propositional terms, using this
signature~\cite{Balbianietal07,KPWZ08aiml}.  (Thus, for example,
${\mathsf{EC}}(\tau_1,\tau_2)$ is equivalent to $C(\tau_1,\tau_2)
\wedge (\tau_1 \cdot \tau_2 = 0)$.)  For this reason, we regard the
propositional language over the signature $(+,\cdot,-,0,1,C)$, here
denoted \cBC, as equivalent to the language \BRCCE{} mentioned above.
The purely Boolean fragment of \cBC{} (without the contact predicate
$C$) is denoted by \cB. This language is in fact equivalent to the
extension of the spatial logic \RCCF~\cite{Bennett94} with the
function symbols $+$, $\cdot$ and $-$.

%


\subsection{Spatial logics with connectedness}
\label{sec:lwc}

Most spatial regions of interest---plots of land in a
cadastre, the space occupied by physical objects, paths swept out by
moving objects---are either connected or at least contain few
connected components~\cite{Cohn&Renz08}. It seems, therefore, that to
be genuinely useful, logics for Qualitative Spatial Reasoning should
possess some means of expressing this notion. The simplest way of
proceeding is to consider languages featuring a unary predicate
denoting this property. Various such languages have been investigated
before~\cite{PrattHartmann02,KPWZ08lpar,iscloes:kp-hwz10,Vakarelov07,Tinchev&Vakarelov10};
the language \cBc{} is chosen for study here because it is so
parsimonious.

It is worth bearing in mind, however, that `connectedness,' in the
topologists' sense may not be exactly what we want. For example, a
region consisting of two closed discs externally touching is, in this
sense, connected, yet, in certain contexts, may be functionally
equivalent to a disconnected region. (Imagine having a garden that
shape.)  In such contexts, it may be more useful to employ the notion
of a region's having a connected {\em interior}, a property we refer
to as {\em interior-connectedness}. Note that every regular closed,
interior-connected set is connected; also, in the space $\R$, the
notions of connectedness and interior-connectedness coincide.  So as
not to prejudge the issue here, we employ predicates for both notions:
$c$ to denote the standard property of connectedness, $\ic$ to denote
the property of interior-connectedness.  Hence, in addition to the
`minimal' language \cBc, we have its counterpart \cBci.

Strikingly, the languages \RCCE{} and \BRCCE, which cannot represent
connectedness (or, for that mater, interior-connectedness), are far
less sensitive to the underlying geometrical interpretation than the
languages \cBc{} and \cBci, which can.  For example, an
\RCCE-formula that is satisfiable over the regular closed algebra of
{\em any} topological space is satisfiable over $\RCP(\R^n)$, for all
$n \geq 1$~\cite{Renz98}.  Or again, a \BRCCE-formula that is
satisfiable over the regular closed algebra of any {\em connected}
topological space is satisfiable over $\RCP(\R^n)$, for all $n \geq
1$~\cite{Wolter&Z00ecai}. Thus, \RCCE{} and \BRCCE{} care neither
about the dimension of the (Euclidean) space we are reasoning about,
nor about the distinction between regular closed polyhedra and
arbitrary regular closed sets. Not so with the languages \cBc{} or
\cBci, which are sensitive both to the dimension of space and to the
restriction to polyhedral regular closed sets. This sensitivity is
easy to demonstrate for \cBci, and we briefly do so here, by way of
illustration.

Consider first sensitivity to dimension. The \cBci-formula
\begin{align}\label{eq:1vs2}
\bigwedge_{1 \leq i \leq 3}\hspace*{-0.5em} \bigl(\ic(r_i) \wedge (r_i \neq 0)\bigr) \ \ \land\ \
\bigwedge_{1 \leq i < j \leq 3}\hspace*{-0.5em}
          \bigl(\ic(r_i + r_j) \wedge (r_i \cdot r_j = 0)\bigr)
\end{align}
`says' that $r_1$, $r_2$ and $r_3$ are non-empty regions with
connected interiors, such that each forms an interior-connected sum
with the other two, without overlapping them. It is obvious that this
formula is not satisfiable over $\RC(\R)$. For the non-empty,
(interior-) connected regular closed sets on the real line are
precisely the non-punctual, closed intervals, and
it is impossible for
three such intervals to touch each other without overlapping. On the
other hand,~\eqref{eq:1vs2} is easily seen to be satisfiable over
$\RC(\R^n)$ for all $n \geq 2$. Likewise, the \cBci-formula
\begin{align}\label{eq:2vs3}
\bigwedge_{1 \leq i \leq 5}\hspace*{-0.5em} \bigl(\ic(r_i) \wedge (r_i \neq 0)\bigr) \ \ \land\ \
\bigwedge_{1 \leq i < j \leq 5}\hspace*{-0.5em}
          \bigl(\ic(r_i + r_j) \wedge (r_i \cdot r_j = 0)\bigr),
\end{align}
which makes the analogous claim for regions $r_1, \dots, r_5$, is not
satisfiable over $\RC(\R^2)$, since any satisfying assignment would
permit a plane drawing of the graph $K_5$. On the other
hand,~\eqref{eq:2vs3} is easily seen to be satisfiable over
$\RC(\R^n)$ for all $n \geq 3$. Thus, the satisfiability problems for
\cBci{} over $\RC(\R)$, $\RC(\R^2)$ and $\RC(\R^3)$ are all different.
(We shall see in \SECT\ref{sec:3d}, however, that the satisfiability
problem for \cBci{} over $\RC(\R^n)$ is the same for all $n \geq 3$.)

Consider next sensitivity to restriction to (regular closed)
\emph{polyhedral} sets.  The \cBci-formula
\begin{align}\label{eq:wiggly}
\bigwedge_{1 \leq i \leq 3}\hspace*{-0.5em} \ic(r_i) \ \ \land\ \ \ic(r_1 + r_2 + r_3) \ \ \land\ \ \bigwedge_{2 \leq i \leq 3} \neg\ic(r_1 + r_i)
\end{align}
is satisfiable over $\RC(\R^2)$, as we see from the regular closed
sets in Fig.~\ref{fig:wiggly}, where $r_2$ and $r_3$ lie,
respectively, above and below the graph of the function
$\sin\frac{1}{x}$ on the interval $(0,1]$.  By contrast,
  formula~\eqref{eq:wiggly} is unsatisfiable over $\RCP(\R^n)$ for all
  $n \geq 1$~\cite{iscloes:kp-hwz10}.
\begin{figure}[ht]
\begin{center}
\setlength{\unitlength}{1mm}%
\begin{picture}(35,25)%
\put(0,0){\includegraphics[scale=0.5]{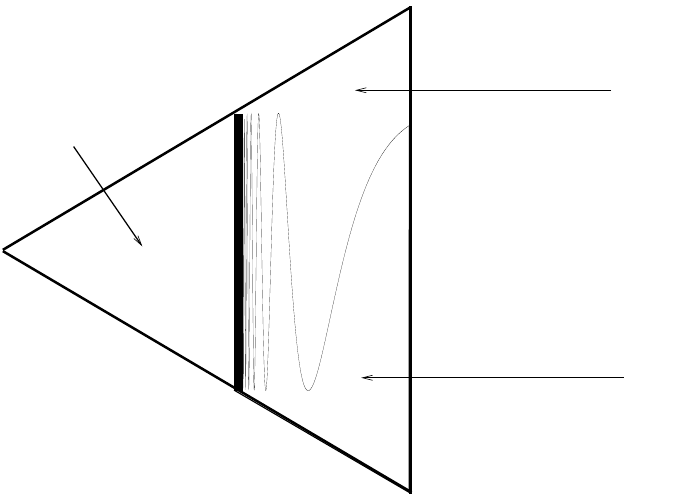}}%
\put(2,19){\small $r_1$}
\put(32,20){\small $r_2$}
\put(32,5.5){\small $r_3$}
\end{picture}%
\end{center}
\caption{Three regions in $\RC(\R^2)$ satisfying \eqref{eq:wiggly}.}
\label{fig:wiggly}
\end{figure}
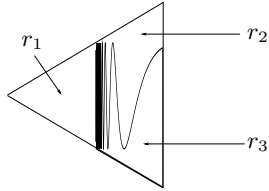
Actually, the result can be sharpened: \eqref{eq:wiggly} is
unsatisfiable over {\em any} Boolean sub-algebra of $\RC(\R^n)$ whose
regions all satisfy a form of the {\em curve selection lemma} from
real algebraic geometry~(see, e.g.~\cite{BCR98}). As we might say, in
dimensions 2 and above, \cBci{} is sensitive to the presence of
`non-tame' regions. And since---at least conceivably---non-tame
regions may be thought implausible models of the space occupied by any
physical objects---it is natural to consider satisfiability of
\cBci{}-formulas over $\RCP(\R^n)$ rather than over $\RC(\R^n)$.

The language \cBc{} is similarly sensitive to the dimension of the
Euclidean space over which it is interpreted, and also to the
restriction to polyhedral regions. For dimensionality, this
sensitivity can be demonstrated by examples similar to~\eqref{eq:1vs2}
and~\eqref{eq:2vs3}; see~\cite{KPWZ08aiml}. For the restriction to
polyhedral regions, this result follows from \SECT\ref{infcomp}, where we
show that there exists a \cBc{}-formula satisfiable in $\RC(\R^n)$ for
all $n \geq 2$, but only by tuples of regions having infinitely many
connected components!


\subsection{Plan of the paper and summary of results}
The remainder of this paper is organized as follows.
\SECT\ref{prelim}~defines the syntax and semantics of \cBc{} and
\cBci{}. To simplify proofs, we also employ the more expressive
languages \cBCc{} and \cBCci{}, obtained by adding the predicates $c$
and $\ic$, respectively, to \cBC{} ($=\BRCCE$). In \SECT\ref{infcomp}, we
prove that there exist \cBCc-, \cBCci- and \cBc-formulas satisfiable
over $\RC(\R^n)$, for all $n\geq 1$, but only by tuples of regions
some of which have infinitely many connected components, and hence
which cannot belong to $\RCP(\R^n)$. By further developing the ideas
encountered in this proof, we show in \SECT\ref{sec:undecidability} that
\cBCc, \cBCci{} and \cBc{} (but not \cBci{}) are r.e.-hard over
$\RCP(\R^n)$, for all $n\ge 2$. Using a different approach, we show in
\SECT\ref{sec:2d} that all four of our logics---\cBc, \cBci{}\!, \cBCc{}
and \cBCci{}---are r.e.-hard over both $\RCP(\R^2)$ and
$\RC(\R^2)$. Finally, we show in \SECT\ref{sec:3d} that \cBci{} is
\NP-complete over $\RC(\R^n)$, and \ExpTime-complete over
$\RCP(\R^n)$, for all $n\ge 3$. The decidability of satisfiability for
\cBCc, \cBCci{} and \cBc{} over $\RC(\R^n)$, for all $n\ge 3$, is left
open.


\section{Preliminaries}\label{prelim}

We begin by formally defining the syntax and semantics of the topological logics considered in this paper. This section also contains the basic technical definitions and results we need in what follows.

\subsection{Basic topological notions}

A \emph{topological space} is a pair $(T,\mathcal{O})$, where $T$ is a
set and $\mathcal{O}$ a collection of subsets of $T$ containing
$\emptyset$ and $T$, and closed under arbitrary unions and finite
intersections. The elements of $\mathcal{O}$ are referred to as {\em
  open sets}; their complements are {\em closed} sets. If
$\mathcal{O}$ is clear from context, we refer to the topological space
$(T,\mathcal{O})$ simply as $T$. If $X \subseteq T$, the {\em closure}
of $X$, denoted $\tc{X}$, is the smallest closed set including $X$,
and the {\em interior} of $X$, denoted $\ti{X}$, is the largest open
set included in $X$. These sets always exist.  The {\em boundary} of
$X$, denoted $\delta X$, is the set $\tc{X} \setminus \ti{X}$. The
Euclidean space $\R^n$ is assumed always to have the usual metric
topology.  We may treat any subset $X
\subseteq T$ as a topological space in its own right by defining the
{\em subspace topology} on $X$ to be the collection of sets
$\mathcal{O}_X = \{O \cap X \mid O \in \mathcal{O}\}$.

We call $X$ {\em regular closed} if it is the closure of an open
set---equivalently, if $X = \tc{(\ti{X})}$. We denote by $\RC(T)$ the
set of regular closed subsets of $T$. It is a standard result that
$\RC(T)$ forms a complete Boolean algebra, with operations $\sum A =
\tc{(\bigcup A)}$, $\prod A = \tc{\ti{(\bigcap A)}}$ and $-X =
\tc{\smash{(T \setminus {X})}}$ (see, e.g.~\cite{Koppelberg89}).  The
partial order induced by this Boolean algebra is simply
$(T,\subseteq)$; we often write $X \leq Y$ in preference to $X
\subseteq Y$ where $X$ and $Y$ are regular closed.  Note that, if $A =
\{X_1, X_2\}$, then $\smash{\sum A} = X_1 + X_2 = X_1 \cup X_2$.

A topological space $T$ is said to be \emph{connected} if it cannot be
decomposed into two disjoint, non-empty closed sets; likewise, $X$ is
\emph{connected} if it is a connected space under the subspace
topology. We call $X$ \emph{interior-connected} if
$\ti{X}$ is connected. A maximal connected subset of $X$ will be
called a {\em component} of $X$ (some authorities prefer the term {\em
  connected component}).  The following facts are easily verified:
every non-empty
connected subset of $X$ is included in a unique component of
$X$; every component of a closed set is closed.

The space $T$ is said to be {\em locally connected} if every
neighbourhood of any point of $T$ includes a connected neighbourhood
of that point (a \emph{neighbourhood} of a point $p$ is a set $X$ that
includes an open set $O$ containing $p$).  In a locally connected
space, every component of an open set is open; note however that
components of regular closed sets are closed but, in general, not
regular closed, even in locally connected spaces.  The space $T$ is
said to be {\em unicoherent} if, for any closed, connected subsets
$X_1$, $X_2$ such that $T = X_1 \cup X_2$, the set $X_1 \cap X_2$ is
connected. For all $n \geq 1$, the Euclidean space $\R^n$ is
(obviously) locally connected and (much less obviously)
unicoherent~\cite{utl:kuratowski}.  A simple example of a non-locally
connected space is the rational numbers $\mathbb Q$ under the usual
metric topology. Simple examples of non-unicoherent spaces are the
Jordan curve and the torus.

The most important properties of local connectedness and unicoherence,
from our point of view, are given by the following lemmas.
\begin{lemma}\label{lma:lc}
Let $X$ be a regular closed subset of a topological space $T$ and $S$
a component of $-X$. If $-X$ has finitely many components, then
$\delta S \subseteq X$.  Alternatively, if $T$ is locally connected,
then $\delta S \subseteq X$.
\end{lemma}
\begin{proof}
For the first statement, let $Z$ be the union of all components of
$-X$ other than $S$.  By definition, $T = \ti{X} \cup S \cup
Z$. Further, both $S \cap \ti{X}$ and $S \cap Z$ are empty, whence $T
\setminus S = \ti{X} \cup Z$. Since $X$ is regular closed, and $Z$ is
closed (as the union of finitely many closed sets), $\tc{(T \setminus
  S)} = X \cup Z$. Finally, since $S$ is closed, and $S \cap Z =
\emptyset$, $\delta S = S \cap \tc{(T \setminus S)} \subseteq X$. For
the second statement, suppose, to the contrary, that $\delta S$
contains a point $p$ lying in $\ti{(-X)}$. Since $S$ is closed, $p \in
S$.  By local connectedness, let $Y$ be a connected open set such that
$p \in Y \subseteq \ti{(-X)}$. Since $p \in S$ and $Y$ is a connected
subset of $-X$, we have $p \in Y \subseteq S$. But this contradicts
the assumption that $p \in \delta S$.
\end{proof}

\begin{lemma}\label{lma:ourNewman}
Let $T$ be a unicoherent space and $X\in\RC(T)$ be connected.  Then
every component of $-X$ has a connected boundary.
\end{lemma}
\begin{proof}
Let $S$ be a
connected component
of $-X$, and let $Z$ be the union of all components of $-X$ other than
$S$.  Thus, $T \setminus S = Z \cup \ti{X}$.  We write $S^* = \tc{(T
  \setminus S)}$.  Since $X$ is regular closed, $S^* = \tc{Z} \cup
X$. By connectedness of $T$, $X$ intersects every component of $-X$.
It follows that $Z \cup X$, and hence $\tc{Z} \cup X = S^*$ are
connected.  By definition, $S$ is connected, whence, by unicoherence
of $T$, $\delta S = S \cap S^*$ is connected.
\end{proof}

\subsection{Frames}
A {\em frame} is a pair $(T, \bS)$, where $T$ is a topological space,
and $\bS$ is a Boolean sub-algebra of $\RC(T)$.  Where $T$ is clear
from context, we refer to $(T, \bS)$, simply, as $\bS$. Furthermore,
where $\bS$ is clear from context, we refer to elements of $\bS$ as
{\em regions}. We denote the class of frames of the form $(T, \RC(T))$
by $\RegC$. Note that not all frames are of this form: in particular,
when working in $n$-dimensional Euclidean spaces, we shall be
principally interested in the following proper sub-algebra of
$\RC(\R^n)$.  Any $(n-1)$-dimensional hyperplane bounds two elements
of $\RC(\R^n)$ called \emph{half-spaces}. We denote by $\RCP(\R^n)$
the Boolean sub-algebra of $\RC(\R^n)$ generated by the half-spaces,
and call the elements of $\RCP(\R^n)$ (regular closed)
\emph{polyhedra}. If $n = 2$, we speak of (regular closed)
\emph{polygons}. Polyhedra may be regarded as `well-behaved' or, in
topologists' parlance, `\emph{tame}.' We call $(T, \bS)$ {\em
  unicoherent} if $T$ is unicoherent, and {\em finitely
  decomposible} if, for all $s \in \bS$, there exist connected
elements $s_1, \dots, s_k$ of $\bS$, such that $s= s_1+ \cdots +
s_k$. Evidently, $(\R^n,\RCP(\R^n))$ is finitely decomposible, since
any product of half-planes is connected. Equally obviously:
\begin{lemma}\label{lma:fd}
Suppose the frame $(T, \bS)$ is finitely decomposible, and $s \in \bS$.
Then every component of $s$ is in $\bS$, and $s$ is equal to the sum of
those components.
\end{lemma}

The following basic concepts will be used repeatedly in the
sequel. Let $(T,\bS)$ be a frame. A tuple of elements $(s_0, \dots,
s_{k-1})$, where $k \geq 1$, will be called a \emph{partition},
provided
\begin{align*}
s_0 + \cdots + s_{k-1} = 1\quad\text{ and }\quad s_i \cdot s_j
= 0 \ \ \text{ for all } 0 \leq i < j < k.
\end{align*}
We do not insist that the $s_i$ are non-empty. We call a partition
$(s_0, \dots, s_{k-1})$ {\em sub-cyclic} if the $s_i$ are non-empty
and
\begin{align*}
s_i \cap s_j = \emptyset, \ \ \text{ for all } 0 \leq i,j < k \text{ such
that } 
1<j-i<k-1.
\end{align*}
The term `sub-cyclic' refers to an imagined graph with nodes $\{s_0,
\dots, s_{k-1}\}$ and edges $\{(s_i,s_j) \mid i \neq j \text{ and }
s_i \cap s_j \neq \emptyset\}$: this graph is required to be a (not
necessarily proper) subgraph of the cyclic graph on $\{s_0, \dots,
s_{k-1}\}$.

Suppose $s$ is a non-empty element of a frame $(T,\bS)$, and $\vec{s}
= (s_0, \dots, s_{k-1})$ a partition in that frame. We say that
$\vec{s}$ is a {\em colouring} of the components of $s$ if every
component of $s$ is included in exactly one of the regions of
$\vec{s}$.  Colourings will be used repeatedly in the sequel,
particularly in situations where we may regard the components of $s$
as positions in a finite sequence; by regarding the set of elements of
$\vec{s}$ as an alphabet, colourings define {\em words} over that
alphabet in the obvious way.

\subsection{Topological logics}
In this paper, the focus of attention is not on frames themselves, but
rather, on frames {\em as they are described in some language}.  The
languages considered here all employ a countably infinite collection
of variables $r_1, r_2, \dots$.
The language \cBC{} is defined by the following syntax:
\begin{align*}
\tau \quad & ::=
\quad r \ \ \mid
\ \ \tau_1 + \tau_2 \ \ \mid
\ \ \tau_1 \cdot \tau_2 \ \ \mid
\ \ - \tau_1 \ \ \mid
\ \ 0 \ \ \mid
\ \ 1,\\
\phi \quad & ::= \quad \tau_1 =  \tau_2 \ \
\mid \ \ C(\tau_1,\tau_2) \ \
\mid \ \ \phi_1 \lor \phi_2 \ \
\mid \ \ \phi_1 \land \phi_2 \ \
\mid \ \ \neg \phi_1.
\end{align*}
The language \cB{} is defined analogously, but without the predicate
$C$.  Thus, \cB{} is the quantifier-free language of the variety of Boolean
algebras.

An \emph{interpretation over} a frame $(T, \bS)$ is a function
$\cdot^\mathfrak{I}$ mapping variables $r$ to elements
$r^\mathfrak{I}$ of $\bS$. We extend $\cdot^\mathfrak{I}$ to terms
$\tau$ by setting $(\tau_1 + \tau_2)^\mathfrak{I} =
\tau_1^\mathfrak{I} + \tau_2^\mathfrak{I}$, $(\tau_1 \cdot
\tau_2)^\mathfrak{I} = \tau_1^\mathfrak{I} \cdot \tau_2^\mathfrak{I}$,
$(- \tau_1)^\mathfrak{I} = - (\tau_1^\mathfrak{I})$, $0^\mathfrak{I} =
\emptyset$ and $1^\mathfrak{I} = T$.  We write $\mathfrak{I} \models
\tau_1 = \tau_2$ if and only if $\tau_1^\mathfrak{I} =
\tau_2^\mathfrak{I}$, and $\mathfrak{I} \models C(\tau_1,\tau_2)$ if
and only if $\tau_1^\mathfrak{I} \cap \tau_2^\mathfrak{I} \neq
\emptyset$, extending this relation to non-atomic formulas in the
standard way.  We read $C(\tau_1, \tau_2)$ as `$\tau_1$
\emph{contacts} $\tau_2$.' If $\phi$ is a formula whose variables,
taken in some order, are $\vec{r} = (r_1, \dots, r_n)$, and
$\mathfrak{I}\models \phi$, then the tuple $\vec{a} =
(r_1^\mathfrak{I}, \dots, r_n^\mathfrak{I})$ is said to {\em satisfy}
$\phi(\vec{r})$; in such a case, we will often say `$\vec{a}$
satisfies $\phi(\vec{r})$.'

We remark that the property that a $k$-tuple $(r_0,\dots,r_{k-1})$
forms a partition is evidently expressible using the \cB-formula
\begin{equation*}
\partition(r_0,\dots,r_{k-1})
\ \ =\ \ \left(\sum_{i=0}^{k-1} r_i = 1\right) \ \ \land\hspace*{-0.3em} \bigwedge_{0 \leq i < j <
  k}\hspace*{-0.5em} (r_i \cdot r_j = 0).
\end{equation*}
The property that a $k$-tuple forms a
sub-cyclic partition is expressible using the \cBC-formula
\begin{multline*}
\cpartition(r_0,\dots,r_{k-1})
\ \ = \ \ \\
\partition(r_0,\dots,r_{k-1}) \wedge
   \bigwedge_{0 \leq i < k} (r_i \neq 0)
\wedge
   \bigwedge_{1<j-i<k-1}\neg C(r_i, r_j).
\end{multline*}
And, assuming that $\partition(r_0,\dots,r_{k-1})$ is satisfied, the
\cBC-formula
$$
\colourComp(r; r_0,\dots,r_{k-1})
\ \ =\ \
   \bigwedge_{0 \leq i < j < k} \neg C((r \cdot r_i), \ (r \cdot r_j))
$$
ensures that the partition $r_0, \dots, r_{k-1}$ colours the
components of $r$. Conversely, over finitely decomposible frames, any
colouring of $r$ by a partition $r_0, \dots, r_{k-1}$ must satisfy
$\colourComp(r; r_0,\dots,r_{k-1})$.

Turning to connectedness predicates, we define the languages \cBc{}
and \cBCc{} to be extensions of \cB{} and \cBC{}, respectively, with
the unary predicate $c$. We set $\mathfrak{I} \models c(\tau)$ if and
only if $\smash{\tau^{\mathfrak I}}$ is connected in the topological
space under consideration. Similarly, we define \cBci{} and \cBCci{}
to be extensions of \cB{} and \cBC{} with the predicate $\ic$, setting
$\mathfrak I \models \ic(\tau)$ if and only if
$\ti{\smash{(\tau^{\mathfrak I})}}$ is connected. If $\cK$ is a class
of frames, and $\cL$ is one of \cBc, \cBCc, \cBci{} or \cBCci{}, then
$\Sat(\mathcal{L},\cK)$ is the set of $\mathcal{L}$-formulas
satisfiable over $\cK$.

Setting $\cK = \RegC$, the complexity of this problem in known for all
of the languages $\cL$ considered
above~\cite{iscloes:kp-hwz10, KPZ10kr}.
If $\phi$ is a formula of any
of the languages \cBc, \cBCc{} or \cBCci{}, and $\phi$ is
satisfiable over $\RegC$, then $\phi$ is satisfiable over some frame
$\RC(T)$, where $|T|$ is bounded by a singly-exponential function of
$|\phi|$; and the problems $\Sat(\mathcal{L},\RegC)$, for $\mathcal{L}
\in \{\cBc, \cBCc, \cBCci\}$, are all \ExpTime-complete. On the
other hand, if $\psi$ is a \cBci-formula satisfiable over $\RegC$,
then $\psi$ is satisfiable over some frame $\RC(T)$, where $|T|$ is
bounded by a polynomial function of $|\psi|$; and the problem
$\Sat(\cBci,\RegC)$ is \NP-complete.  Thus, we observe a difference
between \cBc, \cBCc{} and \cBCci{} on the one hand, and \cBci{} on
the other.

However, satisfiability over $\RegC$ is of little interest from the
point of view of Artificial Intelligence, where almost all conceivable
applications concern the frames over Euclidean space of dimensions 2
or 3.  Accordingly, we shall be concerned with $\Sat(\cL, \cK)$, where
$\cL$ is any of \cBc{}, \cBci{}, \cBCc{} or \cBCci, and $\cK$ is $\{
\RC(\R^n) \}$ or $\{\RCP(\R^n)\}$ for $n \geq 2$. For ease of reading,
we write $\Sat(\cL, \RC(\R^n))$ and $\Sat(\cL, \RCP(\R^n))$ rather
than $\Sat(\cL, \{\RC(\R^n)\})$ and $\Sat(\cL,
\{\RCP(\R^n)\})$.

\subsection{Graphs}
Unless explicitly indicated to the contrary, all graphs in this paper
are taken to be finite, and to have no multiple edges and no loops:
i.e., if $G = (V,E)$ is a graph, $(v,v') \in E$ implies $v \neq
v'$. We also assume that the edges have no direction, i.e., $(v,v')\in
E$ if and only if $(v',v)\in E$.  A \emph{path} in $G$ is a sequence
of distinct vertices $v_0, \dots, v_{n-1}$ such that $(v_i,v_{i+1})$
is an edge, for all $0 \leq i < n-1$; further, a \emph{cycle} in $G$
is a path $v_0, \dots, v_{n-1}$ such that, in addition,
$(v_{n-1},v_0)$ is an edge. Informally, in this case, we speak of the
sequence $v_0, \dots, v_{n-1}, v_0$ as a cycle.  A graph is connected
if any two nodes are joined by some path; a graph which contains no
cycles is {\em acyclic}; and a connected, acyclic graph is a {\em
  tree}. If $G$ is a tree, then any pair of nodes in $G$ is joined by
a unique path. Further, if $v_0, \dots, v_{n-1}$ is a sequence of
nodes in $G$ such that $(v_i, v_{i+1})$ is an edge for all $0 \leq i <
n-1$, and $v_i \neq v_{i+2}$ for all $0 \leq i < n-2$, then this
sequence contains no duplicates, and thus is a path.

Let $\bS$ be a finitely decomposible frame over some topological
space, and $\vec{s}$ a connected partition in $\bS$.  We can associate
a graph with $\vec{s}$, denoted $H(\vec{s})$, as follows: the vertices
of $H(\vec{s})$ are the components of the elements of $\vec{s}$; the
edges of $H(\vec{s})$ are the pairs $(X,Y)$ such that $X \neq Y$ and
$X \cap Y \neq \emptyset$. We refer to $H(\vec{s})$ as the
\emph{component graph} of $\vec{s}$. Note that the number of vertices
of $H(\vec{s})$ is in general larger than the number of elements in
$\vec{s}$; however, since $\bS$ is finitely decomposible, this number
is still finite.

We prove a simple but powerful lemma connecting some of the notions
encountered above.
\begin{lemma}\label{lma:noCycles}
Let $T$ be a unicoherent topological space, $\bS$ a
finitely decomposible frame on $T$, and $\vec{s}$ a sub-cyclic
partition in $\bS$. Then the component graph, $H(\vec{s})$, is a tree.
\end{lemma}
\begin{proof}
Write $\vec{s} = (s_0, \dots, s_{n-1})$.
Since $\bS$ is finitely decomposible, and $T$ is connected, $H(\vec{s})$
is obviously finite and connected. We need only show that it contains
no cycles.  If $n =1$, then $|H(\vec{s})|=1$, and this is trivial.
We assume, for ease of formulation, that $n \geq 4$, since a similar
(and in fact simpler) argument applies if $n =2$ or $n =3$.

Suppose $(X_0, X_1)$ is an edge of $H(\vec{s})$. We may assume,
without loss of generality, that $X_0$ is a component of $s_0$, and
$X_1$ a component of $s_1$. The sub-cyclicity condition ensures that
$s_0 \cap s_i = \emptyset$ for $2 \leq i <n-1$, and $s_1 \cap s_i =
\emptyset$ for $3 \leq i <n$.  Now let $S$ be the component of $-X_1$
containing $X_0$: we claim that $\delta S \subseteq s_0$.  By the
first statement of Lemma~\ref{lma:lc}, $\delta S \subseteq X_1
\subseteq s_1$, whence $\delta S$ contains no point of $s_3 + \cdots +
s_{n-1}$.  On the other hand, $\delta S$ is obviously included in
$-s_1 = s_0 + s_2 + s_3 + \cdots s_{n-1}$, and hence in $s_0 + s_2$.
Since $s_0 \cap s_2 = \emptyset$, and, by Lemma~\ref{lma:ourNewman},
$\delta S$ is connected, we have either $\delta S \subseteq s_0$ or
$\delta S \subseteq s_2$. Now, since $(X_0, X_1)$ is an edge of
$H(\vec{s})$, and any point of $X_0 \cap X_1$ must lie in both $S$ and
$-S$, we have $\delta S \cap X_0 \neq \emptyset$, and, therefore,
$\delta S \subseteq s_0$, as claimed.

Now suppose $(X_1, X_2)$ is also an edge of $H(\vec{s})$, with $X_0$
and $X_2$ distinct.  We claim that $X_0$ and $X_2$ lie in different
components of $-X_1$ (i.e., $X_2 \nsubseteq S$). For suppose
otherwise.  Again, since any point of $X_1 \cap X_2$ lies in
both $S$ and $-S$, $\delta S \cap X_2 \neq \emptyset$. Furthermore,
since $s_1 \cap s_i = \emptyset$ for $3 \leq i < n$, $X_2$ must be a
component of either $s_0$ or $s_2$.  But if $X_2 \subseteq s_0$, then
the connected set $\delta S \subseteq s_0$ has points in common with
the components $X_0$, $X_2$ of $s_0$, contradicting the assumption
that $X_0$ and $X_2$ are distinct. On the other hand, if $X_2
\subseteq s_2$ then $\delta S \subseteq s_0$ contains a point of
$s_2$, which is again impossible.

Finally, suppose that $X_0, X_1, X_2, \dots, X_m$ is a cycle in
$H(\vec{s})$, where $m \geq 3$ and $X_m = X_0$. Then the connected set $X_2 + \cdots +
X_{m-1} + X_0$ lies entirely in $-X_1$, contradicting the fact that
$X_0$ and $X_2$ lie in different components of $-X_1$.
\end{proof}

\subsection{Post correspondence problem}
In the sequel, we make use of the well-known \emph{Post correspondence
  problem} (PCP). Fix finite alphabets $T$ and $U$, where $|T| \geq 7$
and $|U| \geq 2$.  A \emph{morphism} from $T$ to $U$ is a function
$\pcpw\colon T \rightarrow U^*$ mapping each element of $T$ to a word
over $U$.  We extend $\pcpw$ to a mapping $\pcpw\colon T^* \rightarrow
U^*$ by defining, for any word $\tau = t_1 \cdots t_k \in T^*$,
$\pcpw(\tau) = \pcpw(t_1) \cdots \pcpw(t_k)$.  An \emph{instance} of
the PCP is a pair of morphisms $\pcpW = (\pcpw^1, \pcpw^2)$ from $T$
to $U$. The instance $\pcpW$ is \emph{positive} if there exists a
non-empty word $\tau \in T^*$ such that $\pcpw^1(\tau) =
\pcpw^2(\tau)$. Intuitively, we are invited to think of each element
of $T$ as a `tile' inscribed with an `upper' word over $U$, given by
$\pcpw^1(t)$, and a `lower' word over $U$, given by $\pcpw^2(t)$; we
are asked to determine, for the given collection of tiles, whether
there exists a non-empty, finite sequence of these tiles (repeats
allowed) such that the concatenation of their upper words equals the
concatenation of their lower words. The set of positive PCP instances
is known to be r.e.-complete~\cite{KK:Post46}, and remains so even
under the restriction that $\pcpw^k(t)$ is non-empty for every $t \in
T$. In fact, nothing hinges on the exact choice of $T$ and $U$,
subject to the restrictions mentioned above. In particular, we may
assume $T$ and $U$ are disjoint.

\section{Forcing infinitely many components in locally connected unicoherent spaces.}\label{infcomp}
In this section, we construct \cBCc-, \cBCci- and
\cBc-formulas $\phi$ with the following properties: ({\em i}) $\phi$
is satisfiable over $\RC(\R^n)$ for all $n \geq 2$; ({\em ii}) if $T$
is a locally connected, unicoherent space and $\vec{a}$ is a tuple
from $\RC(T)$ satisfying $\phi$, then $\vec{a}$ includes members with
\emph{infinitely many} connected components.  Since $\RCP(\R^n)$ is
finitely decomposible, these properties entail that
$\Sat(\cL,\RC(\R^n)) \neq \Sat(\cL,\RCP(\R^n))$ for $\cL$ any of
\cBCc, \cBCci{} or \cBc, and all $n \geq 2$. Furthermore, the
techniques developed in this section will be used in \SECT\ref{sec:undecidability} to prove
that satisfiability of \cBCc-, \cBCci- and \cBc-formulas over
$\RCP(\R^n)$, for $n\ge 2$, is undecidable.

We now construct a \cBCc-formula, $\phi_\infty$ with properties
({\em i}) and ({\em ii}).  As an aid to intuition, consider any
locally connected unicoherent space $T$. We equivocate between
variables and the regions they are assigned in some putative
interpretation over $\RC(T)$. In this section we write $\md{i}$ to
denote the value of $i$ modulo $4$.  The first conjunct of
$\phi_\infty$ states that $r_0,r_1,r_2,r_3$ form a sub-cyclic partition:
\begin{align}\label{eq:InfPart1}
\cpartition(r_0,r_1,r_2,r_3).
\end{align}
We also require non-empty sub-regions $r'_i$ of $r_i$ and a non-empty
region $t$:
\begin{align}\label{eq:basic-regions}
& \bigwedge_{i = 0}^3
    \bigl((r'_i \ne 0) \land (r'_i \leq r_i)\bigr),\\ & (t\ne 0).
\end{align}
The configuration we have in mind is depicted in
Fig.~\ref{fig:InfCmpSat}, where components of the $r_i$ are arranged
like the layers of an onion.
\begin{figure}[ht]
\begin{center}
\begin{tikzpicture}[yscale=0.8]\small
	\coordinate (LB1) at (0,0); 			
	\coordinate (RT1) at (2,1); 			
	\coordinate (LB2) at (0,.3); 			
	\coordinate (RT2) at (.5,.8);			
	\coordinate (LBP) at (.2pt,.2pt); 		
	\coordinate (RTP) at (-.2pt,-.2pt); 	
	\coordinate (Dx) at (0.5,0); 			
	\coordinate (Dy) at (0,.2); 			
	\coordinate (Dy2) at (0,.15); 			
	\coordinate (HH) at (.25,.5);			
	
	\foreach \i/\c in {0/gray!8,1/gray!50,2/gray!20,3/gray!80,0/gray!8,1/gray!50,2/gray!20,3/gray!80}
	{
		\begin{scope}[even odd rule]
			\clip ($(LBP)-(.2pt,.2pt)$) rectangle ($(RTP)+(.2pt,.2pt)$)
				($(LB1)-(.2pt,.2pt)$) rectangle ($(RT1)+(.2pt,.2pt)$) ;
			\filldraw[fill=\c] (LB1) rectangle (RT1); \node at ($(RT1)-(HH)$) {$r_\i$};
			\filldraw[fill=\c] (LB2) rectangle (RT2) node[midway] {$r_\i'$};
			\coordinate (LBP) at (LB1);
			\coordinate (RTP) at (RT1);
			\coordinate (LB1) at ($(LB1)-(Dx)-(Dy)$);
			\coordinate (RT1) at ($(RT1)+(Dx)+(Dy)$);
			\coordinate (LB2) at ($(LB2)-(Dx)-(Dy2)$);
			\coordinate (RT2) at ($(RT2)-(Dx)+(Dy2)$);
			\coordinate (HH) at ($(HH)+(Dy)$);
		\end{scope}
	}
	\filldraw[fill opacity=0.5, fill=gray] (6.2,.8)--(1.5,.8)--(6.2,1.3);
	\node at (6.5,1.05) {$t$};
	\node at (6,.5) {$\cdots$};
	\node at (-3.9,.5) {$\cdots$};
%
\end{tikzpicture}		
\end{center}
\caption{Regions satisfying $\phi_\infty$.}\label{fig:InfCmpSat}	
\end{figure}
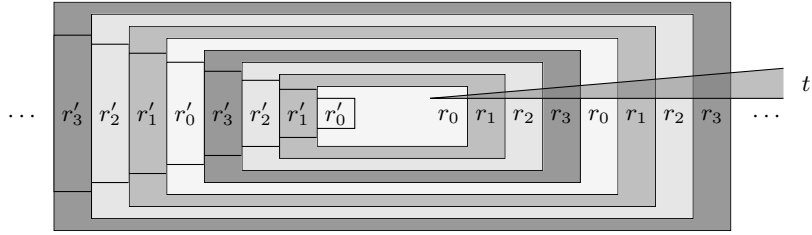
The `innermost' component of $r_0$ is
surrounded by a component of $r_1$, which in turn is surrounded by a
component of $r_2$, and so on. The region $t$ passes through every
layer, but avoids the $r'_i$. To enforce a configuration of this sort,
we need the following formulas:
\begin{align}
\label{eq:InfContact}	& \bigwedge_{i = 0}^3 c(r'_i+r_{\md{i+1}}+t),\\
\label{eq:notC2} & \bigwedge_{i = 0}^3 \neg C(r'_i, t),\\
\label{eq:notC}	& \bigwedge_{i = 0}^3 \neg C(r'_i, \ r_{\md{i+1}}\cdot (-r'_{\md{i+1}})).
\end{align}
Observe that~\eqref{eq:InfContact}--\eqref{eq:notC} ensure
each component of $r'_i$ is in contact with $r'_{\md{i+1}}$.  Denote by
$\phi_\infty$ the conjunction of~\eqref{eq:InfPart1}--\eqref{eq:notC}.

\begin{theorem}\label{theo:inftyCc}
The \cBCc-formula $\phi_\infty$ is satisfiable over
$\RC(\R^n)$, $n \geq 2$. On the other hand, if $T$ is a locally
connected, unicoherent space, then any tuple from $\RC(T)$ satisfying
$\phi_\infty$ features sets that have infinitely many components.
\end{theorem}
\begin{proof}
Fig.~\ref{fig:InfCmpSat} shows how $\phi_\infty$ can be satisfied over
$\RC(\R^2)$. By cylindrification, it is also satisfiable over any
$\RC(\R^n)$, for $n> 2$. This establishes the first statement of the
lemma.  For the second statement, we suppose that $\phi_\infty$ is
satisfied in a frame on a locally connected, unicoherent space $T$; we
show that some members of the satisfying tuple have infinitely many
components.  To avoid clumsy circumlocutions, we equivocate between
variables and the regions to which they are assigned in the satisfying
interpretation: thus we speak about the `regions' $r_0, \ldots, r_3$,
$r'_0, \ldots, r'_3$ and $t$. No confusion should result.

We proceed by constructing a sequence of disjoint components $X_i$ of
$r_{\md i}$ and open sets $V_i$ connecting $X_i$ to $X_{i+1}$; see
Fig.~\ref{fig:InfCmpConstr}.
\begin{figure}[ht]
\hspace*{-7mm}%
\begin{tikzpicture}[yscale=0.8]\small		
	\coordinate (dx1) at (-1,0);
	\coordinate (dx2) at (-1.3,0);
	\coordinate (dy) at (0,.15);
	\coordinate (dx) at (-.1,0);			
	\coordinate (ML) at (0,.5);
		\newcommand{\iterate}[1]{
			\coordinate (BL) at (4,0);
			\coordinate (TR) at (4,1);			
			\coordinate (oBL) at (4,0);
			\coordinate (oTR) at (4,0);
			
			\foreach \n/\x/\l/\c in {0/0/0/gray!8,1/1/0/gray!50,2/2/0/gray!20,3/3/0/gray!80,4/0/1/gray!8}
			{
				#1
				
				\coordinate (oBL) at (BL);
				\coordinate (oTR) at (TR);
				\coordinate (BL) at ($(BL)-2*(dy)-2*(dx)$);
				\coordinate (TR) at ($(TR)+2*(dy)+(dx1)+(dx2)$);
			}
		}
		\iterate{
			\begin{scope}[even odd rule]	
	           \coordinate (RECT) at ($(BL)-2*(dy)-2*(dx)$);
	           \coordinate (RECTSZ) at ($(TR)+2*(dy)+(dx1)+(dx2)$);
				\clip (BL) rectangle (TR) (RECT) rectangle (RECTSZ);
				\ifnum \l=0
					\filldraw[thick,fill=white!20] (RECT) rectangle (RECTSZ);
				\fi
				
				\filldraw[fill=\c] ($(BL)-(dy)-(dx)$) rectangle ($(TR)+(dy)+(dx1)$);
				\node at ($(TR)+.5*(dx1)-(ML)-\n*(dy)-\n*(dy)$){$X_\n$};
			\end{scope}
		}
		\iterate{					
			\ifnum \l=0
		        \coordinate (EC) at ($(TR)+(dx1)+.5*(dx2)-(ML)-\n*(dy)-\n*(dy)-(0,.35)$);
				\draw[thick] (EC) ellipse (.8cm and .25cm);
				\node at (EC) {$V_\x$};
				\ifnum \n>0
					\draw[-latex] ($(TR)+(dx1)-(0,.6)$) --++($.4*(dx2)$) node[midway,above] {$R_\n$};
					\draw[-latex] ($(TR)-(0,.6)$) --++($-.4*(dx2)$) node[midway,above] {$S_\n$};
				\fi
			\fi
		}
\end{tikzpicture}
\caption{Sequences $X_0,X_1,\dots$ of components $X_i$ of $r_{\md{i}}$ and $V_0,V_1,\dots$ of open sets $V_i$ connecting $X_i$ to $X_{i+1}$ with
the `holes' $S_{i+1}$ and $R_{i+1}$ of $X_{i+1}$ containing $X_{i}$ and $X_{i+2}$, respectively.
}
\label{fig:InfCmpConstr}
\end{figure}
By the first conjunct
of~\eqref{eq:basic-regions}, let $X_0$ be a component of $r_0$
containing points in $r'_0$. Suppose $X_i$ has been constructed.
By~\eqref{eq:InfContact}--\eqref{eq:notC}, $X_i$ is in contact with
$r'_{\md{i+1}}$. Using~\eqref{eq:InfPart1} and the fact that $T$ is
locally connected, one can find a component $X_{i+1}$ of
$r_{\md{i+1}}$ which has points in $r'_{\md{i+1}}$, and a connected open set
$V_i$ such that $V_i \cap X_i$ and $V_i \cap X_{i+1}$ are non-empty,
but $V_i \cap r_{\md{i+2}}$ is empty.

To see that the $X_i$ are distinct, let $S_{i+1}$ and $R_{i+1}$ be the
components of $-X_{i+1}$ containing $X_i$ and $X_{i+2}$,
respectively. It suffices to show that we have $S_{i+1}
\subseteq\ti{S}_{i+2}$.  Note that the connected set $V_i$ must
intersect $\delta S_{i+1}$.  By the second statement of
Lemma~\ref{lma:lc}, $\delta S_{i+1} \subseteq X_{i+1} \subseteq
r_{\md{i+1}}$.  Also, $\delta S_{i+1} \subseteq -X_{i+1}$; hence,
by~\eqref{eq:InfPart1}, $\delta S_{i+1} \subseteq r_{\md{i}} \cup
r_{\md{i+2}}$.  By Lemma~\ref{lma:ourNewman}, $\delta S_{i+1}$ is
connected, and therefore, by~\eqref{eq:InfPart1}, $\delta S_{i+1}$ is
entirely contained either in $r_{\md{i}}$ or in $r_{\md{i+2}}$. Since
$V_i \cap \delta S_{i+1} \neq \emptyset$ and $V_i \cap r_{\md{i+2}} =
\emptyset$, we have $\delta S_{i+1} \nsubseteq r_{\md{i+2}}$, so
$\delta S_{i+1} \subseteq r_{\md{i}}$. Similarly, $\delta
R_{i+1}\subseteq r_{\md{i+2}}$.  By~\eqref{eq:InfPart1}, then, $\delta
S_{i+1} \cap \delta R_{i+1} = \emptyset$, and since $S_{i+1}$ and
$R_{i+1}$ are components of the same set, and have non-empty
boundaries, they are disjoint. Hence, we obtain $S_{i+1}\subseteq
\ti{(-R_{i+1})}$, and since $X_{i+2}\subseteq R_{i+1}$, also
$S_{i+1}\subseteq \ti{(-X_{i+2})}$. So, using local connectedness
again, $S_{i+1}$ lies in the interior of a component of $-X_{i+2}$,
and since $\delta S_{i+1}\subseteq X_{i+1}\subseteq S_{i+2}$, that
component must be $S_{i+2}$.
\end{proof}

Now we show how the \cBCc-formula $\phi_\infty$ can be transformed
to \cBCci- and \cBc-formulas with similar properties.  Note first
that all occurrences of $c$ in $\phi_\infty$ have positive polarity.
Let $\ti{\phi}_\infty$ be the result of replacing them with the
predicate $\ic$. In Fig.~\ref{fig:InfCmpSat}, the connected regions
mentioned in~\eqref{eq:InfContact} are in fact interior-connected;
hence $\ti{\phi}_\infty$ is satisfiable over $\RC(\R^n)$. Since
interior-connectedness implies connectedness, $\ti{\phi}_\infty$
entails $\phi_\infty$, and we obtain:

\begin{corollary}\label{cor:inftyCci}
The \cBCci-formula $\ti{\phi}_\infty$ is satisfiable over $\RC(\R^n)$,
\mbox{$n \geq 2$}. On the other hand, if $T$ is a locally connected,
unicoherent space, then any tuple from $\RC(T)$ satisfying
$\ti{\phi}_\infty$ features sets that have infinitely many components.
\end{corollary}

We next consider the language \cBc.  Observe that all occurrences of
$C$ in $\phi_\infty$ are negative. We eliminate these using the
predicate $c$: we use the fact that, if the sum of two connected
regions is not connected, then they are not in contact. If $\tau_1$
and $\tau_2$ are any terms, we employ the abbreviation
\begin{equation*}
	\noncontact(\tau_1, \tau_2) \ \ = \ \ c(\tau_1)\land c(\tau_2)
           \land \neg c(\tau_1+\tau_2).
\end{equation*}
Observe that $\noncontact(\tau_1, \tau_2)$ is always a \cBc-formula.
Furthermore, $\noncontact(\tau_1, \tau_2)$ implies $\neg C(\tau'_1,
\tau'_2)$ for any $\tau'_1 \leq \tau_1$ and $\tau'_2 \leq \tau_2$.
Now we replace~\eqref{eq:notC2} by
\begin{equation}\tag{\ref{eq:notC2}$^c$}
\noncontact(r_0'+r_1'+r_2'+r_3', \  t).
\end{equation}
The resulting formula thus implies the original; on the other hand, it
is satisfied by the configuration of Fig.~\ref{fig:InfCmpSat}.
Next, we replace each conjunct $\neg C(r'_i,r_{\md{i+1}}\cdot (-
r'_{\md{i+1}}))$ in~\eqref{eq:notC} by
\begin{equation}\tag{\ref{eq:notC}$^c$}
\noncontact(r_i' + s, \ \ r_{\md{i+1}}\cdot (- r'_{\md{i+1}}) + t),
\end{equation}
where $s$ is a fresh variable.  Again, the resulting formula implies
the original, and, furthermore, is evidently satisfied by the
configuration of Fig.~\ref{fig:InfCmpElAiAi3}, where $s$ lies inside
$\sum_{j = 0}^3 r_j'$, symmetrically to $t$ lying inside $\sum_{j=0}^3
(r_j\cdot (-r'_j))$.
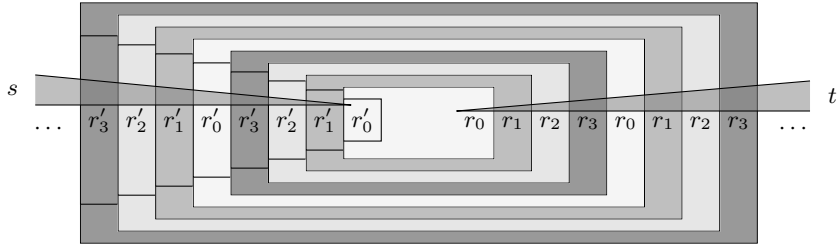
\begin{figure}[ht]
\begin{center}
\begin{tikzpicture}[yscale=0.8]\small
	\coordinate (LB1) at (0,0); 			
	\coordinate (RT1) at (2,1.2); 			
	\coordinate (LB2) at (0,.3); 			
	\coordinate (RT2) at (.5,1);			
	\coordinate (LBP) at (.2pt,.2pt); 		
	\coordinate (RTP) at (-.2pt,-.2pt); 	
	\coordinate (Dx) at (0.5,0); 			
	\coordinate (Dy) at (0,.2); 			
	\coordinate (Dy2) at (0,.15); 			
	\coordinate (HH) at (.25,.6);			
	
	\foreach \i/\c in {0/gray!8,1/gray!50,2/gray!20,3/gray!80,0/gray!8,1/gray!50,2/gray!20,3/gray!80}
	{
		\begin{scope}[even odd rule]
			\clip ($(LBP)-(.2pt,.2pt)$) rectangle ($(RTP)+(.2pt,.2pt)$)
				($(LB1)-(.2pt,.2pt)$) rectangle ($(RT1)+(.2pt,.2pt)$) ;
			\filldraw[fill=\c] (LB1) rectangle (RT1); \node at ($(RT1)-(HH)$) {$r_\i$};
			\filldraw[fill=\c] (LB2) rectangle (RT2) node[midway] {$r_\i'$};
			\coordinate (LBP) at (LB1);
			\coordinate (RTP) at (RT1);
			\coordinate (LB1) at ($(LB1)-(Dx)-(Dy)$);
			\coordinate (RT1) at ($(RT1)+(Dx)+(Dy)$);
			\coordinate (LB2) at ($(LB2)-(Dx)-(Dy2)$);
			\coordinate (RT2) at ($(RT2)-(Dx)+(Dy2)$);
			\coordinate (HH) at ($(HH)+(Dy)$);
		\end{scope}
	}
	\filldraw[fill opacity=0.5, fill=gray] (6.2,.8)--(1.5,.8)--(6.2,1.3);
	\node at (6.5,1.05) {$t$};
	\node at (6,.5) {$\cdots$};
	\filldraw[fill opacity=0.5, fill=gray] (-4.1,.9)--(0.1,.9)--(-4.1,1.4);
	\node at (-4.4,1.15) {$s$};
	\node at (-3.9,.5) {$\cdots$};
%
\end{tikzpicture}		
\end{center}
\caption{Region $s$ lying inside $\sum_{j = 0}^3 r_j'$ and connecting
  the components of each $r_i'$
.}\label{fig:InfCmpElAiAi3}	
\end{figure}%
The only remaining occurrences of the contact predicate $C$ are
in~\eqref{eq:InfPart1}.  We deal with them by partitioning the
regions: instead of each $\neg C(r_i,r_{\md{i+2}})$ we consider the
equivalent conjunction of 4 formulas:
\begin{align*}
	& \neg C(r'_i,r'_{\md{i+2}}) & \land & \qquad \neg C(r_i\cdot(-r'_i), \ \ r_{\md{i+2}}\cdot(-r'_{\md{i+2}}))& \land \\
    & \neg C(r_i\cdot(-r'_i), \ \ r'_{\md{i+2}}) & \land & \qquad  \neg C(r'_i, \ \ r_{\md{i+2}}\cdot(-r'_{\md{i+2}})).
\end{align*}
The formulas in the second row are replaced by
\begin{align*}
\noncontact(r_i\cdot (-r'_i) + t,\ \ r_{\md{i+2}}' + s) \ \ \land \ \ \noncontact(r_i' + s, \ \ r_{\md{i+2}}\cdot (- r'_{\md{i+2}}) + t).
\end{align*}
Again, the resulting formula implies the original and is satisfied
by the configuration of Fig.~\ref{fig:InfCmpElAiAi3}.
The formulas in the first row are replaced by
\begin{multline*}
\noncontact(r'_i + s_i, \ \ r'_{\md{i+2}} + s_{\md{i+2}}) \land \noncontact(r_i\cdot(-r'_i) + t_i, \ \ r_{\md{i+2}}\cdot(-r'_{\md{i+2}}) + t_{\md{i+2}}).
\end{multline*}
Again, the resulting formula implies the original. To see that it is
still satisfiable, we select regions $s_0,\dots,s_3$, with $s_i$ and
$s_{i+2}$ disjoint ($i = 0,1$), such that each $s_i$ ($0 \leq i < 4$)
connects together the components of $r'_i$ as shown in
Fig.~\ref{fig:InfCmpElAiAi2}.
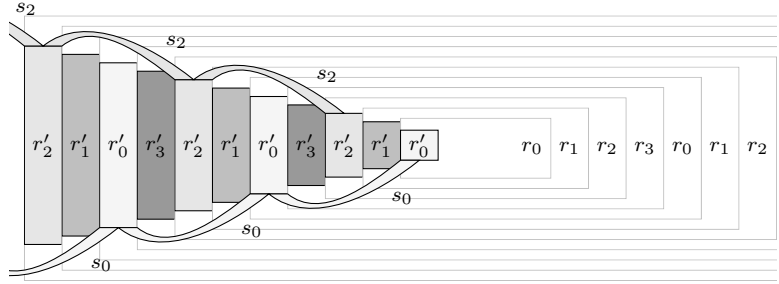
\begin{figure}[ht]
\begin{center}
\begin{tikzpicture}[yscale=0.8]\footnotesize
\clip (-5.2,-2) rectangle (5,3);
	\coordinate (LB1) at (0,0); 			
	\coordinate (RT1) at (2,1); 			
	\coordinate (LB2) at (0,.3); 			
	\coordinate (RT2) at (.5,.8);			
	\coordinate (LBP) at (.2pt,.2pt); 		
	\coordinate (RTP) at (-.2pt,-.2pt); 	
	\coordinate (Dx) at (0.5,0); 			
	\coordinate (Dy) at (0,.17); 			
	\coordinate (Dy2) at (0,.14); 			
	\coordinate (HH) at (.25,.5);			
	
	\foreach \i/\c in {0/gray!8,1/gray!50,2/gray!20,3/gray!80,0/gray!8,1/gray!50,2/gray!20,3/gray!80,0/gray!8,1/gray!50,2/gray!20}
	{
		\begin{scope}[even odd rule]
			\clip ($(LBP)-(.2pt,.2pt)$) rectangle ($(RTP)+(.2pt,.2pt)$)
				($(LB1)-(.2pt,.2pt)$) rectangle ($(RT1)+(.2pt,.2pt)$) ;
			\draw[draw=gray!50] (LB1) rectangle (RT1); \node at ($(RT1)-(HH)$) {$r_\i$};
			\filldraw[fill=\c] (LB2) rectangle (RT2) node[midway] {$r_\i'$};
			\coordinate (LBP) at (LB1);
			\coordinate (RTP) at (RT1);
			\coordinate (LB1) at ($(LB1)-(Dx)-(Dy)$);
			\coordinate (RT1) at ($(RT1)+(Dx)+(Dy)$);
			\coordinate (LB2) at ($(LB2)-(Dx)-(Dy2)$);
			\coordinate (RT2) at ($(RT2)-(Dx)+(Dy2)$);
			\coordinate (HH) at ($(HH)+(Dy)$);
		\end{scope}
	}
	
	\begin{scope}
	\coordinate (LB2) at (0,.3);	
	\foreach \x in {1,2,3}{
		\filldraw[fill=gray!8] ($(LB2)$) to[out=-135,in=-45] ($(LB2)-3*(Dx)-4*(Dy2)$)
			--++($-0.5*(Dx)$) to[in=-135,out=-45] ($(LB2)+.5*(Dx)$)
			node[very near end, below=4pt]{$s_0$} --cycle;
		\coordinate (LB2) at ($(LB2)-4*(Dx)-4*(Dy2)$);
	}
	\coordinate (LT2) at (0,.8);
	\coordinate (LT2) at ($(LT2)-2*(Dx)+2*(Dy2)$);
	\foreach \x in {1,2,3}{
		\filldraw[fill=gray!20] ($(LT2)$) to[out=135,in=45] ($(LT2)-3*(Dx)+4*(Dy2)$)
			--++($-0.5*(Dx)$) to[in=135,out=45] ($(LT2)+.5*(Dx)$)
			node[very near end, above=4pt]{$s_2$} --cycle;
		\coordinate (LT2) at ($(LT2)-4*(Dx)+4*(Dy2)$);
	}
	\end{scope}	
\end{tikzpicture}		
\end{center}
	\caption{Disjoint connected regions $s_0$ and $s_2$ containing the components of $r_0'$ and $r_2'$, respectively.}
	\label{fig:InfCmpElAiAi2}	
\end{figure}
In a symmetric way, select regions
$t_0,\dots,t_3$, with $t_i$ and $t_{i+2}$ disjoint ($i = 0,1$), such
that each $t_i$ ($0 \leq i < 4$) connects together the components of
$r_i \cdot (-r_i')$.
Transforming $\phi_\infty$ in the way just described, we obtain a
\cBc-formula $\phi_\infty^c$ with the required properties.
\begin{theorem}\label{cor:inftyBc}
The \cBc-formula $\phi_\infty^c$ is satisfiable over $\RC(\R^n)$,
\mbox{$n \geq 2$}. On the other hand, if $T$ is a locally connected,
unicoherent space, then any tuple from $\RC(T)$ satisfying
$\phi_\infty^c$ features sets that have infinitely many components.
\end{theorem}

The results of this section make no
reference to the language \cBci.  In fact, an analogue of
Theorem~\ref{cor:inftyBc} for \cBci{} will be proved in the special
case $n = 2$, in \SECT\ref{subsec:infCompPlane}, using a planarity
argument. For $n \geq 3$, however, this result fails, as we show in
\SECT\ref{sec:3d}.  As we observed above, Theorem~\ref{cor:inftyBc} shows
that, for all $n \geq 2$, $\Sat(\cBc,\RC(\R^n)) \neq
\Sat(\cBc,\RCP(\R^n))$. The reader will recall from \SECT\ref{sec:lwc}
that the corresponding inequations for the language \cBci{} hold
anyway, by~\eqref{eq:wiggly}. Finally, we remark on the case
of the real line, $\R$, which was considered in~\cite{KPZ10kr}.  The
analogue of Theorem~\ref{theo:inftyCc} for the case $n=1$ holds
(though we need to use a different formula to force an infinitude of
components); however, the analogue of Theorem~\ref{cor:inftyBc} for
$n=1$ fails: indeed, we have 
$\Sat(\cBc,\RC(\R))=\Sat(\cBc,\RCP(\R))$.

\section{Undecidability: the polyhedral case}\label{sec:undecidability}
We use the techniques developed in the previous section to prove that
the satisfiability problem for any of the languages \cBc, \cBCc{} or
\cBCci{} over the frame $\RCP(\R^n)$, $n \geq 2$, is undecidable.
Recall that a frame $(T, \bS)$ is unicoherent if $T$ is unicoherent;
and $(T, \bS)$ is finitely decomposible if, for all $s \in \bS$, there
exist connected elements $s_1, \dots, s_k$ of $\bS$, such that $s=
s_1+ \cdots + s_k$.

\begin{theorem}\label{theo:cBCcn}
Let $\cK$ be any class of unicoherent, finitely
decomposible frames, such that $\cK$ contains some frame of
the form $(\R^n,\bS)$, $n\geq2$, where $\RCP(\R^n) \subseteq \bS$. Then the
problem $\Sat(\cBCc, \cK)$ is r.e.-hard.
\end{theorem}
\begin{proof}
We proceed via a reduction of the Post correspondence problem (PCP),
constructing, for any instance $\pcpW$, a formula $\phi_\pcpW$ with
the property that (\emph{i}) if $\pcpW$ is positive then $\phi_\pcpW$
is satisfiable over $\RCP(\R^n)$, $n \geq 2$, and (\emph{ii}) if
$\phi_\pcpW$ is satisfiable over a unicoherent, finitely decomposible
frame then $\pcpW$ is positive.  The formula $\phi_\pcpW$ will be a
conjunction of \cBCc-literals. As in the proof
of Theorem~\ref{theo:inftyCc}, we equivocate between variables and the
regions to which they are assigned in some putative interpretation
over a frame in $\cK$: this will allow us to motivate the conjuncts of
$\phi_\pcpW$ as they are presented. In the remainder of this proof, if
$i$ is an integer, $\md{i}$ indicates the value of $i$ modulo 4.

Let the PCP-instance $\pcpW = (\pcpw^1,\pcpw^2)$ over alphabets $T$
and $U$ be given, and let $\vec{r} = (r_0, \dots, r_3)$ and $\vec{s} =
(s_0, \dots, s_3)$ be tuples of variables.  The first conjuncts of
$\phi_\pcpW$ ensure that $\vec{r}$ and $\vec{s}$ are sub-cyclic
partitions:
\begin{align}
\label{eq:partitionA}
& \cpartition(r_0, r_1, r_2, r_3),\\
\label{eq:partitionB}
& \cpartition(s_0, s_1, s_2, s_3).
\end{align}
By Lemma~\ref{lma:noCycles}, the component graphs $H(\vec{r})$ and
$H(\vec{s})$ are trees. Thus, any two vertices of $H(\vec{r})$ are
joined by a unique path, and likewise for $H(\vec{s})$.  The vertices
of $H(\vec{r})$ will be used to represent letters in some word
$\upsilon \in U^*$, and those of $H(\vec{s})$, letters in some word
$\tau \in T^*$.

Let $e^1$ and $e^2$ be fresh variables. We shall use these to
represent the morphisms $\pcpw^1$ and $\pcpw^2$, respectively. The
next conjuncts of $\phi_\pcpW$ ensure that, for all $0 \leq i <
4$, the components of both $r_i \cdot e^1$ and $r_i \cdot e^2$ are
coloured by the elements $\vec{s}$, as defined in \SECT\ref{prelim}:
\begin{align}
\label{eq:trackInclusions}
& \bigwedge_{k = 1}^2\bigwedge_{i=0}^3
\colourComp(r_i\cdot e^k; \ \ s_0, s_1, s_2, s_3).
\end{align}
Fig.~\ref{fig:onion} shows a configuration conforming to these
conditions. In this arrangement, $H(\vec{r})$ has vertices $\{A_1,
\dots, A_7\}$, where $A_i$ is a component of $r_{\md{i}}$, and
$H(\vec{s})$ has vertices $\{ B_1, \dots, B_4\}$ (indicated by thick
boundaries), where $B_j$ is a component of $s_{\md{j}}$. Observe that,
for $0 \leq i <4$ and $1 \leq k \leq 2$, each component of $r_i \cdot
e^k$ is included in exactly one of $s_0, \dots, s_3$, and hence in a
single vertex $B_j$ of $H(\vec{s})$; however, outside
$e^1 + e^2$, elements of $H(\vec{s})$ may intersect elements of
$H(\vec{r})$ without including them.  A word of warning: in the configuration of
Fig.~\ref{fig:onion}, the various sets $A_i \cdot e^k$ and $B_j \cdot
e^k$ are all \emph{connected}; however, the formula $\phi_\pcpW$ does
not enforce this. That is, there is nothing to prevent the sets $e^k$
from chopping elements of $H(\vec{r})$ and $H(\vec{s})$ into several
pieces.
\begin{figure}[ht]
\begin{center}
\begin{tikzpicture}[yscale=0.9]
\small
	\coordinate (LB1) at (0,0); 			
	\coordinate (RT1) at (4,1); 			
	\coordinate (LBP) at (.2pt,.2pt); 		
	\coordinate (RTP) at (-.2pt,-.2pt); 	
	\coordinate (Dx) at (0.5,0); 			
	\coordinate (Dy) at (0,.2); 			
	\coordinate (HH) at (.25,.5);			
	
	\foreach \i/\c in {1/gray!8,2/gray!50,3/gray!20,4/gray!80,5/gray!8,6/gray!50}
	{
		\begin{scope}[even odd rule]
			\clip ($(LBP)-(.2pt,.2pt)$) rectangle ($(RTP)+(.2pt,.2pt)$)
				($(LB1)-(.2pt,.2pt)$) rectangle ($(RT1)+(.2pt,.2pt)$) ;
			\filldraw[fill=\c] (LB1) rectangle (RT1);
            \node at ($(RT1)-(HH)+(0,0.2)$) {$A_\i$};
            \node at ($(LB1)+(HH)+(0,0.2)$) {$A_\i$};
			\coordinate (LBP) at (LB1);
			\coordinate (RTP) at (RT1);
			\coordinate (LB1) at ($(LB1)-(Dx)-(Dy)$);
			\coordinate (RT1) at ($(RT1)+(Dx)+(Dy)$);
			\coordinate (HH) at ($(HH)+(Dy)$);
		\end{scope}
	}
    \fill[fill opacity=0.5,fill=gray!5] (1,-1.25) rectangle (3,2.5);
    \draw (-2.75,-1.25) rectangle (1,2.5);
    \draw (3,-1.25) rectangle (6.75,2.5);
    \node at (-1,2.3) {$e^1$};
    \node at (5,2.3) {$e^2$};
    \begin{scope}[line width=1.5pt]
    \draw (3,1) -- (1,1.4) -- (-1,1.4) -- (-1,-0.4) -- (1,-0.4) -- (3,0) -- (4,0) -- (4,1) -- cycle;
    \draw (3,1.2) -- (1,1.8) -- (-2,1.8) -- (-2,-0.8) -- (1,-0.8) -- (3,-0.2) -- (4.5,-0.2) -- (4.5,1.2) -- cycle;
    \draw (-2.5,-1) rectangle (6.5,2);
    \end{scope}
    \node at (-0.75,-0.1) {$B_1$};
    \node at (-1.75,-0.5) {$B_2$};
    \node at (-2.25,-0.7) {$B_3$};
    \node at (-3.25,0.7) {$A_7$};
    \node at (-3.25,-0.1) {$= w^*$};
    \node at (-3.25,-0.9) {$= B_4$};
    \node at (3.75,0.3) {$B_1$};
    \node at (4.25,0.1) {$B_2$};
    \node at (6.25,-0.7) {$B_3$};
%
\end{tikzpicture}		
%
%
\end{center}
\caption{Arrangement satisfying
  {\eqref{eq:partitionA}--\eqref{eq:stopB}} such that $H(\vec{r})$ has
  nodes $A_1,\dots,A_7$ (different shades) and $H(\vec{s})$ has nodes
  $B_1,\dots,B_4$ (surrounded by thick lines); the unbounded region
  outside the largest of the depicted rectangles is $A_7 = B_4 =
  w^*$.}
\label{fig:onion}
\end{figure}
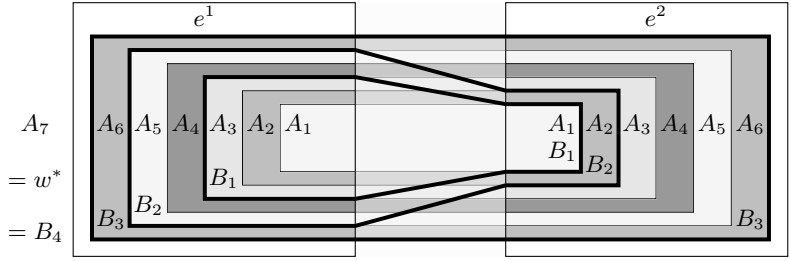

Let $w^*$ be a fresh variable.  The next conjuncts of $\phi_\pcpW$
ensure that the graphs $H(\vec{r})$ and $H(\vec{s})$ contain a common
vertex, $w^*$:
\begin{align}
& c(w^*) \wedge (w^* \neq 0),
\label{eq:stopW}\\
& \colourComp(w^*; \ \ r_0, \dots, r_3) \ \ \wedge \ \
  \bigwedge_{i=0}^{3} \colourComp(r_i; \ w^*, -w^*),
\label{eq:stopA}\\
& \colourComp(w^*; \ \ s_0, \dots, s_3) \ \ \wedge \ \
  \bigwedge_{i=0}^{3} \colourComp(s_i; \ w^*, -w^*).
\label{eq:stopB}
\end{align}
To see why, note that the first conjunct
of~\eqref{eq:stopA} ensures that $w^*$ is included in one of the sets
$r_i$, and hence---since it is connected, by~\eqref{eq:stopW}---in one
of the vertices of $H(\vec{r})$; on the other hand, the remaining
conjunct of~\eqref{eq:stopA} ensures that every vertex of $H(\vec{r})$
is included in either $w^*$ or $-w^*$. Since $w^*$ is non-empty, it
must therefore be identical to a single vertex of $H(\vec{r})$. The
same conclusion holds for $H(\vec{s})$ using~\eqref{eq:stopB}.
In the arrangement of Fig.~\ref{fig:onion}, we have $w^* = A_7 = B_4$.

We need to impose a little more structure on the graphs $H(\vec{r})$
and $H(\vec{s})$. Let $r'_0, \dots, r'_3$ and $w_1$ be fresh variables, and let
$\phi_\pcpW$ contain the conjuncts: 
\begin{align}
%
\label{eq:init1}
& c(w_1) \ \wedge \ (w_1 \leq r_1)  \ \wedge \ (w_1 \leq s_1) \ \wedge \ (w_1 \cdot w^* = 0),\\
\label{eq:onionA1}
& \bigwedge_{i= 0}^3 (r'_i \leq r_i),\\
\label{eq:init2}
& \bigwedge_{k = 1}^2 (e^k \cdot r'_1 \cdot w_1  \neq 0).
\end{align}
Since $w_1$ is a non-empty, connected subset of both $r_1$ and $s_1$,
let $A_1$ be the component of $r_1$ including $w_1$, and let $B_1$ be
the component of $s_1$ including $w_1$. It follows that $A_1 \leq
B_1$; the final conjunct of~\eqref{eq:init1} ensures that $A_1$ and
$B_1$ are both distinct from $w^*$. In the sequel, we shall construct
a path in the graph $H(\vec{r})$ from $A_1$ to $w^*$, and a path in
the graph $H(\vec{s})$ from $B_1$ to $w^*$. The proof will hinge on
analysing the properties of these paths.

Let $t$ be a fresh variable, and let $\phi_\pcpW$ contain the conjuncts:
\begin{align}
\label{eq:onionA1.5}
& \bigwedge_{k = 1}^2 (e^k \cdot t \neq 0),\\
\label{eq:onionA2}
& \bigwedge_{k = 1}^2 \bigwedge_{i = 0}^{3} c(e^k \cdot ((r'_i \cdot (-w^*)) + r_{\md{i+1}} + t)),\\
\label{eq:onionA3}
& \bigwedge_{i = 0}^3 \neg C(r'_i,\  t),\\
\label{eq:onionA4}
& \bigwedge_{i = 0}^3 \neg C(r'_i, \ \ r_{\md{i+1}} \cdot (-r'_{\md{i+1}})).
\end{align}
From~\eqref{eq:init2}, select a point $q^k_1$ in the interior of $e^k
\cdot r'_1\cdot w_1$. By~\eqref{eq:init1}, $q^k_1 \notin w^*$. Let
$X^k_1$ be the component of $e^k \cdot r_1$ containing $q^k_1$, and
$A^k_1$ the component of $r_1$ including $X^k_1$; note that $A^k_1$ is
a vertex of the graph $H(\vec{r})$. Evidently, $A_1^k\ne w^*$, and so
$X^k_1 \leq (e^k \cdot r_1 \cdot (-w^*))$. Now suppose $X^k_i$ has
been defined and contains some point $q^k_i \in e^k\cdot
r'_{\md{i}}\cdot
(-w^*)$. From~\eqref{eq:onionA1.5}--\eqref{eq:onionA3}, $X^k_i$
contains a point $q^k_{i+1} \in e^k \cdot r_{\md{i+1}}$, which,
by~\eqref{eq:onionA4}, is in fact in $e^k \cdot r'_{\md{i+1}}$. Let
$X^k_{i+1}$ be the component of $e^k \cdot r_{\md{i+1}}$ containing
$q^k_{i+1}$, and $A^k_{i+1}$ the component of $r_{\md{i+1}}$ including
$X^k_{i+1}$; again, $A^k_{i+1}$ is a vertex of $H(\vec{r})$.  Note
that either $q^k_{i+1}\in w^*$ or $q^k_{i+1}\notin w^*$, and in the
latter case, $q^k_{i+1} \in e^k \cdot r'_{\md{i+1}} \cdot
(-w^*)$. This process either continues forever, or, at some point,
$q^k_{i+1}\in w^*$. But now consider any sequence $A^k_1, A^k_2, \dots
A^k_\ell$ obtained in this way. Evidently, $(A_i,A_{i+1})$ is an edge
of $H(\vec{r})$ for all $1 \leq i < \ell$; moreover, since $A_i \leq
r_{\md{i}}$, we see from~\eqref{eq:partitionA} that $A_i \neq A_{i+2}$
for all $1 \leq i < \ell-1$, whence, since $H(\vec{r})$ is a tree,
$A^k_1, A^k_2, \dots A^k_\ell$ is a path (i.e., has no repeated
nodes). It follows that, for some value of $i$, denoted by $n^k$, the
condition $q_{i+1}^k \in w^*$ must hold, for otherwise, $H(\vec{r})$
would contain an infinite path, contradicting the assumption that the
frame in question is finitely decomposible. Since $q_{n^k+1}\in
r_{n^k+1}$, we have $A^k_{n^k+1}=w^*$, and hence, for $k = 1, 2$,
there is a path $A^k_1, A^k_2, \dots, A^k_{n^k+1}$ in $H(\vec{r})$
from $A^k_1 = A_1$ to $A^k_{n^k+1} = w^*$.  Indeed, this must be the
\emph{same} path for both $k = 1,2$, so that we may drop the
$k$-superscripts, and write:
\begin{equation*}
	A_1, A_2, \dots, A_n, A_{n+1},\qquad  \text{where } A_{n+1} = w^*.
\end{equation*}
(Note that the letter $n$ here is simply a convenient label for the
length of this path: it has nothing to do with the dimension of the
space.) It is important to remember that the sets $X^1_i$ and $X^2_i$,
for a fixed value of $i$, will in general be distinct
(Fig.~\ref{fig:path}).
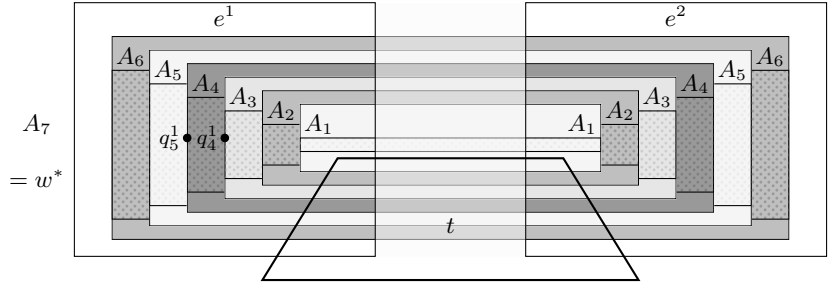
\begin{figure}[ht]
\begin{center}
\begin{tikzpicture}[>=latex,point/.style={circle,draw=black,minimum size=1mm,inner sep=0pt},yscale=0.9]\scriptsize
\small
	\coordinate (LB1) at (0,0); 			
	\coordinate (RT1) at (4,1); 			
	\coordinate (LBP) at (.2pt,.2pt); 		
	\coordinate (RTP) at (-.2pt,-.2pt); 	
	\coordinate (Dx) at (0.5,0); 			
	\coordinate (Dy) at (0,.2); 			
	\coordinate (HH) at (.25,.5);			
	
	\foreach \i/\c/\f in {1/gray!8/gray!16,2/gray!50/gray!80,3/gray!20/gray!40,4/gray!80/gray,5/gray!8/gray!16,6/gray!50/gray!80}
	{
		\begin{scope}[even odd rule]
			\clip ($(LBP)-(.2pt,.2pt)$) rectangle ($(RTP)+(.2pt,.2pt)$)
				($(LB1)-(.2pt,.2pt)$) rectangle ($(RT1)+(.2pt,.2pt)$) ;
			\filldraw[fill=\c] (LB1) rectangle (RT1);
			\filldraw[pattern color=\f,pattern=crosshatch dots] ($(LB1)+(0,0.3)$) rectangle ($(RT1)-(0,0.5)$);
            \node at ($(RT1)-(HH)+\i*(0,0.2)$) {$A_\i$};
            \node at ($(LB1)+(HH)+\i*(0,0.2)$) {$A_\i$};
			\coordinate (LBP) at (LB1);
			\coordinate (RTP) at (RT1);
			\coordinate (LB1) at ($(LB1)-(Dx)-(Dy)$);
			\coordinate (RT1) at ($(RT1)+(Dx)+(Dy)$);
			\coordinate (HH) at ($(HH)+(Dy)$);
		\end{scope}
	}
    \fill[fill opacity=0.5,fill=gray!5] (1,-1.25) rectangle (3,2.5);
    \draw (-3,-1.25) rectangle (1,2.5);
    \draw (3,-1.25) rectangle (7,2.5);
    \node at (-1,2.3) {$e^1$};
    \node at (5,2.3) {$e^2$};
%
    \node at (-3.5,0.7) {$A_7$};
    \node at (-3.5,-0.1) {$= w^*$};
    \node at (-1,0.5) [point,fill=black,label=left:{\footnotesize $q_4^1$\hspace*{-0.3em}}] {};
    \node at (-1.5,0.5) [point,fill=black,label=left:{\footnotesize $q_5^1$\hspace*{-0.3em}}] {};
    \draw[thick] (0.5,0.2) -- ++(3,0) -- ++(1,-1.8) -- ++(-5,0) -- cycle;
    \node at (2,-0.8) {$t$};
%
\end{tikzpicture}		
%
\end{center}
\caption{The regions $A_1, \dots, A_7$ and $e^1$, $e^2$ from Fig.~\ref{fig:onion}. Each region $A_i$ includes two components of $r'_{\md{i}}$: one lying in $e^1$, the other in $e^2$. Given a point $q_i^k\in e^k\cdot r'_{\md{i}} \cdot (-w^*)$, here illustrated for $i = 4$ and $k = 1$, formulas~\eqref{eq:onionA1}--\eqref{eq:onionA4} ensure the existence of a point $q_{i+1}^k$ in the same component of $e^k\cdot r_{\md{i}}$, and also in $e^k\cdot r'_{\md{i+1}}$.}
\label{fig:path}
\end{figure}

Let us now turn our attention to the graph $H(\vec{s})$.  Fix the
value of $k$ for the moment ($1 \leq k \leq 2$), and consider the
sequence $X^k_1, \dots, X^k_{n+1}$. Since $X^k_i$ is a
connected subset of $e^k \cdot r_{\md{i}}$, it follows
from~\eqref{eq:trackInclusions} that each $X^k_i$ is included in some
vertex of $H(\vec{s})$, say, $\hat{B}^k_i$.  Thus, for $k = 1$ and $k=2$, we have a sequence
\begin{equation*}
\hat{B}^k_1, \dots, \hat{B}^k_{n+1}, \qquad \text{where } A_1 \leq \hat{B}^k_1 \ \ \text{ and } \ \ \hat{B}^k_{n+1} =
w^*.
\end{equation*}
Of course, these sequences may contain adjacent duplicates, since
there is nothing to stop $X^k_i$ and $X^k_{i+1}$ being included in the
same vertex of $H(\vec{s})$. Furthermore, the two sequences (for $k =
1,2$) may be distinct, since, for fixed $i$, there is nothing to stop
$X^1_i$ and $X^2_i$ lying in different vertices of $H(\vec{s})$; see
Figs.~\ref{fig:onion} and~\ref{fig:path}. But now suppose we remove adjacent duplicates, obtaining
sequences:
\begin{align*}
B^k_1, \dots , B^k_{m^k+1}, \qquad \text{ where } A_1 \leq B^k_1 \ \ \text{ and } \ \ B^k_{m^k+1} = w^*
\end{align*}
with $m^k \leq n$.
Thus, every $B^k_j$ is the result of coalescing a contiguous block of
identical vertices $\hat{B}^k_i$. Evidently, $(B^k_j, B^k_{j+1})$ must be an edge of
$H(\vec{s})$, for $1 \leq j \leq m^k$.

Now let $\phi_\pcpW$ contain the conjuncts:
\begin{align}
\label{eq:progress}
& \bigwedge_{k =1}^2\bigwedge_{i=0}^3\bigwedge_{j=0}^3
       \neg C(e^k \cdot r_i \cdot s_j, \ \ e^k \cdot r_{\md{i+1}} \cdot s_{\md{j-1}}).
\end{align}
We claim that, in that case, $B^k_j \leq s_{\md{j}}$ for all $1 \leq j
\leq m^k$.  The proof is by induction on $j$. By~\eqref{eq:init1},
$X^k_1 \leq s_1$, and therefore $B^k_1 \leq s_1$.  Suppose, then,
$B^k_j \leq s_{\md{j}}$ for some $1 \leq j < m^k$.  Let $\hat{B}^k_i$,
be the last vertex in the block coalescing to $B^k_j$, so that
$\hat{B}^k_{i+1}$ is the first element of the block coalescing to
$B^k_{j+1}$.  Thus, $X^k_i \leq r_{\md{i}}$ and $X^k_{i+1} \leq
r_{\md{i+1}}$.  But~\eqref{eq:partitionB} and~\eqref{eq:progress} then
ensure that either $B^k_{j+1} \leq s_{\md{j}}$ or $B^k_{j+1} \leq
s_{\md{j+1}}$; and the former is impossible, since then $\hat{B}^k_i$
and $\hat{B}^k_{i+1}$ would have coalesced to the same block. This
proves the claim.  By~\eqref{eq:partitionB} (and the fact that
$B^k_{m^k-1} \neq w^*$), we then have $B^k_j \neq
B^k_{j+2}$ for all $1 \leq j < m^k$.  And since $H(\vec{s})$ is a
tree, it follows that $B^k_1, \dots, B^k_{m^k+1}$ is a path through
$H(\vec{s})$ with $B_1^1 = B_1^2$ and $B^1_{m^1+1} = B^2_{m^2+1} =
w^*$. Indeed, this is the \emph{same} path through $H(\vec{s})$ for
both values of $k$, so that we can again drop the superscripts, and
just write
\begin{equation*}
	B_1, \dots , B_m, B_{m+1},\qquad \text{where } B_{m+1} = w^*.
\end{equation*}

Taking stock, we see that, for each $k = 1,2$, the path
$A_1, \dots, A_n$ may be grouped into $m$ contiguous blocks $E^k_1,\dots,E^k_m$ by taking the
vertex $A_i$ to be in the $j$th block $E^k_j$ just in case $e^k \cdot X^k_i \leq B_j$. We may depict
this grouping as follows:
\begin{equation*}
  A_1,\dots,A_n = \underbrace{A^k_{1,1}, \dots, A^k_{1,h^k_1}}_{E^k_1},\dots,
  \underbrace{A^k_{j,1}, \dots, A^k_{j,h^k_j}}_{E^k_j},\dots,
  \underbrace{A^k_{m,1}, \dots, A^k_{m,h^k_m}}_{E^k_m}.
\end{equation*}
It is important to realize that, although there is only one path $A_1,
\dots, A_{n+1}$ and one path $B_1, \dots, B_{m+1}$, the two values $k = 1$
and $k = 2$ will in general give rise to different groupings of the
vertices of the former into blocks corresponding to the vertices of the latter (hence, the two sequences of indices $h^k_1,\dots,h^k_m$).

\bigskip

Recall the PCP-instance $\pcpW = (\pcpw^1, \pcpw^2)$ over the
alphabets $T$ and $U$, which we wish to encode.  We regard the
elements of these alphabets as fresh variables, and order them in some
way to form tuples $\vec{t}$ and $\vec{u}$. We use these variables to
colour the vertices of $H(\vec{s})$ and $H(\vec{r})$, respectively, by
taking $\phi_\pcpW$ to contain the conjuncts:
\begin{align}
\label{eq:patternLabels}& \partition(\vec{t}) \ \ \land \ \ \bigwedge_{j = 0}^3 \colourComp(s_j; \ \vec{t}),\\
\label{eq:letterLabels} & \partition(\vec{u}) \ \ \land \ \ \bigwedge_{i = 0}^3 \colourComp(r_i; \ \vec{u}).
\end{align}
In this way, the path $A_1, \dots, A_n$ defines a word $\upsilon \in
U^*$, and the path $B_1, \dots, B_m$ defines a word $\tau \in
T^*$. Using the groupings of the sequence $A_1, \dots, A_n$ obtained
above, we shall write conjuncts of $\phi_\pcpW$ ensuring that
$\pcpw^k(\tau) = \upsilon$ for $k = 1,2$.  This will mean that, if
$\phi_\pcpW$ has a satisfying assignment over some frame $\bS \in
\cK$, then the PCP-instance $\pcpW$ is positive.

For $k = 1, 2$, let $\bigl\{ p^k_{h,\ell} \mid 1 \leq h \leq |T|, 1 \leq \ell \leq |\pcpw^k(t_h)| \}$,
be a collection of fresh variables, enumerated in some way as
$\vec{p}^k$, which we shall use to colour the vertices of
$H(\vec{r})$. That is, we add to $\phi_\pcpW$ the conjuncts:
\begin{align}
\label{eq:positionLabels}
& \bigwedge_{k=1}^2\bigl(\partition(\vec{p}^k)  \ \land \ \bigwedge_{i=0}^3 \colourComp(r_i; \ \vec{p}^k)\bigr).
\end{align}
We
refer to these variables as \emph{position colours}, because we are
to think of $p^k_{h,\ell}$ as denoting the $\ell$th position in the
word $\pcpw^k(t_h)$. In particular, any position colour $p^k_{h,\ell}$
is naturally associated to the letter $t_h$ of $T$.  Fixing $k$ for
the moment, consider the vertices $A^k_{j,1}, \dots, A^k_{j,h^k_j}$
grouped into the $j$th block, $E_j^k$. What ensures that these
vertices belong to one block is the existence of a single $B_j$ such
that, if $A_i$ is one of these vertices, then the corresponding set
$X^k_i \leq A_i$ is included in $B_j$.  By~\eqref{eq:patternLabels},
$X^k_i \leq t_h$, for some $1 \leq h \leq |T|$.  It follows that the
conjuncts
\begin{align}
& \bigwedge_{k = 1}^2 \bigwedge_{h=1}^{|T|} \bigl(t_h \leq \sum_{\ell = 1}^{|\pcpw^k(t_h)|} p^k_{h,\ell}\bigr)
\end{align}
ensure that the vertices in this block are assigned
`position' colours associated to the common tile $t_h$.

We proceed to write constraints ensuring that these colours are
assigned in exactly the canonical order: $p^k_{h,1}, \dots
p^k_{h,|\pcpw^k(t_h)|}$ (from which, incidentally, it follows that
$|\pcpw^k(t_h)| = h^k_j$).
We begin by adding to $\phi_\pcpW$ the conjuncts:
\begin{align}
\label{eq:block1j}
& \bigwedge_{k = 1}^2
    \bigwedge_{h = 1}^{|T|}  \bigwedge_{\ell = 2}^{|\pcpw^k(t_h)|}
     \bigl((w_1 \cdot p^k_{h,\ell} = 0) \ \land \ \bigwedge_{i = 0}^3\neg C(s_i, \ s_{\md{i+1}} \cdot p^k_{h,\ell})\bigr).
\end{align}
These ensure that the first vertex of each block is assigned one of
the colours $p^k_{h,1}$, for $1 \leq h \leq |T|$. The rules for
colouring successive vertices can now be simply stated.  Consider the
following binary relation on the variables in $\vec{p}^k$:
\begin{multline*}
\Psi^k = \bigl\{ (p_{h,\ell}^k, p_{h,\ell + 1}^k )
         \mid 1 \leq h \leq |T|,
                                         1 \leq \ell < |\pcpw^k(t_h)|\bigr\}\\
  \cup
     \bigl\{ (p_{h,|\pcpw^k(t_h)|}^k, p_{h', 1}^k) \mid 1 \leq h, h' \leq |T| \bigr\}.
\end{multline*}
This relation captures the rules of possible succession for colouring
by the variables $p_{h,\ell}^k$: (\emph{a}) if $A_i$ is coloured $p_{h,\ell}^k$,
where $\ell$ indicates a non-final position in the word
$\pcpw^k(t_h)$, then $A_{i+1}$ must be coloured $p_{h,\ell+1}^k$;
(\emph{b}) if $A_i$ is coloured $p_{h,|\pcpw^k(t_h)|}^k$, indicating the final
position in the word $\pcpw^k(t_h)$, then $A_{i+1}$ must be coloured
$p_{h',1}^k$ for some $1 \leq h' \leq |T|$. We therefore
add to $\phi_\pcpW$ the conjuncts
\begin{equation}
\bigwedge_{k=1}^2
\bigwedge_{i=0}^3\hspace*{1em}
\bigwedge_{(p^k_{h,\ell}, p^k_{h',\ell'}) \notin \Psi^k}
\neg C(p^k_{h,\ell} \cdot r_i, \ \ p^k_{h',\ell'} \cdot r_{\md{i+1}}).
\label{eq:blockjk}
\end{equation}
We also ensure that each block spells out only one word. That is, we
ensure that no vertex of the sequence $A_1,\dots,A_n$ can be coloured
with the starting position in a word if the previous vertex belongs to
the same block:
\begin{equation}
\bigwedge_{k = 1}^2\bigwedge_{i = 0}^3\bigwedge_{j = 0}^3\bigwedge_{h = 1}^{|T|}
\neg C( r_i\cdot s_j, \ \  r_{\md{i+1}} \cdot s_j \cdot p^k_{h,1}).
\end{equation}
Lastly, we ensure that the final vertex of the final block corresponds to the final position in a word. In other words,
we ensure that the vertex $A_n$ (which contacts $A_{n+1} = w^*$) is
coloured $p^k_{h,|\pcpw^k(t_h)|}$, for some $1 \leq h \leq |T|$:
\begin{align}
\label{eq:endBlocks}
& \bigwedge_{k = 1}^2 \bigwedge_{h=1}^{|T|} \bigwedge_{\ell=1}^{(|\pcpw^k(t_h)|-1)}
   \neg C(p^k_{h,\ell}, \ w^*).
\end{align}
%


At this stage, we have ensured that, for $k=1$ and $k=2$, vertices of
each block $A^k_{j,1}, \dots, A^k_{j,h_j^k}$, $1\leq j\leq m$, are
coloured $p^k_{h,1}, \dots, p^k_{h,|\pcpw^k(t_h)|}$, where $t_h$ is
the $j$th letter of the word $\tau$.  This easily enables us to
enforce the sought-after conditions $\pcpw^k(\tau) = \upsilon$ for $k
= 1,2$.  Denoting by $u^k_{h,\ell}$ the variable in $\vec{u}$
(i.e. that letter of the alphabet $U$) that is the $\ell$th letter in
the word $\pcpw^k(t_h)$, we add to $\phi_\pcpW$ the conjuncts:
\begin{align}
& \bigwedge_{k = 1}^2
\bigwedge_{h = 1}^{|T|}
\bigwedge_{\ell = 1}^{|\pcpw^k(t_h)|} (p^k_{h,\ell} \leq u^k_{h,\ell}).
\label{eq:morphism}
\end{align}
That $\pcpw^k(\tau) = \upsilon$ for $k = 1,2$ then follows from the
fact that each vertex $A_i$ is assigned a unique colour from
$\vec{u}$. Thus, if $\phi_\pcpW$ is satisfiable over $\cK$
then $\pcpW$ is positive.

Conversely, if $\pcpW$ is positive, it is obvious that $\phi_\pcpW$
may be satisfied over $\RCP(\R^n)$, $n \geq 2$, by suitably extending
a configuration similar to that shown in Fig.~\ref{fig:onion}.
\end{proof}

\begin{corollary}\label{cor:cBCcin}
Let $\cK$ be any class of unicoherent, finitely
decomposible frames, such that $\cK$ contains some frame of
the form $(\R^n,\bS)$, $n\geq2$, where $\RCP(\R^n) \subseteq \bS$. Then the
problem $\Sat(\cBCci, \cK)$ is r.e.-hard.
\end{corollary}
\begin{proof}
We start with the formula $\phi_\pcpW$ of Theorem~\ref{theo:cBCcn},
and replace all occurrences of $c$ by $\ic$. Denote the resulting
\cBCci-formula by $\ti{\phi}_\pcpW$.  Since all atoms of the form
$c(\tau)$ in $\phi_\pcpW$ occur with positive polarity,
$\ti{\phi}_\pcpW$ entails $\phi_\pcpW$. On the other hand, by
inspection of Fig.~\ref{fig:onion}, we see that if $\pcpW$ is
positive, then $\ti{\phi}_\pcpW$ will be satisfiable in $\RCP(\R^2)$,
and hence in $\RCP(\R^n)$ for all $n \geq2$.  This proves the
corollary.
\end{proof}

\begin{corollary}\label{cor:cBcin}
Let $\cK$ be any class of unicoherent, finitely
decomposible frames, such that $\cK$ contains some frame of
the form $(\R^n,\bS)$, $n\geq2$, where $\RCP(\R^n) \subseteq \bS$. Then the
problem $\Sat(\cBc, \cK)$ is r.e.-hard.
\end{corollary}
\begin{proof}
	Consider again the formula $\phi_\pcpW$ of
        Theorem~\ref{theo:cBCcn}.  Since all occurrences of the
        predicate $C$ in $\phi_\pcpW$ have negative polarity, we can
        replace them with \cBc-formulas as we did in the proof of
        Theorem~\ref{cor:inftyBc}.  The resulting formula
        $\phi^c_\pcpW$ implies $\phi_\pcpW$, and is satisfiable in
        $\RCP(\R^n)$, for all $n \geq 2$, whenever $\pcpW$ is
        positive.
\end{proof}

\section{Undecidability: the plane case}\label{sec:2d}
In \SECT\ref{infcomp}, we established that, if $\cL$ is any of the
languages \cBc, \cBCc{} or \cBCci, there exists an $\cL$-formula
that is satisfiable over $\RC(\R^n)$, $n \geq 2$, but only by
regions having infinitely many components. Nothing was mentioned in this
regard about the language \cBci. In \SECT\ref{sec:undecidability},
we established the undecidability of $\Sat(\cL, \RCP(\R^n))$, $n \geq
2$, where $\cL$ is any of the languages \cBc, \cBCc, \cBCci.
Nothing was mentioned in this regard about the problems
$\Sat(\cBci,\RCP(\R^n))$, $n \geq 2$, or indeed about the problems
$\Sat(\cL,\RC(\R^n))$ where $\cL$ is any of \cBc, \cBci, \cBCc,
\cBCci.  In this section, we complete the picture in the case
$n=2$. Specifically, we establish the existence of a \cBci-formula
satisfiable over $\RC(\R^2)$, but only by regions having infinitely
many components; and we establish the undecidability of the problems
$\Sat(\cL,\RC(\R^2))$ and $\Sat(\cL,\RCP(\R^2))$, where $\cL$ is any
of \cBc, \cBci, \cBCc{} or \cBCci.

We employ the standard terminology of Jordan arcs and curves: a {\em
  non-degenerate Jordan arc} is a continuous, 1--1 function $\alpha$
from the unit interval to $\R^2$; a {\em degenerate Jordan arc} is a
constant function from the unit interval to $\R^2$; a {\em Jordan arc}
is a degenerate Jordan arc or a non-degenerate Jordan arc.  A {\em
  Jordan curve} is a continuous, 1--1 function from the unit circle to
$\R^2$. Where no confusion results, we identify Jordan arcs and curves
with their loci (ranges). If $\alpha_1$ and $\alpha_2$ are Jordan arcs
which intersect in the unique point $\alpha_1(1) = \alpha_2(0)$, then
we write $\alpha_1\alpha_2$ to denote, ambiguously, any Jordan arc
$\alpha$ with locus $\alpha_1 \cup \alpha_2$ such that $\alpha(0) =
\alpha_1(0)$ and $\alpha(1) = \alpha_2(1)$.  We employ the following
notation: if $\alpha$ is a Jordan arc, $\alpha^{-1}$ denotes a Jordan
arc with the same locus but opposite direction, e.g., $\alpha^{-1}(t)
= \alpha(1-t)$, for all $0 \leq t \leq 1$.  If, in addition, $p_1 =
\alpha(t_1)$, $p_2 = \alpha(t_2)$ are points on $\alpha$ with $t_2
\geq t_1$, $\alpha[p_1,p_2]$ denotes a Jordan arc whose locus is the
segment of $\alpha$ between $p_1$ and $p_2$, and which has the same
direction as $\alpha$: $\alpha[p_1,p_2](t) = \alpha(t_1 + t(t_2 -
t_1))$.
An \emph{end-cut} to $p$ in a set $X$ is a Jordan arc $\alpha\subseteq
\ti{X}\cup\{p\}$ such that $\alpha(1) = p$.  A \emph{cross-cut} in $X$
is a Jordan arc $\alpha$ in $X$ intersecting the boundary $\delta X$
of $X$ only at its endpoints $\alpha(0)$ and $\alpha(1)$.

\subsection{Forcing infinitely many components with \cBci{}}
\label{subsec:infCompPlane}
We begin by showing that there exists a \cBci-formula that is
satisfiable over $\RC(\R^2)$, but only by regions having infinitely
many components.  Many of the techniques we employ will prove useful
in \SECT\ref{subsec:undecPlane}. Our basic tools are two formulas that
enable us to construct Jordan arcs and curves containing points in
specified regions. But before presenting
these formulas, we need to establish the following property of regular
closed sets:
\begin{lemma}\label{lemma:rc-intersect}
Let $T$ be any topological space, and $a,b_1$ and $b_2$ elements of
$\RC(T)$ such that $b_1 \cdot b_2 = 0$. Then $\ti{(a +
  b_1)}\cap\ti{(a + b_2)} = \ti{a}$.
\end{lemma}
\begin{proof}
Note that, for any $s \in \RC(T)$, $-s = T \setminus \ti{s}$.  Since
$b_1 \cdot b_2 = 0$, we have
$-a  = \sum_{i = 1,2} (-(a +b_i))$, which is then equal to
$\bigcup_{i = 1,2} (T \setminus \ti{(a + b_i)})  = T \setminus (\bigcap_{i = 1,2} \ti{(a + b_i)})$.
\end{proof}

Consider now the following \cBci-formula:
\begin{multline*}
\frameFlai(r_0,\dots,r_{n-1}) \ \ = \bigwedge_{i = 0}^{n-1} \bigl((r_i\neq 0) \land 
		\ic(r_i+ r_{\md{i+1}}) \bigr) \ \land
\bigwedge_{0 \leq i < j < n}(r_i\cdot r_j=0),	
\end{multline*}
where $\md{k}$ denotes the value of $k$ modulo $n$. This formula allows us to
construct Jordan curves that contain points of all regions $r_0,\dots,r_{n-1}$:
\begin{lemma}\label{lma:FrameLemmaInt}
Fix $n \geq 3$, and let $(a_0, \dots, a_{n-1})$ be a tuple of elements
of $\RC(\R^2)$ satisfying $\frameFlai(r_0, \dots,r_{n-1})$.  Then
there exist Jordan arcs $\alpha_0, \dots, \alpha_{n-1}$ and points
$p_0,\dots,p_{n-1}$ such that: for all $0 \leq i < n$, $\alpha_i$ is a
Jordan arc from $p_i$ to $p_{\md{i+1}}$, with $\alpha_i \subseteq
\ti{(a_i+ a_{\md{i+1}})}$; $\alpha_0 \cdots \alpha_{n-1}$ is a Jordan
curve lying in $\ti{(a_0+\cdots+a_{n-1})}$; and $p_i\in\ti{a}_i$, for
all $0 \leq i < n$.
\end{lemma}
\begin{proof}
For every $0\leq i<n$, select points $p_i'$ in the interior of $a_i$
and connect each $p'_i$ to $p'_{\md{i+1}}$ with an arc
$\alpha_i''\subseteq \ti{(a_i+a_{\md{i+1}})}$. Let $p_1$ be the first
point on $\alpha_0''$ that is on $\alpha_1''$, let $\alpha_0'$ be the
initial segment of $\alpha_0''$ ending at $p_1$, and let $\alpha_1'$
be the final segment of $\alpha_1''$ starting at $p_1$. Note that
$\alpha_0'\cap\alpha_1'=\{p_1\}$. For $2\leq i<n-1$, let $p_i$ be the
first point on $\alpha_{i-1}'$ that is on $\alpha_i''$, let
$\alpha_{i-1}$ be the initial segment of $\alpha_{i-1}'$ ending at
$p_i$, and let $\alpha_i'$ be the final segment of $\alpha_i''$
starting at $p_i$. Note that $\alpha_{i-1}\cap \alpha_i'=\{p_i\}$.
Finally, let $p_0$ be the first point on $\alpha'_{n-1}$ that is on $\alpha_0'$, let $\alpha_{n-1}$ be the initial
segment of $\alpha_{n-1}'$ ending at $p_0$, and let $\alpha_0$ be the
final segment of $\alpha_0'$ starting at $p_0$. Note that
$\alpha_{n-1}\cap\alpha_0=\{p_0\}$.  By construction, for every $0\leq
i<n$, $\alpha_i$ connects points $p_i$ and $p_{\md{i+1}}$, and
$\alpha_{\md{i-1}}\cap\alpha_i=\{p_i\}$, whence, by
Lemma~\ref{lemma:rc-intersect}, $p_i\in \ti{a}_i$.
\end{proof}

Consider now the \cBCci-formula, for $n > 1$,
\begin{multline*}
 \stacki(r_1,\dots, r_n) \ \ = \\ \bigwedge_{i=1}^n \ic(r_i+\cdots+r_n) \ \ \ \land
	\bigwedge_{1 \leq i < j \leq n} (r_i\cdot r_j=0) \ \ \ \land \ \bigwedge_{\begin{subarray}{c}1 \leq i<j \leq n\\j - i > 1\end{subarray}}  \neg C(r_i, r_j),
\end{multline*}
which will allow us to construct arcs containing points of all the regions $r_1,\dots,r_n$:
\begin{lemma}\label{lma:StackLemmai}
Let $(a_1, \dots, a_n)$ be a tuple of elements of $\RC(\R^2)$ satisfying 
\linebreak
$\stacki(r_1,\dots,r_n)$.
Then every point $p_1\in \ti a_1$ can
be connected to every point $p_n\in \ti a_n$ by a Jordan arc
$\alpha=\alpha_1 \cdots \alpha_{n-1}$ such that, for all $1\leq i< n$, $\alpha_i$ is a non-degenerate
Jordan arc in $\ti{(a_i+a_{i+1})}$, starting at a point $p_i\in\ti a_i$.
\end{lemma}
\begin{proof}
Since $a_1+\dots+a_n$ is interior-connected, let $\alpha_1' \subseteq
\ti{(a_1+\dots+a_n)}$ be a Jordan arc connecting $p_1$ to $p_n$.
Since $\neg C(a_1,(a_3 + \dots + a_n))$, $\alpha_1'$ must contain a
point $p_1' \in \ti{a}_2$ such that $\alpha'_1[p_1,p'_1] \subseteq
\ti{(a_1 + a_{2})}$.  For convenience, let $p_0 = p_1$, let $\alpha_0$
be the degenerate Jordan arc located at $p_1$, and let $a_0$ be the
empty region.

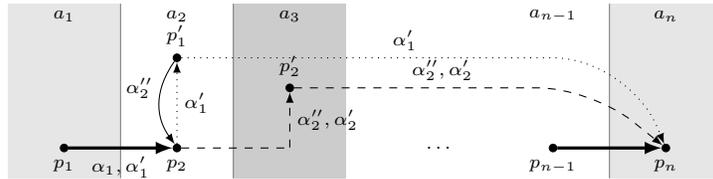
\begin{figure}[ht]
\begin{center}
	\begin{tikzpicture}[>=latex,point/.style={circle,draw=black,minimum size=1mm,inner sep=0pt},yscale=0.8]\scriptsize
            \fill[gray!20] (-0.75,1.5) rectangle +(1.5,2.9);
            \fill[gray!40] (2.25,1.5) rectangle +(1.5,2.9);
            \fill[gray!20] (7.25,1.5) rectangle +(1.5,2.9);
            \draw[gray] (0.75,1.5) -- ++(0,2.9);
            \draw[gray] (2.25,1.5) -- ++(0,2.9);
            \draw[gray] (7.25,1.5) -- ++(0,2.9);
            \node at (0,4.2)  {$a_1$};
            \node at (1.5,4.2)  {$a_2$};
            \node at (3,4.2)  {$a_3$};
            \node at (8,4.2)  {$a_n$};
            \node at (6.5,4.2)  {$a_{n-1}$};
            \node at (5,2) {\dots};
            \node (p1) at (0,2) [point,fill=black,label=below:{$p_1$}] {};
            \node (p2) at (1.5,2) [point,fill=black,label=below:{$p_2$}] {};
            \node (p2p) at (1.5,3.5) [point,fill=black,label=above:{$p_1'$}] {};
            \node (p3) at (3,2) [inner sep=0pt,minimum size=0pt] {};
            \node (p3p) at (3,3) [point,fill=black,label=above:{$p_2'$}] {};
            \node (pn) at (8,2) [point,fill=black,label=below:$p_n$] {};
			\draw[very thick,->] (p1) -- (p2) node[below,midway]{$\alpha_1,\alpha_1'$};
			\draw[dotted,->,rounded corners=10mm] (p2p) -- ++(6,0) node[above,midway]{$\alpha_1'$} -- (pn);
			\draw[dotted,->] (p2) -- (p2p) node[right,midway]{$\alpha_1'$};
			\draw[thin,->] (p2p) to[out=240,in=120] (p2);
            \node at (1,3) {$\alpha_2''$};
			\draw[dashed,->] (p2) -- (p3) -- (p3p) node[right,midway]{$\alpha_2'',\alpha_2'$};
			\draw[dashed,->,rounded corners=7mm] (p3p) -- ++(4,0) node[above,midway]{$\alpha_2'',\alpha_2'$} -- (pn);
            \node (pn1) at (6.5,2) [point,fill=black,label=below:$p_{n-1}$] {};
            \draw[very thick, ->] (pn1) -- (pn);
	\end{tikzpicture}
\end{center}
\caption{The constraint $\stacki(a_1,\dots, a_n)$ ensures the existence of a Jordan arc
$\alpha=\alpha_1\cdots\alpha_{n-1}$ which connects a point $p_1\in \ti a_1$ to a point $p_n\in\ti a_n$.}
\label{fig:stacki}
\end{figure}

We inductively define, for all $1 \leq i<n$, arcs
$\alpha_{i-1}, \alpha'_i$ and points $p_i, p'_i$ with the
following properties: $\alpha_{i-1} \subseteq \ti{(a_{i-1} +
  a_{i})}$ and runs from $p_{i-1}$ to $p_i$;  $\alpha'_i
\subseteq \ti{(a_i + \dots + a_n)}$; $\alpha_0 \cdots
\alpha_{i-1} \alpha'_{i}$ is a Jordan arc from $p_0$ to $p_n$; and
$p'_i \in \alpha'_{i} \cap \ti{a}_{i+1}$ with
$\alpha'_i[p_i,p'_i] \subseteq \ti{(a_i + a_{i+1})}$.  Suppose that,
for some $1 \leq i \leq n-2$, the requisite entities have
already been defined (notice that this is already the case for
$i=1$).  Since $a_{i+1}+\cdots+a_n$ is interior-connected, let $\alpha_{i+1}''
\subseteq \ti{(a_{i+1}+\cdots+a_n)}$ be a Jordan arc connecting $p_i'$
to $p_n$. Since we certainly have $(a_1 + \dots + a_i) \cdot
(a_{i+1} + \dots + a_n) = 0$, $\alpha''_{i+1}$ can intersect
$\alpha_0 \cdots\alpha_{i-1}\alpha_{i}'$ only in its final segment
$\alpha_i'$. Let $p_{i+1}$ be the first point of $\alpha_{i}'$ lying
on $\alpha_{i+1}''$; let $\alpha_i$ be the initial segment of
$\alpha_{i}'$ ending at $p_{i+1}$; and let $\alpha_{i+1}'$ be the
final segment of $\alpha_{i+1}''$ starting at $p_{i+1}$. By
construction, then, $\alpha_0 \cdots \alpha_i \alpha_{i+1}'$ is a
Jordan arc from $p_0$ to $p_n$, and $\alpha_i \subseteq
\ti{(a_i + a_{i+1})}$.  Moreover, since $\neg C(a_{i+1}, (a_{i+3}
+ \cdots + a_{n}))$, $\alpha'_{i+1}$ must contain a point $p'_{i+1} \in
\ti{a}_{i+2}$ such that $\alpha'_{i+1}[p_{i+1},p'_{i+1}] \subseteq
\ti{(a_{i+1} + a_{i+2})}$. Continuing up to the value $i = n-1$, we
have defined $\alpha_1, \dots, \alpha_{n-2}$, $\alpha'_{n-1}$
and $p_1, \dots,
p_{n-1}$.  It remains only to define $\alpha_{n-1}$; for this we
simply set $\alpha_{n-1} =\alpha_{n-1}'$.

For all $1 < i < n$, we have $p_i\in \alpha_{i-1}\cap\alpha_i$,
whence, by Lemma~\ref{lemma:rc-intersect}, $p_i\in \ti{a}_i$. It also
follows that the $\alpha_i$ are non-degenerate.
\end{proof}

It should be noted that $\stacki(r_1,\dots, r_n)$ is a not a
\cBci-formula, as it contains (negative) occurrences of the contact
predicate $C$. It turns out, however, that we can eliminate them. To
this end, consider the \cBci-formula
\begin{multline*}
\noncontacti(r_1,\dots,r_5) \ \ = \ \
\bigwedge_{i=1}^5 \bigl(\ic(r_i) \land (r_i \neq 0)\bigr) \ \ \land \ \ \bigwedge_{1 \leq i < j \leq 5} (r_i \cdot r_j = 0) \ \ \land \\
\bigwedge_{j=3}^5 \ic(r_1 + r_j) \ \ \land \ \ \bigwedge_{2 \leq i< j \leq 5}\hspace*{-0.5em} \ic(r_i + r_j).
\end{multline*}	
This formula is similar to formula~\eqref{eq:2vs3} encoding the
non-planar graph $K_5$ (hence the name); however, there is no
requirement that $r_1 + r_2$ is interior-connected.
\begin{lemma}\label{lma:Cci2BciStar}
\textup{(}i\textup{)} For each tuple $(a_1,\dots,a_5)$ of elements of
$\RC(\R^2)$ satisfying $\noncontacti(r_1,\dots,r_5)$, we have $\neg
C(a_1,a_2)$.  \textup{(}ii\textup{)} If regions $b_1$ and $b_2$ can be
separated by a Jordan curve then there exist polygons
$(a_1,\dots,a_5)$ satisfying $\noncontacti(r_1,\dots,r_5)$ such that
$b_1 \leq a_1$ and $b_2\leq a_2$.
\end{lemma}
\begin{proof}
({\em i}) For all $i$ ($1 \leq i \leq 5$), pick a point $p_i \in
  \ti{a}_i$. Then, for all $j$ ($3 \leq j \leq 5$) let $\gamma_{1,j}$ be an
  arc from $p_1$ to $p_j$ lying in $\ti{(a_1 + a_j)}$, and, for all $i$,
  $j$ ($2 \leq i < j \leq 5$), let $\gamma_{i,j}$ be an arc from $p_i$
  to $p_j$ lying in $\ti{(a_i + a_j)}$. It is routine to show that the
  various $\gamma_{i,j}$ can be chosen so that they intersect only at
  their endpoints. Thus, $\Gamma = \gamma_{3,4} \gamma_{4,5}
  \gamma_{3,5}^{-1}$ forms a Jordan curve in $\ti{(a_3 + a_4 +a_5)}$, and
  the arcs $\gamma_{2,3}$, $\gamma_{2,4}$ and $\gamma_{2,5}$ join
  $\Gamma$ to the point $p_2$ lying in one of its residual domains.
  (Fig.~\ref{fig:SepBci} illustrates the situation where $p_2$ lies in
  the bounded residual domain.)  Since $a_1$ and $a_2$ are (interior-)
  connected and cannot intersect $\Gamma$, they each lie in one of its
  residual domains. It suffices to show that $a_1$ and
  $a_2$ lie in {\em different} residual domains. To see this, observe
  that the arcs $\gamma_{2,3}$, $\gamma_{2,4}$ and $\gamma_{2,5}$
  divide the residual domain of $\Gamma$ containing $p_2$ into three
  regions, bounded by arcs lying in $\ti{(a_2 + a_3+ a_4)}$, $\ti{(a_2
    + a_4+ a_5)}$ and $\ti{(a_2 + a_5+ a_3)}$, respectively.  But if
  $a_1$ and $a_2$ lie on the same side of $\Gamma$, then $p_1$ lies in
  one of these regions, contradicting the existence of arcs
  $\gamma_{1,j} \subseteq \ti{(a_1 + a_j)}$ connecting $p_1$ to $p_j$;
  for $j = 3,4,5$.
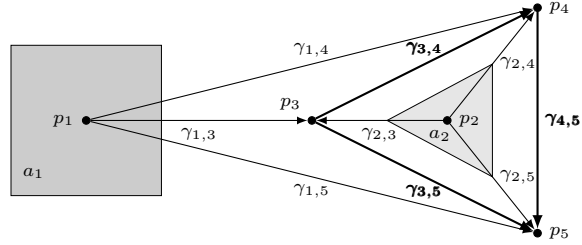
\begin{figure}[ht]
	\begin{center}
		\begin{tikzpicture}[>=latex,point/.style={circle,draw=black,minimum size=1mm,inner sep=0pt}]\scriptsize
            \filldraw[ultra thin,fill=gray!40] (-4,-1) rectangle +(2,2);
            \node at (-3.7,-0.7) {$a_1$};
            \filldraw[ultra thin,fill=gray!20] (1,0) -- (2.4,0.75) -- (2.4,-0.75) -- cycle;
            \node at (1.7,-0.2) {$a_2$};
            \node (p3) at (0,0) [point,fill=black,label=above left:{$p_3$}] {};
            \node (p4) at (3,1.5) [point,fill=black,label=right:{$p_4$}] {};
            \node (p5) at (3,-1.5) [point,fill=black,label=right:{$p_5$}] {};
            \draw[thick,->] (p3) -- (p4) node[midway,above] {$\pmb{\gamma_{3,4}}$};
            \draw[thick,->] (p4) -- (p5) node[midway,right] {$\pmb{\gamma_{4,5}}$};
            \draw[thick,<-] (p5) -- (p3) node[midway,below] {$\pmb{\gamma_{3,5}}$};
            \node (p1) at (-3,0) [point,fill=black,label=left:{$p_1$}] {};	
            \draw[<-] (p3) -- (p1) node[midway,below] {$\gamma_{1,3}$};
            \draw[<-] (p4) -- (p1) node[midway,above] {$\gamma_{1,4}$};
            \draw[<-] (p5) -- (p1) node[midway,below] {$\gamma_{1,5}$};
            \node (p2) at (1.8,0) [point,fill=black,label=right:{$p_2$}] {};	
            \draw[<-] (p3) -- (p2) node[midway,below] {$\gamma_{2,3}$};
            \draw[<-] (p4) -- (p2) node[midway,right] {$\gamma_{2,4}$};
            \draw[<-] (p5) -- (p2) node[midway,right] {$\gamma_{2,5}$};
	\end{tikzpicture}
	\end{center}
	\caption{The Jordan curve
    $\Gamma=\gamma_{3,4}\gamma_{4,5}\gamma_{3,5}^{-1}$ (thick lines) separating $a_1$ from $a_2$.}
	\label{fig:SepBci}
\end{figure}

({\em ii}) Let $\Gamma$ be a Jordan curve separating $b_1$ and
$b_2$. We may assume that $\Gamma$ is piecewise-linear. Now thicken $\Gamma$ to form an annular element of
$\RCP(\R^2)$, still disjoint from $b_1$ and $b_2$, and divide it into
the three interior-connected and non-overlapping polygons
$a_3,a_4,a_5$.  Choose $a_1$ and $a_2$ to be the components of the
complement of $a_3+a_4+a_5$ containing $b_1$ and $b_2$, respectively.
\end{proof}

We remark that
Lemma~\ref{lma:Cci2BciStar} guarantees that $(a_1,\dots,a_5)$ are
regular closed {\em polygons}. This fact will be important in
\SECT\ref{subsec:undecPlane}, where we prove the undecidability of
$\Sat(\cBci,\RCP(\R^2))$; for the main result of this section,
however, we require only that $(a_1,\dots,a_5)$ are regular closed
sets in $\R^2$:
\begin{theorem}\label{theo:inftyBci}
There is a \cBci-formula satisfiable over $\RC(\R^2)$, but only by
tuples featuring sets with infinitely many components.
\end{theorem}
\begin{proof}
We first write a \cBCci-formula, $\phi^*_\infty$ with the required
properties, and then show that all occurrences of $C$ in it can be
eliminated.  Note that $\phi^*_\infty$ is not the same as the formula
$\ti{\phi}_\infty$ constructed for the proof of
Corollary~\ref{cor:inftyCci}.  As with the proof of
Theorem~\ref{theo:cBCcn}, we equivocate between variables and the
regions to which they are assigned in some putative interpretation
over $\RC(\R^2)$. If $k$ is an integer,
$\md{k}$ indicates the value of $k$ modulo 2.

Let $s_0,\dots,s_3$, $a$, $b$, $a_{i,j}$ and $b_{i,j}$ be variables,
for $0 \leq i <2$, $1 \leq j \leq 3$. The constraints
\begin{align}
	&\frameFlai(s_0,s_1,b,s_2,a,s_3),\label{eq:BciInf1}\\
	&\stacki(s_0,b_{1,1},b_{1,2},b_{1,3},b),\label{eq:BciInf2}\\ %
	& \bigwedge_{i = 0}^1
           \stacki(b_{i,2},a_{i,1},a_{i,2},a_{i,3},a),\label{eq:BciInf3}\\ %
	& \bigwedge_{i = 0}^1
           \stacki(a_{\md{i-1},2},b_{i,1},b_{i,2},b_{i,3},b)\label{eq:BciInf4} %
\end{align}
are evidently satisfied by the arrangement of Fig.~\ref{fig:InfBci}.
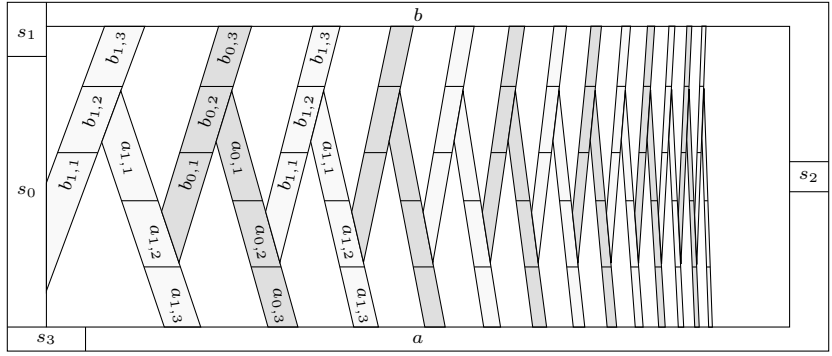
\begin{figure}[ht]
\begin{center}
\begin{tikzpicture}[xscale=1.3,yscale=0.8,
		c0/.style={fill=gray!5},
		c1/.style={fill=gray!5},
		c2/.style={fill=gray!25},
		c3/.style={fill=gray!25}
	]
		\scriptsize{
		\newcounter{md}\newcounter{mdf} \newcounter {mdd}
		\coordinate (H) at (0,5);
		\coordinate (P) at (0,0);
		\coordinate (Q) at (H);
		\coordinate (start) at (0,1.5);\coordinate (start') at ($(start)+(1,0)$);
		\coordinate (first) at (0,2.9);\coordinate (first') at ($(first)+(1,0)$);
		\coordinate (second) at (0,4);\coordinate (second') at ($(second)+(1,0)$);
		\coordinate (end) at (H);\coordinate (end') at ($(end)+(1,0)$);
		
		\coordinate (width) at (.8,0);
		
		\foreach \d in {0,...,25}
		{
			\pgfmathparse{.91^\d}			\let\factor=\pgfmathresult 
			\pgfmathparse{2*mod(\d,2)-1}	\let\prt=\pgfmathresult
			\pgfmathsetcounter{mdf}{mod(\d,4)}
			\pgfmathsetcounter{md}{mod(\d,2)}
			\pgfmathsetcounter{mdd}{mod((\d+2)/2,2)}
			
			\coordinate (Pm) at ($(intersection of P--Q and start--start')$);
			\coordinate (Qm) at ($(end-|Pm)+\factor*(width)$);
			\foreach \a/\ind in {-90/0,90/} 
			{
				\coordinate (Pm') at ($(Pm)!{.2cm*\factor}!{\prt*\a}:(Qm)$);
				\coordinate (Qm') at ($(Qm)!{.2cm*\factor}!{-\prt*\a}:(Pm)$);
				\coordinate (P\ind) at ($(intersection of Pm'--Qm' and P--Q)$);
				\coordinate (Q\ind) at ($(intersection of Pm'--Qm' and end--end')$);
			}
			\draw[c\themdf] (P0)--(Q0)--(Q)--(P)--cycle;			
			
			\draw ($(intersection of first--first' and P0--Q0)$)
				--($(intersection of first--first' and P--Q)$);
			\draw ($(intersection of second--second' and P0--Q0)$)
				--($(intersection of second--second' and P--Q)$);			
			
			\coordinate (start) at ($(H)-(start)$);\coordinate (start') at ($(start)+(1,0)$);
			\coordinate (first) at ($(H)-(first)$);\coordinate (first') at ($(first)+(1,0)$);			
			\coordinate (second) at ($(H)-(second)$);\coordinate (second') at ($(second)+(1,0)$);
			\coordinate (end) at ($(H)-(end)$);\coordinate (end') at ($(end)+(1,0)$);
			
			\ifnum \d<6
				\ifnum \themd=0
					\node[rotate=81] at ($(Pm)!.3!(Qm)$) {$b_{\themdd,1}$};
					\node[rotate=81] at ($(Pm)!.6!(Qm)$) {$b_{\themdd,2}$};
					\node[rotate=81] at ($(Pm)!.9!(Qm)$) {$b_{\themdd,3}$};
				\fi
				\ifnum \themd=1
					\node[rotate=-81] at ($(Pm)!.2!(Qm)$) {$a_{\themdd,1}$};
					\node[rotate=-81] at ($(Pm)!.6!(Qm)$) {$a_{\themdd,2}$};
					\node[rotate=-81] at ($(Pm)!.9!(Qm)$) {$a_{\themdd,3}$};
				\fi
			\fi
		}
		}
		\draw (-0.4,-0.4) rectangle ($(8,.4)+(H)$); 	\draw (0,0) rectangle ($(7.6,0)+(H)$);
		\draw (0,0)--++(-.4,0);
		\draw (H)--++(0,.4);
		\draw ($(H)!.1!(0,0)$)--++(-.4,0);
		\draw ($(.4,0)$)--++(0,-.4);
		\draw ($(H)!.45!(0,0)+(7.6,0)$) --++(.4,0);
		\draw ($(H)!.55!(0,0)+(7.6,0)$) --++(.4,0);
		\node at ($(H)!.55!(0,0)+(-.2,0)$){$s_0$};
		\node at ($(H)!.025!(0,0)+(-.2,0)$){$s_1$};
		\node at (0,-.2){$s_3$};
		\node at ($(H)+(3.8,.2)$){$b$};
		\node at ($(H)!.5!(0,0)+(7.8,0)$){$s_2$};
		\node at (3.8,-.2){$a$};
		
	\end{tikzpicture}
\end{center}
\caption{A tuple satisfying
  {\eqref{eq:BciInf1}--\eqref{eq:BciInf4}}: the pattern of components of the
  $a_{i,j}$ and $b_{i,j}$ repeats forever.}
\label{fig:InfBci}
\end{figure}%
Let $\phi^*_\infty$ be the conjunction of~\eqref{eq:BciInf1}--\eqref{eq:BciInf4}
as well as formulas
\begin{align}
	(r\cdot r'=0), \qquad  \text{ for distinct variables $r$ and $r'$.}\label{eq:BciInf5}
\end{align}
Note that the regions $a_{i,j}$ and $b_{i,j}$ have infinitely many
components. We will show that this is true for every satisfying tuple
of $\phi^*_\infty$.

By \eqref{eq:BciInf1} and Lemma~\ref{lma:FrameLemmaInt}, there is a
Jordan curve $\sigma\lambda_0\mu_0^{-1}$ whose segments are Jordan
arcs lying in the respective sets $\ti{(s_3 + s_0 + s_1)}$, $\ti{(s_2 +
a+s_3)}$ and $\ti{(s_1 + b+s_2)}$; see
Fig.~\ref{subfig:BciFrame}.  Note that all points in $s_0$ that are on
$\sigma\lambda_0\mu^{-1}_0$ are on $\sigma$.  Let $o'$ be the common
point of $\mu_0$ and $\lambda_0$ and $\tilde q_{1,1}\in\sigma
\cap\ti{s}_0$.

A word is required concerning the generality of this and other
diagrams in this section.  The reader is to imagine the figure drawn
on a \emph{spherical} canvas, of which the sheet of paper or computer
screen in front of him is simply a small part.  This sphere represents
the plane with a `point' at infinity, under the usual stereographic
projection. We do not say where this point at infinity is, other than
that it never lies on a drawn arc.  In this way, a diagram in which
the spherical canvas is divided into $n$ cells represents $n$
different configurations in the plane---one for each of the cells in
which the point at infinity may be located. For example,
Fig.~\ref{subfig:BciFrame} represents two topologically distinct
configurations in $\R^2$, and, as such, depicts the arcs $\sigma$,
$\lambda_0$ and $\mu_0$ and points $\tilde q_{1,1}$, $o'$ in full
generality.  All diagrams in this proof are to be interpreted in this
way.  We stress that our `spherical diagrams' are simply a convenient
device for using one drawing to represent several possible
configurations in the Euclidean plane: in particular, we are
interested only in the satisfiability of \cBci-formulas over
$\RC(\R^2)$, not over the regular closed algebra of any other space!

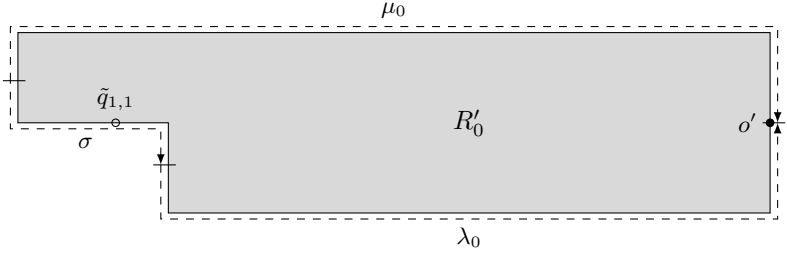
\begin{figure}[ht]
\begin{center}
	\begin{tikzpicture}[>=latex,point/.style={circle,draw=black,minimum size=1mm,inner sep=0pt},yscale=0.8]\small
        \filldraw[fill=gray!30] (-2,1.5) -- ++(2,0) -- ++(0,-1.5) -- ++(8,0) -- ++(0,3) -- ++(-10,0) -- cycle;
        \draw[thin] (7.9,1.5) -- +(0.3,0);
        \draw[thin] (-2.2,2.2) -- +(0.3,0);
        \draw[dashed,->] (-2.1, 2.2) -- (-2.1,3.1) -- (8.1,3.1) node[midway,above]{$\mu_0$} -- (8.1,1.5);
        \draw[thin] (-0.2,0.8) -- +(0.3,0);
        \draw[dashed,->] (-0.1, 0.8) -- (-0.1,-0.1) -- (8.1,-0.1) node[midway,below]{$\lambda_0$} -- (8.1,1.5);
        \node (q11t) at (-0.7,1.5) [point,label=above:{$\tilde q_{1,1}$}] {};
        \draw[dashed,<-] (-0.1, 0.8) -- (-0.1,1.4) -- (-2.1,1.4) node[midway,below]{$\sigma$} -- (-2.1,2.2);
        \node at (8,1.5) [point, fill=black, label=left:{$o'$}] {};
\node at (4,1.5) {\normalsize $R'_0$};

	\end{tikzpicture}
\end{center}
\caption{The arcs $\sigma$, $\mu_0$ and $\lambda_0$.}\label{subfig:BciFrame}
\end{figure}

Let $\tilde q_{1,3}\in \mu_0\cap\ti{b}$. By \eqref{eq:BciInf2} and
Lemma~\ref{lma:StackLemmai}, we can connect $\tilde q_{1,1}$ to
$\tilde q_{1,3}$ by a Jordan arc
$\tilde\beta_{1,1}\beta_{1,2}\tilde\beta_{1,3}$ whose segments lie in
the respective sets $\ti{(s_0+b_{1,1})}$,
$\ti{(b_{1,1}+b_{1,2}+b_{1,3})}$ and $\ti{(b_{1,3} + b)}$; see
Fig.~\ref{subfig:BciBeta0}.  Let $q_{1,1}$ be the last point on
$\tilde\beta_{1,1}$ that is on $\sigma$ and let $\beta_{1,1}$ be the
final segment of $\tilde\beta_{1,1}$ starting at $q_{1,1}$; by
Lemma~\ref{lemma:rc-intersect}, $q_{1,1}\in \ti{s}_0$. Similarly, let
$q_{1,3}$ be the first point on $\tilde\beta_{1,3}$ that is on $\mu_0$
and let $\beta_{1,3}$ be the initial segment of $\tilde\beta_{1,3}$
ending at $q_{1,3}$; by Lemma~\ref{lemma:rc-intersect}, $q_{1,3}\in
\ti{b}$. Hence, the arc $\beta_{1,1}\beta_{1,2}\beta_{1,3}$ lies in
exactly one of the regions bounded by $\sigma\lambda_0\mu_0^{-1}$: for
reasons that will emerge in the course of the proof, we denote that
region $R'_0$. Now, $\beta_{1,1}\beta_{1,2}\beta_{1,3}$ divides $R'_0$
into two sub-regions: we denote the sub-region whose boundary is
disjoint from $a$ by $S_1$, and the other sub-region by $S_1'$.  Let
$\mu_1=\beta_{1,3}\mu_0[q_{1,3},o']\subseteq \ti{(b_{1,3} + s_1 + b +
  s_2)}$. The arc $\beta_{1,2}$ contains a point $\tilde p_{1,1}\in
\ti{b}_{1,2}$; moreover, all points of
$\beta_{1,1}\beta_{1,2}\beta_{1,3}$ in $\ti{b}_{1,2}$ lie on
$\beta_{1,2}$.

\begin{figure}[ht]
\begin{center}
	\begin{tikzpicture}[>=latex,point/.style={circle,draw=black,minimum size=1mm,inner sep=0pt},yscale=0.8]\small
        \filldraw[fill=gray!30] (0,0) rectangle +(8,3);
        \filldraw[fill=gray!60] (-2,1.5) rectangle +(4,1.5);
        \node (q11) at (0,1.5) [point,fill=black,label=above:{$q_{1,1}$}] {};
        \node (q11t) at (-0.7,1.5) [point,label=below:{$\tilde q_{1,1}$}] {};
		\draw[densely dotted,decorate,decoration=snake] (q11t) -- (q11);
        \node (q13) at (2,3) [point,fill=black,label=above:{$q_{1,3}$}] {};
        \node (q13t) at (1.3,3) [point,label=below:{$\tilde q_{1,3}$}] {};
		\draw[densely dotted,decorate,decoration=snake] (q13t) -- (q13);
        \node (p11t) at (1.3,1.5) [point,label=below:{$\tilde p_{1,1}$}] {};
        \node at (8,1.5) [point, fill=black, label=left:{$o'$}] {};
        \draw[thin] (0.7,1.3) -- +(0,0.4);
        \draw[thin] (1.8,2.2) -- +(0.4,0);
        \draw[dashed,->] (0.7, 1.6) -- (1.9,1.6) node[midway,above]{$\beta_{1,2}$} -- (1.9,2.2);
        \draw[dashed,->] (0, 1.4) -- (0.7,1.4) node[midway,below]{$\beta_{1,1}$};
        \draw[dashed,->] (2.1, 2.2) -- (2.1,2.9) -- (7.9,2.9) node[midway,below]{$\mu_1$} -- (7.9,1.5);
        \draw[thin] (7.8,1.5) -- +(0.4,0);
        \draw[thin] (-2.2,2.2) -- +(0.3,0);
        \draw[dashed,->] (-2.1, 2.2) -- (-2.1,3.1) -- (8.1,3.1) node[midway,above]{$\mu_0$} -- (8.1,1.5);
        \node at (0,2.4) {\normalsize $S_1$};
        \node at (4,1.5) {\normalsize $S_1'$};
%
%
	\end{tikzpicture}
\end{center}
\caption{The regions $S_1$ and $S_1'$.}\label{subfig:BciBeta0}
\end{figure}
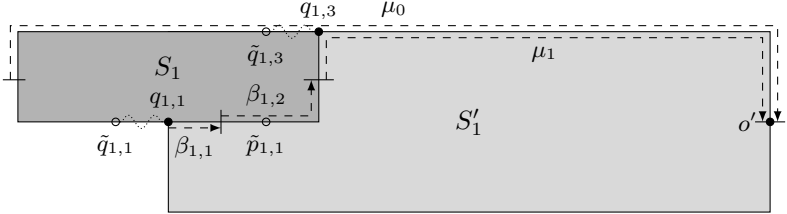

We will now construct a cross-cut $\alpha_{1,1}\alpha_{1,2}
\alpha_{1,3}$ in $S_1'$. Let $\tilde p_{1,3}$ be a point in
$\lambda_0\cap\ti{a}$. By \eqref{eq:BciInf3} and
Lemma~\ref{lma:StackLemmai}, we can connect $\tilde p_{1,1}$ to
$\tilde p_{1,3}$ by a Jordan arc
$\tilde\alpha_{1,1}\alpha_{1,2}\tilde\alpha_{1,3}$ whose segments lie
in the respective sets $\ti{(b_{1,2}+a_{1,1})}$,
$\ti{(a_{1,1}+a_{1,2}+a_{1,3})}$ and $\ti{(a_{1,3} + a)}$; see
Fig.~\ref{subfig:BciAlpha0}.  Let $p_{1,1}$ be the last point on
$\tilde\alpha_{1,1}$ that is on $\beta_{1,2}$ and let $\alpha_{1,1}$
be the final segment of $\tilde\alpha_{1,1}$ starting at $p_{1,1}$; by
Lemma~\ref{lemma:rc-intersect}, $p_{1,1}\in \ti{b}_{1,2}$.  Similarly,
let $p_{1,3}$ be the first point on $\tilde\alpha_{1,3}$ that is on
$\lambda_0$ and let $\alpha_{1,3}$ be the initial segment of
$\tilde\alpha_{1,3}$ ending at $p_{1,3}$; by
Lemma~\ref{lemma:rc-intersect}, $p_{1,3}\in \ti{a}$.  Since
$\alpha_{1,1}\alpha_{1,2}\alpha_{1,3}$ does not intersect the
boundaries of $S_1$ and $S_1'$ except at its endpoints, it is a
cross-cut in one of these regions. Moreover, that region has to be
$S_1'$ since the boundary of $S_1$ is disjoint from $a$. So,
$\alpha_{1,1}\alpha_{1,2}\alpha_{1,3}$ divides $S_1'$ into two
sub-regions: we denote the sub-region whose boundary is disjoint from
$b$ by $R_1$, and the other sub-region by $R'_1$.  Let
$\lambda_1=\alpha_{1,3}\lambda_0[p_{1,3},o']\subseteq \ti{(a_{1,3} +
  s_3 + a + s_2)}$.  The arc $\alpha_{1,2}$ contains a point $\tilde
q_{2,1}\in \ti{a}_{1,2}$; moreover, all points of
$\alpha_{1,1}\alpha_{1,2}\alpha_{1,3}$ in $\ti{a}_{1,2}$ lie on
$\alpha_{1,2}$.

\begin{figure}[ht]
\begin{center}
	\begin{tikzpicture}[>=latex,point/.style={circle,draw=black,minimum size=1mm,inner sep=0pt},yscale=0.8]\small
        \filldraw[fill=gray!30] (0,0) rectangle +(8,3);
        \filldraw[fill=gray!60] (-2,1.5) rectangle +(4,1.5);
        \filldraw[fill=gray!5] (0,0) rectangle +(4,1.5);
        \node (p11) at (2,1.5) [point,fill=black,label=below:{$p_{1,1}$}] {};
        \node (p11t) at (1.3,1.5) [point,label=above:{$\tilde p_{1,1}$}] {};
		\draw[densely dotted,decorate,decoration=snake] (p11t) -- (p11);
        \node (p13) at (4,0) [point,fill=black,label=below:{$p_{1,3}$}] {};
        \node (p13t) at (3.3,0) [point,label=above:{$\tilde p_{1,3}$}] {};
		\draw[densely dotted,decorate,decoration=snake] (p13t) -- (p13);
        \node (q21t) at (3.3,1.5) [point,label=above:{$\tilde q_{2,1}$}] {};
        \node at (8,1.5) [point, fill=black, label=left:{$o'$}] {};
        \draw[thin] (2.7,1.7) -- +(0,-0.4);
        \draw[thin] (3.8,0.8) -- +(0.4,0);
        \draw[dashed,->] (2.7, 1.4) -- (3.9,1.4) node[midway,below]{$\alpha_{1,2}$} -- (3.9,0.8);
        \draw[dashed,->] (2, 1.6) -- (2.7,1.6) node[midway,above]{$\alpha_{1,1}$};
        \draw[dashed,->] (4.1, 0.8) -- (4.1,0.1) -- (7.9,0.1) node[midway,above]{$\lambda_1$} -- (7.9,1.5);
        \draw[thin] (7.8,1.5) -- +(0.4,0);
        \draw[thin] (-0.2,0.8) -- +(0.3,0);
        \draw[dashed,->] (-0.1, 0.8) -- (-0.1,-0.1) -- (8.1,-0.1) node[near end,below]{$\lambda_0$} -- (8.1,1.5);
        \node at (0,2.4) {\normalsize $S_1$};
        \node at (2,0.6) {\normalsize $R_1$};
        \node at (6,1.5) {\normalsize $R'_1$};
	\end{tikzpicture}
\end{center}
\caption{The regions $S_1$, $R_1$ and $R'_1$.}\label{subfig:BciAlpha0}
\end{figure}
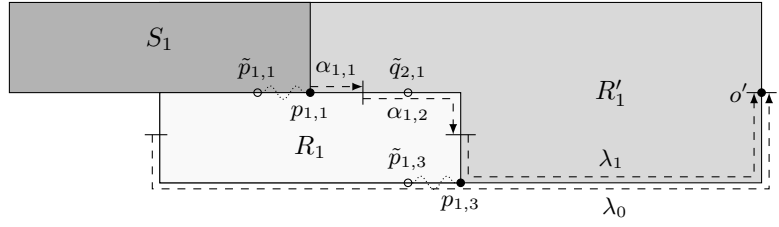

We can now forget about the region $S_1$, and start constructing a
cross-cut $\beta_{2,1}\beta_{2,2}\beta_{2,3}$ in $R'_1$. As before,
let $\tilde q_{2,3}\in\mu_1\cap\ti{b}$.  Then there is a Jordan arc
$\tilde\beta_{2,1}\beta_{2,2}\tilde\beta_{2,3}$ connecting $\tilde q_{2,1}$ to $\tilde
q_{2,3}$ such that its segments are contained in the respective sets
$\ti{(a_{1,2} + b_{0,1})}$, $\ti{(b_{0,1}+b_{0,2}+b_{0,3})}$ and
$\ti{(b_{0,3} + b)}$. As before, we choose $\beta_{2,1}\subseteq
\tilde\beta_{2,1}$ and $\beta_{2,3}\subseteq \tilde\beta_{2,3}$ so
that the Jordan arc $\beta_{2,1}\beta_{2,2}\beta_{2,3}$ without its
endpoints is disjoint from the boundaries of $R_1$ and $R'_1$. Hence
$\beta_{2,1}\beta_{2,2}\beta_{2,3}$ has to be a cross-cut in $R_1$ or
$R'_1$, and since the boundary of $R_1$ is disjoint from $b$ it has to
be a cross-cut in $R'_1$. So, $\beta_{2,1}\beta_{2,2}\beta_{2,3}$
separates $R'_1$ into two regions $S_2$ and $S_2'$ so that the
boundary of $S_2$ is disjoint from $a$.  Let
$\mu_2=\beta_{2,3}\mu_1[q_{2,3},o']\subseteq \ti{(b_{0,3} + b_{1,3} +
  s_1 + b+ s_2)}$.  Now, we can ignore the region $R_1$, and reasoning
as before we can construct a cross-cut
$\alpha_{2,1}\alpha_{2,2}\alpha_{2,3}$ in $S_2'$ dividing it into two
sub-regions $R_2$ and $R'_2$.

Evidently, this process continues forever: $R'_{i-1}$ is divided into
$S_i$ and $S'_i$ and $S'_i$ is divided into $R_i$ and $R'_i$. Now, the
boundary of $S_i$ contains the arc $\beta_{i,2}$, whence the interior
of $S_i$ contains points of $b_{\md{i},2}$. On the other hand, $S_i$
certainly lies outside $S'_{i+1}$; moreover, 
$\delta S'_{i+1}$ is a subset of $\alpha_{i,2} \cup \beta_{i+1,1} \cup \beta_{i+1,2}
  \cup \mu_{i+1} \cup \lambda_i$, whence
\begin{multline*}
\delta S'_{i+1}  \subseteq
\ti{(a_{\md{i},1} + a_{\md{i},2} + a_{\md{i},3})} \cup
   \ti{(a_{\md{i},2} + b_{\md{i+1},1})} \cup \\
 \ti{(b_{\md{i+1},1} + b_{\md{i+1},2} + b_{\md{i+1},3})} \cup{}\\
   \ti{(b_{0,3} + b_{1,3} + b + s_1 + s_2)} \cup 
      \ti{(a_{0,3} + a_{1,3} + a + s_2 + s_3)}.
\end{multline*}
Hence $\delta S'_{i+1}$ contains no points of $b_{\md{i},2}$. Yet
$S'_{i+1}$ evidently includes all the regions $S_{i +2k}$ for all $k
\geq 1$, each of which contains points of $b_{\md{i},2}$. It follows
that $b_{\md{i},2}$ has infinitely many components.

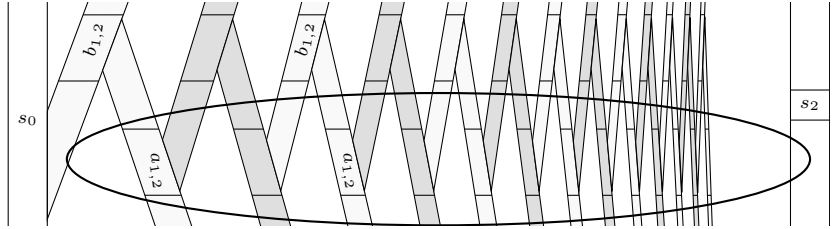
\begin{figure}[ht]
\begin{center}
\begin{tikzpicture}[xscale=1.3,yscale=0.8,
		c0/.style={fill=gray!5},
		c1/.style={fill=gray!5},
		c2/.style={fill=gray!25},
		c3/.style={fill=gray!25}
	]
		\clip (-.5,.5) rectangle (8,4.2);
		\scriptsize{
		\coordinate (H) at (0,5);
		\coordinate (P) at (0,0);
		\coordinate (Q) at (H);
		\coordinate (start) at (0,1.5);\coordinate (start') at ($(start)+(1,0)$);
		\coordinate (first) at (0,2.9);\coordinate (first') at ($(first)+(1,0)$);
		\coordinate (second) at (0,4);\coordinate (second') at ($(second)+(1,0)$);
		\coordinate (end) at (H);\coordinate (end') at ($(end)+(1,0)$);
		
		\coordinate (width) at (.8,0);
		
		\foreach \d in {0,...,25}
		{
			\pgfmathparse{.91^\d}			\let\factor=\pgfmathresult 
			\pgfmathparse{2*mod(\d,2)-1}	\let\prt=\pgfmathresult
			\pgfmathsetcounter{mdf}{mod(\d,4)}
			\pgfmathsetcounter{md}{mod(\d,2)}
			\pgfmathsetcounter{mdd}{mod((\d+2)/2,2)}
			
			\coordinate (Pm) at ($(intersection of P--Q and start--start')$);
			\coordinate (Qm) at ($(end-|Pm)+\factor*(width)$);
			\foreach \a/\ind in {-90/0,90/} 
			{
				\coordinate (Pm') at ($(Pm)!{.2cm*\factor}!{\prt*\a}:(Qm)$);
				\coordinate (Qm') at ($(Qm)!{.2cm*\factor}!{-\prt*\a}:(Pm)$);
				\coordinate (P\ind) at ($(intersection of Pm'--Qm' and P--Q)$);
				\coordinate (Q\ind) at ($(intersection of Pm'--Qm' and end--end')$);
			}
			\draw[c\themdf] (P0)--(Q0)--(Q)--(P)--cycle;			
			
			\draw ($(intersection of first--first' and P0--Q0)$)
				--($(intersection of first--first' and P--Q)$);
			\draw ($(intersection of second--second' and P0--Q0)$)
				--($(intersection of second--second' and P--Q)$);			
			
			\coordinate (start) at ($(H)-(start)$);\coordinate (start') at ($(start)+(1,0)$);
			\coordinate (first) at ($(H)-(first)$);\coordinate (first') at ($(first)+(1,0)$);			
			\coordinate (second) at ($(H)-(second)$);\coordinate (second') at ($(second)+(1,0)$);
			\coordinate (end) at ($(H)-(end)$);\coordinate (end') at ($(end)+(1,0)$);
			
			\ifnum \d<6
				\ifnum \themdd=1
					\ifnum \themd=0
						\node[rotate=81] at ($(Pm)!.6!(Qm)$) {$b_{\themdd,2}$};
					\fi
					\ifnum \themd=1
						\node[rotate=-81] at ($(Pm)!.6!(Qm)$) {$a_{\themdd,2}$};
					\fi
				\fi
			\fi
		}
		}
		\draw[thick] (4,1.6) ellipse (3.8 and 1.1);
		\draw (-0.4,-0.4) rectangle ($(8,.4)+(H)$); 	\draw (0,0) rectangle ($(7.6,0)+(H)$);
		\draw (0,0)--++(-.4,0);
		\draw (H)--++(0,.4);
		\draw ($(H)!.1!(0,0)$)--++(-.4,0);
		\draw ($(.4,0)$)--++(0,-.4);
		\draw ($(H)!.45!(0,0)+(7.6,0)$) --++(.4,0);
		\draw ($(H)!.55!(0,0)+(7.6,0)$) --++(.4,0);
		\node at ($(H)!.55!(0,0)+(-.2,0)$){$s_0$};
		\node at ($(H)!.025!(0,0)+(-.2,0)$){$s'$};
		\node at (0,-.2){$a'$};
		\node at ($(H)+(3.8,.2)$){$b$};
		\node at ($(H)!.5!(0,0)+(7.8,0)$){$s_2$};
		\node at (3.8,-.2){$a$};
	\end{tikzpicture}
\end{center}
\caption {Separating $a_{1,2}$ from $b_{1,2}$ by a Jordan curve.}
\label{fig:BciInftySep}
\end{figure}

So far we know that the \cBCci-formula $\phi^*_{\infty}$ forces
infinitely many components.  Now we replace every conjunct in
$\phi^*_{\infty}$ of the form $\lnot C(s_1,s_2)$ by $\noncontacti(\vec{r}) \land (s_1 \leq r_1) \land (s_2 \leq r_2)$,
where $\vec{r}$ is a vector of fresh variables. By Lemma~\ref{lma:Cci2BciStar}~(\emph{i}), the resulting formula
entails $\phi^*_{\infty}$. Conversely, to show that the formula is satisfiable, we apply Lemma~\ref{lma:Cci2BciStar}~(\emph{ii}): it suffices
to separate every pair of disjoint regions in
Fig.~\ref{fig:InfBci} by a Jordan curve. Such a Jordan curve is shown in
Fig.~\ref{fig:BciInftySep} for $b_{1,2}$ and $a_{1,2}$. Other pairs of disjoint regions
are treated analogously.
\end{proof}

\subsection{Undecidability in the plane}
\label{subsec:undecPlane}
We now return to the question of decidability.  We know from
\SECT\ref{sec:undecidability} that $\Sat(\cL,\RCP(\R^2))$ is
undecidable, where $\cL$ is any of the languages \cBc, \cBCc{} or
\cBCci. We proceed to establish the undecidability of the problems
$\Sat(\cL,\RC(\R^2))$, where $\cL$ is any of the languages \cBc,
\cBci, \cBCc{} or \cBCci, and also of the problem
$\Sat(\cBci,\RCP(\R^2))$. Most of the techniques required have been
rehearsed in the proof of Theorem~\ref{theo:inftyBci}.  However, we
face a new difficulty. In the language \cBci, we can say that the
\emph{interior} of a region (rather than merely the region itself) is
connected.  Since, for open sets, connectedness implies
arc-connectedness, we were able, in the proof of
Theorem~\ref{theo:inftyBci}, to write formulas enforcing various
arrangements of Jordan arcs in the plane. When dealing with \cBc{} and
\cBCc, however, we can speak merely of the connectedness of a region
(rather than of its interior), which, for elements of $\RC(\R^2)$ does
not imply arc-connectedness; this complicates the business of
enforcing the requisite arrangements of Jordan arcs.

To overcome this difficulty, we employ the
technique of `wrapping' a region inside two bigger ones.  If $a$ and
$b$ are regions such that $\neg C(a, -b)$, we write $a \ll b$
(pronounced: $a$ \emph{is right inside} $b$).  Let us say that a
\emph{3-region} is a triple $\tseq{a} =
(a,\intermediate{a},\inner{a})$ of elements of $\RC(\R^2)$ such that
$0 \neq \inner{a} \ll \intermediate{a} \ll a$.  It helps to think of
$\tseq{a} = (a,\intermediate{a},\inner{a})$ as consisting of a kernel,
$\inner{a}$, encased in two protective layers: an inner shell,
$\intermediate{a}$ and an outer shell, $a$. As a simple example,
consider the sequence of 3-regions $\tseq{a}_1, \tseq{a}_2,
\tseq{a}_3$ depicted in Fig.~\ref{fig:stack}, where the kernels form a
sequence of externally touching polygons.
\begin{figure}[ht]
\begin{center}
\begin{tikzpicture}[	s0/.style={thick,fill=gray!40,fill opacity=0.5},			
						s1/.style={fill=gray!20,fill opacity=0.3},			
						s2/.style={thin,fill=white,fill opacity=0}] 		
	
	\newcount\mod
	\newcount\jj
	
	\coordinate (H0) at (2,0); 
	\coordinate (V0) at (0,0.8); 
	\coordinate (H1) at (2,0); 
	\coordinate (V1) at (0,.4); 
	
	\coordinate (HW) at (.4,0); 
	\coordinate (VW) at (0,.2); 
	\coordinate (P) at (0,0);
	
	\foreach \j in {2,1,0}
	{
	   \foreach \i in {1,2,3}
		{			
		    \pgfmathsetcount{\mod}{mod(\i,2)}
			\coordinate (H) at ($\j*(0.2,0)$); 
			\coordinate (V) at ($\j*(0,0.33)$); 
			
			\filldraw[s\j] ($(P)-(H)-(V)-.5*(V\the\mod)+(VW)$)
				--++($(V\the\mod)+2*(V)-2*(VW)$)
				--++($(HW)+(VW)$)
				--++($(H\the\mod)+2*(H)-2*(HW)$)
				--++($(HW)-(VW)$)
				--++($-1*(V\the\mod)-2*(V)+2*(VW)$)
				--++($-1*(HW)-(VW)$)
				--++($-1*(H\the\mod)-2*(H)+2*(HW)$)
				--++($-1*(HW)+(VW)$);			
			\ifnum \j=0 \node at ($(P)-0.5*(V\the\mod)+0.5*(H\the\mod)+(0,.2)$) {\small $\inner{a}_\i$}; \fi
			\ifnum \j=1 \node at ($(P)-0.5*(V\the\mod)+0.5*(H\the\mod)-.5*(V)+(0.,-0.05)$) {\small $\intermediate{a}_\i$}; \fi
			\ifnum \j=2 \node at ($(P)-0.5*(V\the\mod)+0.5*(H\the\mod)-.75*(V)+(0.,-0.05)$) {\small $a_\i$}; \fi
		    \coordinate (P) at ($(P)+(H\the\mod)$);
		}
	\coordinate (P) at (0,0);
	}
\end{tikzpicture}
\end{center}
\caption{A chain of 3-regions satisfying $\stack(\tseq{a}_1,
\tseq{a}_2, \tseq{a}_3)$.
}\label{fig:stack}
\end{figure}
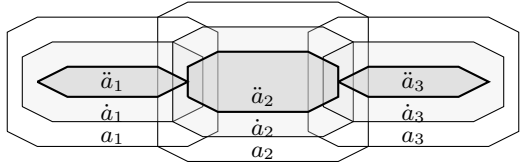
When describing arrangements of 3-regions, we use the variable
$\tseq{r}$ for the triple of variables $(r, \intermediate{r},
\inner{r})$, taking the following conjuncts to be implicit:
\begin{equation*}
(\inner{r} \neq 0) \ \ \land \ \ (\inner{r} \ll \intermediate{r}) \ \ \land \ \ (\intermediate{r} \ll r).
\end{equation*}
In the sequel, when depicting arrangements of 3-regions, we standardly
draw only the kernels of these 3-regions, leaving the reader to
imagine the encasing layers of shell. (This is simply to reduce
diagrammatic clutter.)

For $n \geq 2$, define the formula
\begin{multline*}
\stack(\tseq{r}_1, \dots,
\tseq{r}_n) \ \ =\ \ \bigwedge_{i = 1}^{n-1}
     c(\intermediate{r}_i + \inner{r}_{i + 1} + \cdots + \inner{r}_n) \ \  \land \ \ c(\intermediate{r}_n) \ \
\land\ \
\bigwedge_{\begin{subarray}{c}1 \leq i, j \leq n\\j-i > 1\end{subarray}} \neg C(r_i,r_j).
\end{multline*}
(Observe that the term $c(\intermediate{r}_i + \inner{r}_{i + 1} +
\cdots + \inner{r}_n)$ features the {\em inner shell} of $\tseq{r}_1$,
and the {\em kernels} of $\tseq{r}_2, \dots, \tseq{r}_n$.) Thus, the
triple of 3-regions $(\tseq{a}_1, \tseq{a}_2, \tseq{a}_3)$ in
Fig.~\ref{fig:stack} satisfies $\stack(\tseq{r}_1, \tseq{r}_2,
\tseq{r}_3)$.  This formula allows us to construct sequences of arcs
with useful properties.
\begin{lemma}
\label{lma:stackLemma}
Fix $n \geq 2$, and let $\tseq{a}_1,\dots,\tseq{a}_n$ be a tuple of
3-regions satisfying $\stack(\tseq{r}_1,\dots,\tseq{r}_n)$. Then, for
every point $p_0\in \intermediate{a}_1$ and every point $p_n \in
\inner{a}_n$, there exist points $p_1, \dots, p_{n-1}$ and Jordan arcs
$\alpha_1, \dots, \alpha_n$ such that\textup{:}
\textup{(}i\textup{)} $\alpha = \alpha_1 \cdots \alpha_n$ is a Jordan arc from $p_0$ to
		$p_n$\textup{;} \textup{(}ii\textup{)} $p_i \in \intermediate{a}_{i+1} \cap \alpha_i$, for all $1 \leq i < n$\textup{;} and
	\textup{(}iii\textup{)}	$\alpha_i \subseteq a_i$, for all $1 \leq i \leq n$.
\end{lemma}
\begin{proof}
Let $v_0 = p_0$. Since $v_0\in \ti{a}_1$, $p_n\in
\ti{\intermediate{a}}_n$ and $\intermediate{a}_1 + \inner{a}_2 +
\cdots + \inner{a}_n$ is connected, we see that $v_0$ and $p_0$ lie in
the same component of $\ti{(a_1 + \intermediate{a}_2 + \cdots +
  \intermediate{a}_n)}$. So let $\beta_1$ be a Jordan arc connecting
$v_0$ to $p_n$ in $\ti{(a_1 + \intermediate{a}_2 + \cdots +
  \intermediate{a}_n)}$. Since $a_1$ is disjoint from all the $a_i$
except $a_2$, let $p_1$ be the first point of $\beta_1$ lying in
$\intermediate{a}_2$, so $\beta_1[v_0,p_1]\subseteq \ti{a}_1\cup \{
p_1\}$, i.e., the arc $\beta_1[v_0,p_1]$ is either included in
$\ti{a}_1$, or is an end-cut of $\ti{a}_1$. (We do not rule out $v_0 =
p_1$.) Similarly, let $\beta'_2$ be a Jordan arc connecting $p_1$ to
$p_n$ in $\ti{(a_2 + \intermediate{a}_3 + \cdots +
  \intermediate{a}_n)}$, and let $q_1$ be the last point of $\beta'_2$
lying on $\beta_1[v_0,p_1]$. If $q_1 = p_1$, then set $v_1 = p_1$,
$\alpha_1 = \beta_1[v_0,p_1]$, and $\beta_2 = \beta'_2$, so that the
endpoints of $\beta_2$ are $v_1$ and $p_n$.  Otherwise, we have $q_1
\in \ti{a}_1$.  We can now construct an arc $\gamma_1 \subseteq
\ti{a}_1 \cup \{ p_1 \}$ from $p_1$ to a point $v_1$ on
$\beta'_2[q_1,p_n]$, such that $\gamma_1$ intersects
$\beta_1[v_0,p_1]$ and $\beta'_2[q_1,p_n]$ only at its endpoints,
$p_1$ and $v_1$; see Fig.~\ref{fig:stackLemma}. Let $\alpha_1 =
\beta_1[v_0,p_1]\gamma_1$, and let $\beta_2 = \beta'_2[v_1,p_n]$.

\begin{figure}[ht]
\begin{center}
\begin{tikzpicture}[>=latex, point/.style={circle,draw=black,minimum size=1mm,inner sep=0pt}]
    \node[point,label=below:{\small $v_0 = p_0$}] (P0) at (-0.25,0) {};
    \node[point,label=above right:{\small $q_1$}] (Q1) at (1,0) {};
    \node[point,label=below:{\small $v_1$}] (Q1') at (1.5,-1) {};
    \node[point,label=above:{\small $p_1$}] (P1) at (3,0) {};
    \draw[->,thick] (P0) to node[above] {\footnotesize $\beta_1$} (Q1);
    \draw[->,thick] (Q1) to node[below] {\footnotesize $\beta_1$} (P1);
    \draw[->,dashed] (P1) to node[above] {\footnotesize $\beta_1$} (4.5,0);
	\draw[->,ultra thin] (P1) -- ++(.5,0.5) --++(-.5,.5)  -- (2.3,1);
	\draw[ultra thin] (2.3,1) to node[below, very near start]{\footnotesize $\beta_2'$} (1.2,1) -- ++(-.5,-0.5) -- (Q1);
    \draw[->,ultra thin] (Q1) to node[left]{\footnotesize $\beta_2'$} (Q1');
	\draw[->,thick] (P1) to node[right]{\footnotesize $\gamma_1$} (Q1');
	
	\draw[->,thick] (Q1') to node[below] {\footnotesize $\beta_2$} (3,-1);
    \node at (1.7,-1.8) {$\alpha_1 = \beta_1[v_0,p_1]\gamma_1$};
    \node[point,label=below:{\small $v_{n-2}$}] (Qn2) at (5,-1) {};
    \node[point,label=above right:{\small $q_{n-1}$}] (Qn1) at (6.5,-1) {};
    \node[point,label=below:{\small $v_{n-1}$}] (Qn1') at (7,-2) {};
    \node[point,label=above:{\small $p_{n-1}$}] (Pn1) at (8.5,-1) {};
    \node[point,label=above:{\small $p_n$}] (Pn) at (10,-1) {};
    \draw[->,thick] (Qn2) to node[above] {\footnotesize $\beta_{n-1}$} (Qn1);
    \draw[->,thick] (Qn1) to node[below] {\footnotesize $\beta_{n-1}$} (Pn1);
    \draw[->,dashed] (Pn1) to node[above] {\footnotesize $\beta_{n-1}$} (Pn);
	\draw[->,ultra thin] (Pn1) -- ++(.5,0.5) --++(-.5,.5)  -- (7.8,0);
	\draw[ultra thin] (7.8,0) to node[below, very near start]{\footnotesize $\beta_n'$} (6.7,0) -- ++(-.5,-0.5) -- (Qn1);
    \draw[->,ultra thin] (Qn1) to node[left]{\footnotesize $\beta_n'$} (Qn1');
	\draw[->,thick] (Pn1) to node[right]{\footnotesize $\gamma_{n-1}$} (Qn1');
	\draw[->,thick] (Qn1') -- ++(1.5,0) to node[below] {\footnotesize $\beta_n$} (Pn);
    \node at (4.5,-2.2) {$\alpha_i = \beta_i[v_{i-1},p_i]\gamma_i$};
    \node at (8,-2.7) {$\alpha_n = \beta_n$};
\end{tikzpicture}
\end{center}
\caption{Proof of Lemma~\ref{lma:stackLemma}.}\label{fig:stackLemma}
\end{figure}
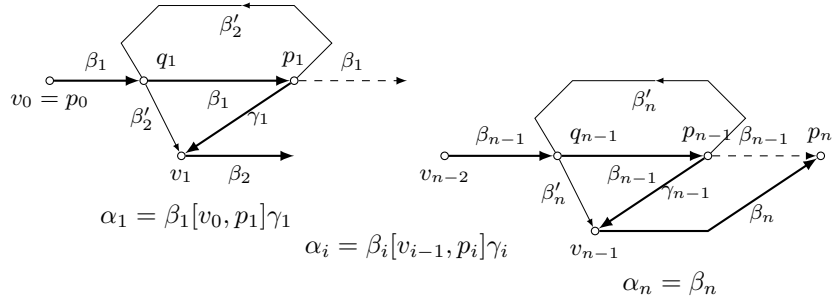

Since $\beta_2$ contains a point $p_2 \in \intermediate{a}_3$, we may
iterate this procedure, obtaining $\alpha_2, \alpha_3, \dots
\alpha_{n-1}, \beta_{n}$. We remark that $\alpha_i$ and $\alpha_{i+1}$
have a single point of contact by construction, while $\alpha_i$ and
$\alpha_j$ ($i < j-1$) are disjoint by the constraint $\neg C(a_i,
a_j)$.  Finally, we let $\alpha_n = \beta_n$; see Fig.~\ref{fig:stackLemma}.
\end{proof}

In fact, we can add a `switch' to the formula $\stack(\tseq{r}_1,
\dots, \tseq{r}_n)$, in the following sense.  Recall from
\SECT\ref{prelim} that if $a, a_0, \dots, a_{n-1}$ are regions satisfying
$\colourComp(r; r_0,\dots,r_{n-1})$, then every connected subset of
$a$---and in particular, any component of $a$---is included in exactly
one of the $a_0,\dots,a_{n-1}$.  Let $z$ be a variable, and consider
what happens when we replace the variable $\intermediate{r}_1$ in
$\stack(\tseq{r}_1, \dots, \tseq{r}_n)$ by the term $(- z) \cdot
\intermediate{r}_1$, and add the conjunct
$\colourComp(\intermediate{r}_1; \ z, -z)$. The result is
\begin{multline*}
\stack_z(\tseq{r}_1, \dots,
\tseq{r}_n) \ \ = \ \
  \colourComp(\intermediate{r}_1; \ z, -z) \ \wedge \
  c(((- z) \cdot \intermediate{r}_1) + \inner{r}_{2} + \cdots + \inner{r}_n)\ \wedge \\
  \bigwedge_{i = 2}^{n-1}
       c(\intermediate{r}_i + \inner{r}_{i+1} + \cdots + \inner{r}_n) \ \ \land \ \ c(\intermediate{r}_n) \ \
  \land \ \
  \bigwedge_{\begin{subarray}{c}1 \leq i, j \leq n\\j-i > 1\end{subarray}} \neg C(r_i,r_j).
\end{multline*}
Now let $\tseq{a}_1,\dots,\tseq{a}_n$ be 3-regions and $d$ a region
satisfying $\stack_d(\tseq{r}_1,\dots,\tseq{r}_n)$. The first conjunct
of the formula ensures that any component of $\intermediate{a}_1$ is
either included in $d$ or included in $-d$. The remaining conjuncts
then have the same effect as $\stack(\tseq{r}_1, \dots,
\tseq{r}_n)$---but only for those components of $\intermediate{a}_1$
included in $-d$. That is, if $p \in(-d)\cdot \intermediate{a}_1$, we
can find an arc $\alpha_1 \cdots \alpha_n$ starting at $p$, with the
properties of Lemma~\ref{lma:stackLemma}.  However, if $p \in d\cdot
\intermediate{a}_1$, no such arc need exist.  Thus, the variable $z$
functions so as to `de-activate' $\stack_z(\tseq{r}_1, \dots,
\tseq{r}_n)$ when we are dealing with a component of
$\intermediate{r}_1$ satisfying $\intermediate{r}_1 \leq z$.

As a further application of Lemma~\ref{lma:stackLemma}, consider the
formula
\begin{multline*}
 \frameFla(\tseq{r}_0, \dots, \tseq{r}_n) \ \ = \ \  \stack(\tseq{r}_0, \dots, \tseq{r}_{n-1}) \wedge \neg C(r_n,r_1+\dots+r_{n-2})\wedge{}\\
     c(\intermediate r_n)\wedge (\intermediate{r}_0\cdot \intermediate r_n\neq 0)\wedge
	(\inner{r}_{n-1}\cdot \intermediate r_n\neq 0).
\end{multline*}
This formula allows us to construct Jordan curves in the plane, in the
following sense:
\begin{lemma}
\label{lma:FrameLemma}
Fix $n \geq 3$, and let $\tseq{a}_0,\dots,\tseq{a}_{n}$ be a tuple of
3-regions satisfying $\frameFla(\tseq{r}_0,\dots,\tseq{r}_{n})$. Then there exist Jordan arcs
$\gamma_0, \dots, \gamma_n$ such that $\gamma_0\cdots\gamma_n$ is a
Jordan curve and $\gamma_i \subseteq a_i$, for all $0 \leq i \leq n$.
\end{lemma}
\begin{proof}
By Lemma~\ref{lma:stackLemma}, let
$\alpha_0,\gamma_1,\dots,\gamma_{n-2},\alpha_{n-1}$ be Jordan arcs in
the respective regions $a_0,\dots,a_{n-1}$ such that
$\alpha_0\cdots\alpha_{n-1}$ is a Jordan arc connecting a point
$\tilde p\in \intermediate{a}_0\cdot \intermediate a_n$ to a point
$\tilde q\in \inner{a}_{n-1}\cdot \intermediate a_n$; see
Fig.~\ref{fig:FrameLemma}.
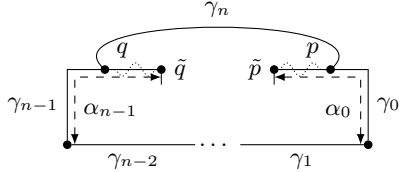
\begin{figure}[ht]
\begin{center}
\begin{tikzpicture}[>=latex,point/.style={circle,draw=black,minimum size=1mm,inner sep=0pt}]\small
%
%
			\node (tp) at (.75,1.5) [point,fill=black,label=left:{$\tilde p$}] {};
			\node (p) at (1.5,1.5) [point,fill=black,label=above left:{$p$}]{};
            \node (pq) at (2,0.5) [point,fill=black] {};
			\draw[decorate,decoration=snake,densely dotted] (tp) -- (p);
			\draw (tp) -| (pq) node[near end,right] {$\gamma_0$} --++(-1.7,0) node[midway,below]{$\gamma_1$}node[left]{$\dots$};
			\node (tq) at (-.75,1.5) [point,fill=black,label=right:{$\tilde q$}] {};
			\node (q) at (-1.5,1.5) [point,fill=black,label=above right:{$q$}] {};
            \node (qq) at (-2,0.5) [point,fill=black] {};
			\draw[decorate,decoration=snake,densely dotted] (tq) -- (q);
			\draw (tq) -| (qq)  node[near end,left] {$\gamma_{n-1}$} --++(1.7,0) node[midway,below]{$\gamma_{n-2}$};
			\draw(1.5,1.5) to[out=40,in=140] (-1.5,1.5);
            \node at (0, 2.3) {$\gamma_n$};
            \draw[<->,dashed] (-1.9,0.5) -- (-1.9,1.4) node[midway,right]{$\alpha_{n-1}$} -- (-0.75,1.4);
            \draw[ultra thin] (tq) -- +(0,-0.2);
            \draw[<->,dashed] (1.9,0.5) -- (1.9,1.4) node[midway,left]{$\alpha_0$} -- (0.75,1.4);
            \draw[ultra thin] (tp) -- +(0,-0.2);
\end{tikzpicture}
\end{center}
\caption{Establishing a Jordan curve.}
\label{fig:FrameLemma}
\end{figure}%
Because $\intermediate a_n$ is a connected subset of the interior of
$a_n$, let $\alpha_n\subseteq \ti a_n$ be an arc connecting $\tilde p$
and $\tilde q$. Note that $\alpha_n$ does not intersect $\alpha_i$,
for $1\leq i \leq n-2$. Let $p$ be the last point on $\alpha_0$ that
is on $\alpha_n$ (possibly $\tilde p$), and $q$ be the first point on
$\alpha_{n-1}$ that is on $\alpha_n$ (possibly $\tilde q$). Let
$\gamma_0$ be the final segment of $\alpha_0$ starting at $p$ and let
$\gamma_{n-1}$ be the initial segment of $\alpha_{n-1}$ ending at $q$.
Finally, let $\gamma_n = \alpha_n[p,q]$ or $\gamma_n = \alpha_n[q,p]$,
depending on whether $p$ or $q$ is encountered first on
$\alpha_n$. Then the arcs $\gamma_i$, $0\leq
i\leq n$, are as required.
\end{proof}

We are now ready to prove the main result of this section.  Again,
recall that if $a, a_0, \dots, a_{n-1}$ are regions satisfying
$\colourComp(r; r_0,\dots,r_{n-1})$, then
every connected subset of $a$---and in particular, any Jordan arc
$\alpha \subseteq a$---is included in exactly one of the $a_i$, for $0
\leq i < n$. In this case, it is sometimes helpful to think of
$\alpha$ as being `labelled' by a letter of the alphabet $a_0, \dots,
a_{n-1}$.

\begin{theorem}\label{theo:cBCc2}
The problems $\Sat(\cBCc,\RC(\R^2))$ and $\Sat(\cBCc,\RCP(\R^2))$ are r.e.-hard.
\end{theorem}
\begin{proof}
Again,
we proceed via a reduction of the Post correspondence problem
(PCP), constructing, for any instance $\pcpW$, a formula $\psi_\pcpW$
with the property that the following are equivalent: (\emph{i}) $\pcpW$
is positive; (\emph{ii}) $\psi_\pcpW$ is satisfiable over
$\RCP(\R^2)$; (\emph{iii}) $\psi_\pcpW$ is satisfiable over
$\RC(\R^2)$.  This establishes the theorem.  As with the proofs of
Theorems~\ref{theo:cBCcn} and~\ref{theo:inftyBci}, we equivocate
between variables and the regions to which they are assigned in some
putative interpretation over $\RC(\R^2)$. In this proof, if $k$ is an
integer, $\md{k}$ indicates the value of $k$ modulo 3. The proof
proceeds in six stages.
\begin{figure}[ht]\begin{center}
	\begin{tikzpicture}[yscale=0.8]\small
		\coordinate (sh1) at (0,.32);	\coordinate (sw1) at (.45,0); 
		\coordinate (sh) at (0,.25);	\coordinate (sw) at (.35,0); 
		\coordinate (bh) at ($9*(sh)$);	
		\coordinate (bw) at (10.5,0); 
		\coordinate (dw) at (.75,0); 
				
		\draw[fill=gray!20] ($-1*(sw)-(bh)-2*(sh)$) rectangle ($(bw)+(sw)+(bh)+2*(sh)$);
		\draw[fill=white] ($(sw)-(bh)$) rectangle ($(bw)-(sw)+(bh)$);
		\draw[fill=gray!50] ($(sw1)-(sh1)$) rectangle ($(bw)-(sw1)+(sh1)$) node[midway]{$\tseq{d}_3$};
		\draw ($(sw1)-(sh1)$) rectangle ($(sw1)+(sh1)+(dw)$) node[midway]{$\tseq{d}_1$};
		\draw ($(sw1)-(sh1)+(dw)$) rectangle ($(sw1)+(sh1)+2*(dw)$) node[midway]{$\tseq{d}_2$};
		
		\draw ($(bw)-(sw1)+(sh1)$) rectangle ($(bw)-(sw1)-(sh1)-(dw)$) node[midway]{$\tseq{d}_5$};
		\draw ($(bw)-(sw1)+(sh1)-(dw)$) rectangle ($(bw)-(sw1)-(sh1)-2*(dw)$) node[midway]{$\tseq{d}_4$};
		
		\draw[fill=gray!50] ($-1*(sh1)-(sw1)$) rectangle ($(sw1)+(sh1)$);
		\draw ($-1*(sh)-(sw)$) rectangle ($(sw)+(sh)$) node[midway]{$\tseq{s}_0$};
		
		\draw[fill=gray!50] ($(bw)-(sh1)-(sw1)$) rectangle ($(bw)+(sw1)+(sh1)$);
		\draw ($(bw)-(sh)-(sw)$) rectangle ($(bw)+(sw)+(sh)$) node[midway]{$\tseq{s}_9$};
		
		\foreach \t/\s in {-1/,1/'}
		{
			\foreach \x in {1,2,3,4,5}
				\draw ($2*\t*\x*(sh)-(sh)-(sw)$) rectangle ($2*\t*\x*(sh)+(sw)+(sh)$) node[midway]{$\tseq{s}_\x\s$};
			\foreach \x/\l in {1/8,2/7}
				\draw ($-1*\t*(sh)+4*\t*\x*(sh)-2*(sh)+(bw)-(sw)$) rectangle ($-1*\t*(sh)+4*\t*\x*(sh)+(bw)+(sw)+2*(sh)$)
					node[midway]{$\tseq{s}_\l\s$};
		}		
		\node at ($.5*(bw)-(bh)-(sh)$) {$\tseq{s}_6$};
		\node at ($.5*(bw)+(bh)+(sh)$) {$\tseq{s}_6'$};
		
		\draw[ultra thin] ($.15*(bw)+.5*(bh)$) node[above right]{$\tseq{d}_0$} -- ($.5*(sh1)+.9*(sw1)$);
		\draw[ultra thin] ($.85*(bw)+.5*(bh)$) node[above left]{$\tseq{d}_6$} -- ($(bw)+.5*(sh1)-.9*(sw1)$);
		
	\end{tikzpicture}
	\end{center}
\caption{A tuple of 3-regions
  satisfying~{\eqref{eq:PCPFrame}}--{\eqref{eq:PCPCord}} (showing
  kernels only).  The 3-regions $\tseq{s}_0$ and $\tseq{s}_9$ are
  right inside $\tseq{d}_0$ and $\tseq{d}_6$, respectively, as
  specified by~\eqref{eq:PCPCord:endpoints}.}
\label{fig:concrete1}
\end{figure}
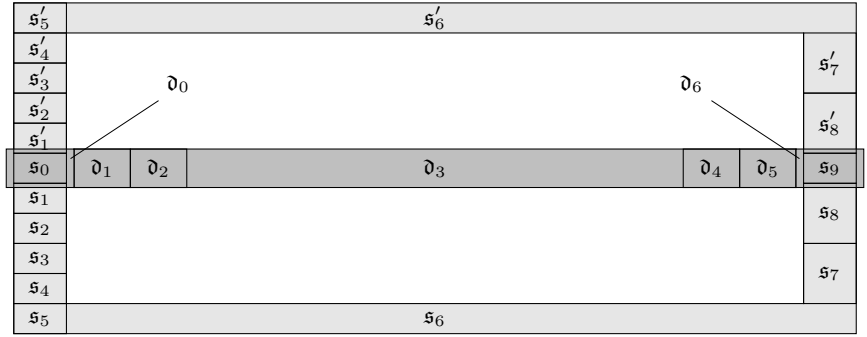

\medskip

\noindent
\textbf{Stage 1.} In the first stage, we define an assemblage of arcs
that will serve as scaffolding for the ensuing construction.  Consider
the arrangement of polygonal 3-regions depicted in
Fig.~\ref{fig:concrete1}, assigned to the 3-region variables
$\tseq{s}_0, \dots, \tseq{s}_9$, $\tseq{s}_8', \dots, \tseq{s}_1'$,
$\tseq{d}_0,\dots, \tseq{d}_6$ as indicated.  (Note that we have here
followed the convention of depicting only the kernels of 3-regions.)

It is easy to verify that this arrangement can be made to satisfy the
following formulas:
\begin{align}
\label{eq:PCPFrame}
& \frameFla(\tseq{s}_0, \tseq{s}_1,\dots, \tseq{s}_8,\tseq{s}_9,\tseq{s}_8',\dots, \tseq{s}_1'),\\
\label{eq:PCPCord:endpoints}
& (s_0 \leq \intermediate{d}_0) \land (s_9 \leq \inner{d}_6),\\
\label{eq:PCPCord} & \stack(\tseq{d}_0,\dots,\tseq{d}_6).
\end{align}
And obviously, the arrangement can be made to satisfy any formula
\begin{equation}\label{eq:pcp:C}
\neg C(r,r'),
\end{equation}
for which the corresponding 3-regions $\tseq{r}$ and $\tseq{r}'$ are
drawn as not being in contact. (Remember, $r$ is the outer shell of the 3-region $\tseq{r}$, and similarly for
$r'$; so we must take these shells to hug the kernels depicted in
Fig.~\ref{fig:concrete1} quite closely.)  Thus, for example,
\eqref{eq:pcp:C} includes $\neg C(s_0, d_1)$, but not $\neg C(s_0,
d_0)$ or $\neg C(d_0, d_1)$.

Now suppose $\tseq{s}_0, \dots, \tseq{s}_9$, $\tseq{s}_8', \dots,
\tseq{s}_1'$, $\tseq{d}_0,\dots, \tseq{d}_6$ is {\em any} collection
of 3-regions (not necessarily polygonal)
satisfying~\eqref{eq:PCPFrame}--\eqref{eq:pcp:C}.  By
Lemma~\ref{lma:FrameLemma} and~\eqref{eq:PCPFrame}, let $\gamma_0,
\dots, \gamma_9,\gamma_8',\dots,\gamma_1'$ be Jordan arcs included in
the respective regions\linebreak
$s_0, \dots, s_9,s_8',\dots,s_1'$, such that
$\Gamma = \gamma_0 \cdots \gamma_9 \cdot \gamma_8' \cdots \gamma_1'$
is a Jordan curve (note that $\gamma_i'$ and $\gamma_i$ have opposite
directions). We select points $\tilde{o}$ on $\gamma_0$ and
$\tilde{o}'$ on $\gamma_9$; see Fig.~\ref{fig:arcs0}.
By~\eqref{eq:PCPCord:endpoints}, $\tilde{o} \in \intermediate{d}_0$
and $\tilde{o}' \in \inner{d}_6$. By Lemma~\ref{lma:stackLemma}
and~\eqref{eq:PCPCord}, let $\tilde{\chi}_1$, $\chi_2$, $\tilde{\chi}_3$
be Jordan arcs in the respective regions
\begin{equation*}
(d_0 + d_1),\qquad (d_2 + d_3 + d_4),\qquad (d_5+ d_6)
\end{equation*}
such that $\tilde{\chi}_1 \chi_2 \tilde{\chi}_3$ is a Jordan arc from
$\tilde{o}$ to $\tilde{o}'$. Let $o$ be the last point of
$\tilde{\chi}_1$ lying on $\Gamma$, and let $\chi_1$ be the final
segment of $\tilde{\chi}_1$, starting at $o$.  Let $o'$ be the
first point of $\tilde{\chi}_3$ lying on $\Gamma$, and let $\chi_3$ be
the initial segment of $\tilde{\chi}_3$, ending at $o'$.
By~\eqref{eq:pcp:C}, we see that the arc $\chi_1\chi_2\chi_3$
intersects $\Gamma$ only in its endpoints, and is thus a chord of
$\Gamma$, as shown in Fig.~\ref{fig:arcs0}.
\begin{figure}[ht]
\begin{center}
	\begin{tikzpicture}[>=latex,point/.style={circle,draw=black,minimum size=1mm,inner sep=0pt},yscale=0.7]\small
        \draw (0,0) rectangle +(10,4);
        \draw (0,2) -- +(10,0);
        \draw[ultra thin] (0,0.5) -- +(-0.2,0);
        \draw[ultra thin] (0,3.5) -- +(-0.2,0);
        \draw[->,dashed] (-0.1,3.5) -- +(0,-3) node[near start,left] {$\gamma_0$};
        \draw[ultra thin] (10,0.5) -- +(0.2,0);
        \draw[ultra thin] (10,3.5) -- +(0.2,0);
        \draw[->,dashed] (10.1,0.5) -- +(0,3) node[near end,right] {$\gamma_9$};
        \draw[->,dashed] (-0.1,0.5) -- ++(0,-0.6) -- ++(10.2,0) node[midway,below] {$\gamma_1\cdots\gamma_8$} -- ++(0,0.6);
        \draw[->,dashed] (10.1,3.5) -- ++(0,0.6) -- ++(-10.2,0) node[midway,above] {$\gamma_8'\cdots\gamma_1'$} -- ++(0,-0.6);
        \node (to) at (0,1.3) [point,fill=black,label=left:{$\tilde o$}] {};
        \node (o) at (0,2) [point,fill=black,label=left:{$o$}] {};
		\draw[decorate,decoration=snake,densely dotted] (to) -- (o);
        \node (top) at (10,1.3) [point,fill=black,label=right:{$\tilde o'$}] {};
        \node (op) at (10,2) [point,fill=black,label=right:{$o'$}] {};
		\draw[decorate,decoration=snake,densely dotted] (top) -- (op);
        \draw[ultra thin] (2,1.8) -- +(0,0.4);
        \draw[ultra thin] (8,1.8) -- +(0,0.4);
        \draw[ultra thin] (0,1.3) -- +(0.2,0);
        \draw[ultra thin] (10,1.3) -- +(-0.2,0);
        \draw[->,dashed] (0.1,1.3) -- (0.1,1.9) -- (2,1.9) node[midway,below] {$\tilde{\chi}_1$};
        \draw[->,dashed] (8,1.9) -- (9.9,1.9) node[midway,below] {$\tilde{\chi}_3$} -- (9.9,1.3);
        \draw[->,dashed] (0,2.1) -- +(2,0) node[midway,above] {$\chi_1$};
        \draw[->,dashed] (2,2.1) -- +(6,0) node[midway,above] {$\chi_2$};
        \draw[->,dashed] (8,2.1) -- +(2,0) node[midway,above] {$\chi_3$};
    \end{tikzpicture}
\end{center}
\caption{The arcs $\gamma_0, \dots, \gamma_9,\gamma_8',\dots,\gamma_1'$ and $\chi_1, \chi_2, \chi_3$.}
\label{fig:arcs0}
\end{figure}
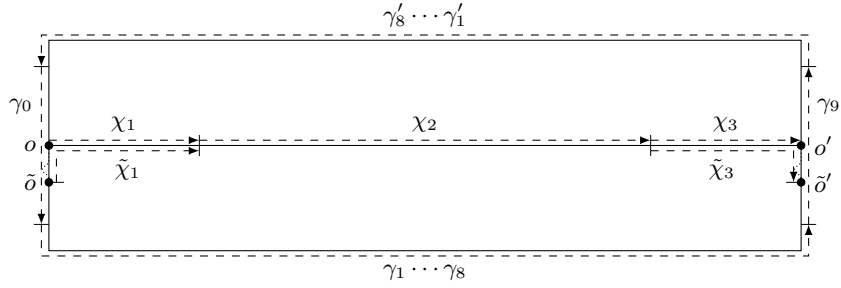

As before, we treat these diagrams as being drawn on a spherical
canvas.  For ease of reference, we refer to the two rectangles in
Fig.~\ref{fig:arcs0} as the `upper window' and `lower window', it
being understood that these are simply handy labels: in particular,
either (but not both) of these `windows' may be unbounded.

\medskip

\noindent
\textbf{Stage 2.}  In this stage, we construct a sequence
of triples $(\zeta_i,\eta_i,\kappa_i)$ of arcs of indeterminate length
$n \geq 1$, such that the members of the former sequence all lie in
the lower window. (Recall that $\md{k}$ denotes $k$ modulo 3).  Let
$\tseq{a}$, $\tseq{b}$, $\tseq{z}$, $\tseq{a}_{i,j}$ and
$\tseq{b}_{i,j}$ ($0 \leq i < 3$, $1 \leq j \leq 6$) be 3-region
variables, 
and consider the formulas
\begin{align}
\label{eq:aSeq1}
& (s_6 \leq \inner{a}) \wedge (s'_6 \leq \inner{b})
\wedge (s_3 \leq \intermediate{a}_{0,3}), \\
\label{eq:aSeq:b}
& \bigwedge_{i= 0}^2
\stack_{\intermediate{z}}(\tseq{a}_{\md{i-1},3}, \tseq{b}_{i,1}, \dots, \tseq{b}_{i,6}, \tseq{b}),\\
\label{eq:aSeq:a}
&  \bigwedge_{i= 0}^2
\stack(\tseq{b}_{i,3}, \tseq{a}_{i,1}, \dots, \tseq{a}_{i,6}, \tseq{a}).
\end{align}
The arrangement of polygonal 3-regions depicted in
Fig.~\ref{fig:concrete2} (with $\tseq{z}$ assigned appropriately) is one such
satisfying assignment.
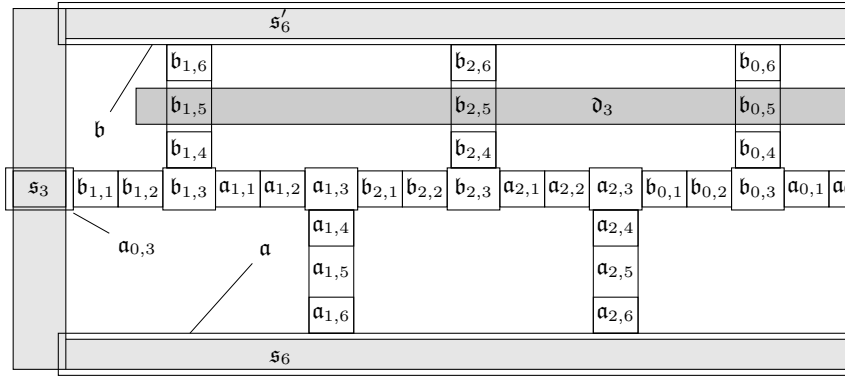
\begin{figure}[ht]
	\begin{center}
	\begin{tikzpicture}[yscale=0.8]\small
		\coordinate (sh1) at (0,.35);	\coordinate (sw1) at (.45,0); 
		\coordinate (sh) at (0,.3);	\coordinate (sw) at (.35,0); 
		\coordinate (dsh) at ($(sh1)-(sh)$); \coordinate (dsw) at ($(sw1)-(sw)$); 		
		\coordinate (bh) at ($8*(sh)$);	
		\coordinate (bw) at (10.7,0); 
			
		\clip[] ($-1*(sw1)-1.3*(bh)$) rectangle ($(bw)+1.3*(bh)$);
		
		\draw[fill=gray!20] ($(sw)+(bw)-(bh)-2*(sh)$) --($-1*(sw)-(bh)-2*(sh)$) -- ($-1*(sw)+(bh)+2*(sh)$)-- ($(bw)+(sw)+(bh)+2*(sh)$);
		\draw[fill=white] ($(bw)+(sw)-(bh)-(0,0.1)$)--($(sw)-(bh)-(0,0.1)$)--($(sw)+(bh)+(0,0.1)$)--($(bw)+(sw)+(bh)+(0,0.1)$);
		\draw ($(sw)-(bh)-(0,0.1)$) --++ ($-2*(sh)+(0,0.1)$);\draw ($(sw)+(bh)+(0,0.1)$) --++ ($2*(sh)-(0,0.1)$);
		
		\draw ($(sw)+(bw)-(bh)-(0,0.05)+(dsh)$) rectangle ++($-1*(bw)-(dsw)-2*(sh1)$);
		\draw ($(sw)+(bw)+(bh)+(0,0.05)-(dsh)$) rectangle ++($-1*(bw)-(dsw)+2*(sh1)$);
		\node (b) at (3,-1) {$\tseq{a}$}; 
        \draw[ultra thin] (b) -- (2,-2.4);
		\node at ($.3*(bw)-(bh)-(sh)-(0,0.1)$) {$\tseq{s}_6$};
		\node (a) at (0.8, 1) {$\tseq{b}$}; 
        \draw[ultra thin] (a) -- (1.5,2.4);
		\node at ($.3*(bw)+(bh)+(sh)+(0,0.1)$){$\tseq{s}_6'$};
		\draw ($-1*(sh)-(sw)$) rectangle ($(sh)+(sw)$) node[midway]{$\tseq{s}_3$};
		\draw ($-1*(sh1)-(sw1)$) rectangle ($(sh1)+(sw1)$);
		\node (a03) at (1.3,-1) {$\tseq{a}_{0,3}$}; 
		\draw[ultra thin] (a03) -- (0.45,-0.4); 
		
		\draw[fill=gray!40] ($.12*(bw)+.5*(sh1)+.5*(bh)-(sh)$) rectangle ($1.1*(bw)+.5*(sh1)+.5*(bh)+(sh)$);
		\node at ($.7*(bw)+.5*(sh1)+.5*(bh)$) {$\tseq{d}_3$};
		\foreach \sc in {.85}
		{
			\coordinate (P) at ($(sw1)+\sc*(sw)$);		
			\foreach \x in {1,2,0,1}
			{
				\foreach \y/\d/\do/\s in {1/1/-1/b,2/1/-1/b,3/1/-1/b,1/-1/1/a,2/-1/1/a,3/-1/1/a}
				{
					\ifnum \y=3
						\draw ($(P)-\sc*(sw)-(sh1)$) rectangle ($(P)+2*(sw)-\sc*(sw)+(sh1)$)
							node[midway]{$\tseq{\s}_{\x,\y}$}; 
						\draw ($(P)+(sw)-2*\sc*(sw)+\d*(sh1)$) rectangle ($(P)+(sw)+\d*(bh)$)
							node[midway]{$\tseq{\s}_{\x,5}$}; 
						\draw ($(P)+(sw)-2*\sc*(sw)+\d*(sh1)$) rectangle ($(P)+(sw)+\d*(sh1)+2*\d*(sh)$)
							node[midway]{$\tseq{\s}_{\x,4}$}; 
						\draw ($(P)+(sw)-2*\sc*(sw)+\d*(bh)+2*\do*(sh)$) rectangle ($(P)+(sw)+\d*(bh)$)
							node[midway]{$\tseq{\s}_{\x,6}$}; 
						\coordinate (P) at ($(P)+2*(sw)$); 
					\else
						\draw ($(P)-\sc*(sw)-(sh)$) rectangle ($(P)+\sc*(sw)+(sh)$)
							node[midway]{$\tseq{\s}_{\x,\y}$}; 
						\coordinate (P) at ($(P)+2*\sc*(sw)$); 
					\fi
				}
			}
		}
	\end{tikzpicture}
	\end{center}
\caption{A tuple of 3-regions
  satisfying~{\eqref{eq:aSeq1}}--{\eqref{eq:aSeq:a}} (showing kernels
  only).  The arrangement of components of the $\tseq{a}_{i,j}$ and
  $\tseq{b}_{i,j}$ repeats an indeterminate number of times.
  The 3-regions $\tseq{s}_3$, $\tseq{s}_6$, $\tseq{s}'_6$ and
  $\tseq{d}_3$ are as in Fig.~\ref{fig:concrete2}.
  }
\label{fig:concrete2}
\end{figure}
We stipulate that~\eqref{eq:pcp:C} applies now to all regions depicted
in either Fig.~\ref{fig:concrete1} or Fig.~\ref{fig:concrete2}, and we
further stipulate
\begin{equation}\label{eq:znots}
z \cdot (s_0 + \dots + s_9 + s'_1 + \dots + s'_8 + d_1 + \dots + d_4) = 0.
\end{equation}
Note that $d_5$ does not appear in this constraint; thus, $\ti{z}$ may
intersect the arc $\chi_3$.  Again, these additional constraints are
evidently satisfiable.

Now suppose we are given \emph{any} collection of regions (not
necessarily polygonal)
satisfying~\eqref{eq:PCPFrame}--\eqref{eq:znots}.  And let the arcs
$\gamma_0, \dots, \gamma_9,\gamma_8',\dots,\gamma_1'$ and $\chi_1,
\chi_2, \chi_3$ be as defined above.  It will be convenient in this
stage to rename $\gamma_6$ and $\gamma'_6$ as $\lambda_0$ and $\mu_0$,
respectively.  Thus, $\lambda_0$ forms the bottom edge of the lower
window, and $\mu_0$ the top edge of the upper window. Likewise, we
rename $\gamma_3$ as $\alpha_0$, forming part of the left-hand side of
the lower window. Let $\tilde{q}_{1,1}$ be any point of $\alpha_0$,
$p^*$ any point of $\lambda_0$, and $q^*$ any point of $\mu_0$; see
Fig.~\ref{subfig:arcs1}.  By~\eqref{eq:aSeq1}, then, $\tilde{q}_{1,1}
\in \intermediate{a}_{0,3}$, $p^* \in \inner{a}$, and $q^* \in
\inner{b}$.  Certainly, the constraint~\eqref{eq:znots} ensures that $\tilde{q}_{1,1} \in
(-\intermediate{z})$.  By Lemma~\ref{lma:stackLemma}
and~\eqref{eq:aSeq:b}, we may draw an arc $\tilde{\beta}_1$ from
$\tilde{q}_{1,1}$ to $q^*$, with successive segments
$\tilde{\beta}_{1,1}, \beta_{1,2}, \dots, \beta_{1,5},
\tilde{\beta}_{1,6}$ lying in the respective regions
\begin{equation*}
a_{0,3} +
b_{1,1},\ \  b_{1,2}, \ \ \dots, \ \ b_{1,5}, \ \ b_{1,6} + b;
\end{equation*}
further, we can guarantee that $\beta_{1,2}$ contains a point
$\tilde{p}_{1,1} \in \intermediate{b}_{1,3}$.  Denote the last point
of $\beta_{1,5}$ by $q_{1,2}$. Also, let $q_{1,1}$ be the last point
of $\tilde{\beta}_1$ lying on $\alpha_0$, and $q_{1,3}$ the first
point of $\tilde{\beta}_1$ lying on $\mu_0$. Finally, let $\beta_1$ be
the segment of $\tilde{\beta}_1$ between $q_{1,1}$ and $q_{1,2}$; and
let $\mu_1$ be the segment of $\tilde{\beta}_1$ from $q_{1,2}$ to
$q_{1,3}$ followed by the final segment of $\mu_0$ from $q_{1,3}$; see
Fig.~\ref{subfig:arcs1}.  By repeatedly using the constraints
in~\eqref{eq:pcp:C}, it is easy to see that $\beta_1$ and the initial segment of $\mu_1$ up to $q_{1,3}$
together form a chord of $\Gamma$. Adding the constraint
\begin{equation}\label{eq:connectbd}
c(b_{1,5} + d_3),
\end{equation}
and taking into account the constraints in~\eqref{eq:pcp:C} ensures
that this chord divides the residual domain of $\Gamma$ containing $\chi_2$
into the regular closed sets $S_1$ and $S_1'$, as shown in Fig.~\ref{subfig:arcs1}.
The wiggly lines indicate that we do not care about the
exact positions of $\tilde{q}_{1,1}$ or $q^*$; otherwise,
Fig.~\ref{subfig:arcs1} is again completely general.
Note that $\mu_1$ lies entirely in $b_{1,6} + b$, and hence
certainly in the region
\begin{equation}\label{eq:bstar}
b^* = b_{0,6} + b_{1,6} + b_{2,6} + b.
\end{equation}

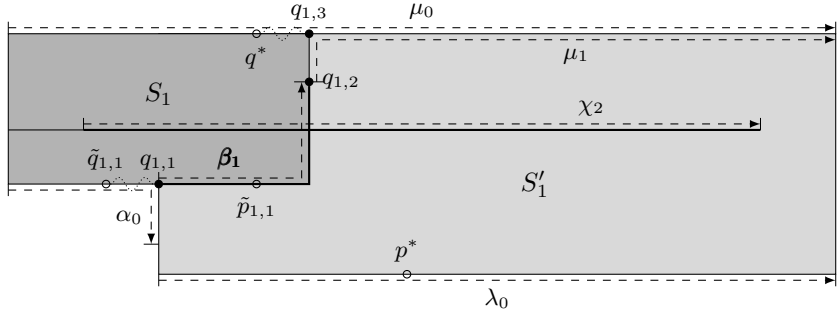
\begin{figure}[ht]
\begin{center}
	\begin{tikzpicture}[>=latex,point/.style={circle,draw=black,minimum size=1mm,inner sep=0pt},yscale=0.8]\small
        \filldraw[fill=gray!30] (0,0) rectangle +(9,4);
        \filldraw[fill=gray!60] (-2,1.5) rectangle +(4,2.5);
        \node (q11) at (0,1.5) [point,fill=black,label=above:{$q_{1,1}$}] {};
        \node (q11t) at (-0.7,1.5) [point,label=above:{$\tilde q_{1,1}$}] {};
		\draw[densely dotted,decorate,decoration=snake] (q11t) -- (q11);
        \node (q13) at (2,4) [point,fill=black,label=above:{$q_{1,3}$}] {};
        \node (q13t) at (1.3,4) [point,label=below:{$q^*$}] {};
		\draw[densely dotted,decorate,decoration=snake] (q13t) -- (q13);
        \node (p11t) at (1.3,1.5) [point,label=below:{$\tilde p_{1,1}$}] {};
        \draw[ultra thin] (0,1.5) -- +(0,0.2);
        \draw[ultra thin] (1.8,3.2) -- +(0.4,0);
        \node (q12) at (2,3.2) [point,fill=black,label=right:{$q_{1,2}$}] {};
        \draw[dashed,->] (0, 1.6) -- (1.9,1.6) node[midway,above]{$\pmb{\beta_1}$} -- (1.9,3.2);
        \draw[dashed,->] (2.1, 3.2) -- (2.1,3.9) -- (9,3.9) node[midway,below]{$\mu_1$};
        \draw[ultra thin] (-2,4) --++(0,.2); \draw[ultra thin] (9,4) --++(0,.2);
        \draw[dashed,->] (-2,4.1) -- (9,4.1) node[midway,above]{$\mu_0$};
	\draw[ultra thin] (0,0) --++(0,-.2);\draw[ultra thin] (9,0) --++(0,-.2);
        \draw[dashed,->] (0,-0.1) -- (9,-0.1) node[midway,below]{$\lambda_0$};
	\draw[ultra thin] (-2,1.5) --++(0,-.2);\draw[ultra thin] (0,0.5)--++(-.2,0);
        \draw[dashed,->] (-2,1.4) -- (-0.1,1.4) -- (-0.1,0.5) node[midway,left]{$\alpha_0$};
        \node (q13t) at (3.3,0) [point,label=above:{$p^*$}] {};
        \node at (0,3) {\normalsize $S_1$};
        \node at (5,1.5){\normalsize $S_1'$};
        \draw (-2,2.4) -- +(9,0);
        \draw[thick] (-1,2.4) -- +(9,0);
        \draw[->,dashed] (-1,2.5) -- +(9,0) node[near end,above] {$\chi_2$};
        \draw[ultra thin] (-1,2.4) -- +(0,0.2);
        \draw[ultra thin] (8,2.4) -- +(0,0.2);
        \draw[thick] (q11) -| (q12);
	\end{tikzpicture}
\end{center}
\caption{The arc $\beta_1$.}\label{subfig:arcs1}
\end{figure}

Recall that $\tilde{p}_{1,1} \in \intermediate{b}_{1,3}$ and $p^* \in
\inner{a}$.  By Lemma~\ref{lma:stackLemma} and~\eqref{eq:aSeq:a}, we
may draw an arc $\tilde{\alpha}_1$ from $\tilde{p}_{1,1}$ to $p^*$,
with successive segments $\tilde{\alpha}_{1,1}, \alpha_{1,2},
\dots, \alpha_{1,5}, \tilde{\alpha}_{1,6}$ lying in the respective
regions
\begin{equation*}
b_{1,3} + a_{1,1}, \ \ a_{1,2}, \ \ \dots, \ \ a_{1,5}, \ \ a_{1,6} +
a;
\end{equation*}
further, we can guarantee that the segment $\alpha_{1,2}$ contains a
point $\tilde{q}_{2,1}\in \intermediate{a}_{1,3}$.  (Thus:
$\alpha_{1,2}$ lies in $a_{1,2}$, but nevertheless contains at least
one point lying in $\intermediate{a}_{1,3}$.) Denote the last point of
$\alpha_{1,5}$ by $p_{1,2}$.  Also, let $p_{1,1}\in b_{1,3}$ be the
last point of $\tilde{\alpha}_1$ lying on $\beta_1$, and $p_{1,3}$ the
first point of $\tilde{\alpha}_1$ lying on $\lambda_0$.
From~\eqref{eq:pcp:C}, and these points must be arranged as shown in
Fig.~\ref{subfig:arcs2}. In particular, the segment of
$\tilde{\alpha}_1$ between $p_{1,1}$ and $p_{1,3}$ is a chord in
$S_1'$ and divides it into regions $R_1$ and $R_1'$. Let $\alpha_1$ be
the segment of $\tilde{\alpha}_1$ between $p_{1,1}$ and $p_{1,2}$.
\begin{figure}[ht]
\begin{center}
	\begin{tikzpicture}[>=latex,point/.style={circle,draw=black,minimum size=1mm,inner sep=0pt},yscale=0.8]\small
        \filldraw[fill=gray!30] (0,0) rectangle +(9,4);
        \filldraw[fill=gray!60] (-2,1.5) rectangle +(4,2.5);
        \filldraw[fill=gray!5] (0,0) rectangle +(4,1.5);
        \node (q11) at (0,1.5) [point,fill=black,label=above:{$q_{1,1}$}] {};
        \node (q13) at (2,4) [point,fill=black,label=above:{$q_{1,3}$}] {};
        \node (q13t) at (1.3,4) [point,label=below:{$q^*$}] {};
        \draw[ultra thin] (0,1.5) -- +(0,0.2);
        \draw[ultra thin] (1.8,3.2) -- +(0.4,0);
        \node (q12) at (2,3.2) [point,fill=black,label=right:{$q_{1,2}$}] {};
        \draw[dashed,->] (0, 1.6) -- (1.9,1.6) node[midway,above]{$\pmb{\beta_1}$} -- (1.9,3.2);
        \draw[dashed,->] (2.1, 3.2) -- (2.1,3.9) -- (9,3.9) node[midway,below]{$\mu_1$};
		\draw[ultra thin] (-2,4) --++(0,.2); \draw[ultra thin] (9,4) --++(0,.2);
        \draw[dashed,->] (-2,4.1) -- (9,4.1) node[midway,above]{$\mu_0$};		
	\draw[ultra thin] (0,0) --++(0,-.2);\draw[ultra thin] (9,0) --++(0,-.2);
        \draw[dashed,->] (0,-0.1) -- +(9,0) node[near end,below]{$\lambda_0$};
	\draw[ultra thin] (-2,1.5) --++(0,-.2);\draw[ultra thin] (0,0.5)--++(-.2,0);
        \draw[dashed,->] (-2,1.4) -- (-0.1,1.4) -- (-0.1,0.5) node[midway,left]{$\alpha_0$};
        \node at (0,3) {\normalsize $S_1$};
        \node at (6,1.5){\normalsize $R_1'$};		
        \draw (-2,2.4) -- +(11,0);
        \draw[ultra thin] (-1,2.4) -- +(0,0.2);
        \draw[ultra thin] (8,2.4) -- +(0,0.2);
        \draw[thick] (-1,2.4) -- +(9,0);
        \draw[->,dashed] (-1,2.5) -- +(9,0) node[near end,above] {$\chi_2$};
        \node (p11) at (2,1.5) [point,fill=black,label=below:{$p_{1,1}$}] {};
        \node (p11t) at (1.3,1.5) [point,label=below:{$\tilde p_{1,1}$}] {};
		\draw[densely dotted,decorate,decoration=snake] (p11t) -- (p11);
        \node (p13) at (4,0) [point,fill=black,label=below:{$p_{1,3}$}] {};
        \node (p13t) at (3.3,0) [point,label=above:{$p^*$}] {};
		\draw[densely dotted,decorate,decoration=snake] (p13t) -- (p13);
        \node (q21t) at (3.3,1.5) [point,label=above:{$\tilde q_{2,1}$}] {};
        \node (p12) at (4,0.8) [point,fill=black,label=right:{$p_{1,2}$}] {};
        \draw[ultra thin] (2,1.5) -- +(0,-0.2);
        \draw[ultra thin] (3.8,0.8) -- +(0.4,0);
        \draw[dashed,->] (2, 1.4) -- (3.9,1.4) node[midway,below]{$\pmb{\alpha_1}$} -- (3.9,0.8);
        \draw[dashed,->] (4.1, 0.8) -- (4.1,0.1) -- (9,0.1) node[midway,above]{$\lambda_1$};
        \node at (2,0.6) {\normalsize $R_1$};
        \draw[thick] (q11) -| (q12);
        \draw[thick] (p11) -| (p12);
        \node at (2,2.4) [point,fill=black,label=below right:{$\pmb{o_1}$}] {};
	\end{tikzpicture}
\end{center}
\caption{The arcs $\beta_1$ and $\alpha_1$.}\label{subfig:arcs2}
\end{figure}%
Noting
that~\eqref{eq:pcp:C} entails
\begin{equation*}
\neg C(a_{1,1} + \dots + a_{1,6}, \ \ s_0 + s_9 + d_0+ \cdots + d_5),
\end{equation*}
we can be sure that $\alpha_1$ lies entirely in the `lower' window,
whence $\beta_1$ crosses the central chord, $\chi_2$
at least once. Let
$o_1$ be the first such point (measured along $\chi_2$ from left to
right).  Finally, let $\lambda_1$ be the segment of $\tilde{\alpha}_1$
between $p_{1,2}$ and $p_{1,3}$, followed by the final segment of
$\lambda_0$ from $p_{1,3}$. Note that $\lambda_1$ lies entirely in
$a_{1,6} + a$, and hence certainly in the region
\begin{equation}\label{eq:r2:astar}
a^* = a_{0,6} + a_{1,6} + a_{2,6} + a.
\end{equation}
%
The region
$S_1$ may now be forgotten. 

By construction, the point $\tilde{q}_{2,1}$ lies in some component of
$\intermediate{a}_{1,3}$, and, from the presence of the `switching'
variable $\intermediate{z}$ in~\eqref{eq:aSeq:a}, that component is
either included in $\intermediate{z}$ or included in
$-\intermediate{z}$. Suppose the latter.  Then we can repeat the above
construction to obtain an arc $\tilde{\beta}_2$ from $\tilde{q}_{2,1}$
to $q^*$, with successive segments $\tilde{\beta}_{2,1}$,
$\beta_{2,2}$, \dots, $\beta_{2,5}$, $\tilde{\beta}_{2,6}$ lying in
the respective regions $a_{1,3} + b_{2,1}$, $b_{2,2}$, \dots,
$b_{2,5}$, $b_{2,6} + b$; further, we can guarantee that $\beta_{2,2}$
contains a point $\tilde{p}_{2,1} \in \intermediate{b}_{2,3}$.  Denote
the last point of $\beta_{2,5}$ by $q_{2,2}$. Also, let $q_{2,1}$ be
the last point of $\tilde{\beta}_2$ lying on $\alpha_1$, and $q_{2,3}$
the first point of $\tilde{\beta}_2$ lying on $\mu_1$.  Again, we let
$\beta_2$ be the segment of $\tilde{\beta}_2$ between $q_{2,1}$ and
$q_{2,2}$; and we let $\mu_2$ be the segment of $\tilde{\beta}_2$ from
$q_{2,1}$ to $q_{2,3}$, followed by the final segment of $\mu_1$ from
$q_{2,3}$.  Note that $\mu_2$ lies in the set $b^*$.  It is easy to
see that the segment of $\tilde{\beta}_2$ from $q_{2,1}$ to $q_{2,3}$
is a cross-cut in $R_1'$ dividing it into regions $S_2$ and $S_2'$, as
shown in Fig.~\ref{subfig:arcs3}.  Indeed,
$\beta_2=\tilde{\beta}_2[q_{2,1},q_{2,2}]$ cannot enter the interior
of the region $R_1$, for, by construction, it can have only one point
of contact with $\alpha_1$, and the constraints~\eqref{eq:pcp:C}
ensure that it cannot intersect any other part of $\delta R_1$.  Since
$q^* \in \inner{a}$ is guaranteed to lie outside $R_1$, we evidently
have that $\beta_2 \subseteq -R_1$. By the
constraints~\eqref{eq:pcp:C}, $\beta_2$ lies in the interior of $R_1'$
except for its first point, which lies on the boundary of $R_1'$;
hence the reversal of $\beta_2$ is an end-cut in $R_1'$.  Similarly,
$\tilde{\beta}_2[q_{2,2},q_{2,3}]$ is an end-cut in $R_1'$ as well,
and thus $\tilde{\beta}_2[q_{2,1},q_{2,3}]$ is a cross-cut in
$R_1'$. This observation having been made, $R_1$ may now be forgotten.
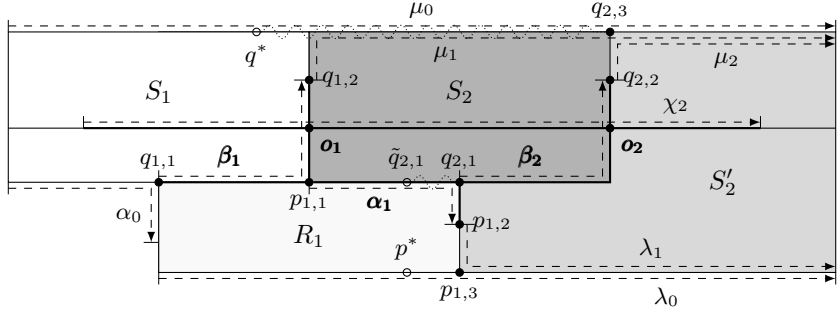
\begin{figure}[hbt]
\begin{center}
	\begin{tikzpicture}[>=latex,point/.style={circle,draw=black,minimum size=1mm,inner sep=0pt},yscale=0.8]\small
        \filldraw[fill=gray!30] (0,0) rectangle +(9,4);
        \filldraw[fill=white,thin] (-2,1.5) rectangle +(4,2.5);
        \filldraw[fill=gray!60] (2,1.5) rectangle +(4,2.5);
        \filldraw[fill=gray!5] (0,0) rectangle +(4,1.5);
        \node (q11) at (0,1.5) [point,fill=black,label=above:{$q_{1,1}$}] {};
        \node (q21) at (4,1.5) [point,fill=black,label=above:{$q_{2,1}$}] {};
        \node (q21t) at (3.3,1.5) [point,label=above:{$\tilde q_{2,1}$}] {};
		\draw[densely dotted,decorate,decoration=snake] (q21t) -- (q21);
        \node (q23) at (6,4) [point,fill=black,label=above:{$q_{2,3}$}] {};
        \node (q13t) at (1.3,4) [point,label=below:{$q^*$}] {};
		\draw[densely dotted,decorate,decoration=snake] (q13t) -- (q23);
        \draw[ultra thin] (0,1.5) -- +(0,0.2);
        \draw[ultra thin] (4,1.5) -- +(0,0.2);
        \draw[ultra thin] (1.8,3.2) -- +(0.4,0);
        \draw[ultra thin] (5.8,3.2) -- +(0.4,0);
        \node (q12) at (2,3.2) [point,fill=black,label=right:{$q_{1,2}$}] {};
        \node (q22) at (6,3.2) [point,fill=black,label=right:{$q_{2,2}$}] {};
        \draw[dashed,->] (0, 1.6) -- (1.9,1.6) node[midway,above]{$\pmb{\beta_1}$} -- (1.9,3.2);
        \draw[dashed,->] (4, 1.6) -- (5.9,1.6) node[midway,above]{$\pmb{\beta_2}$} -- (5.9,3.2);
        \draw[dashed,->] (2.1, 3.2) -- (2.1,3.9) -- (9,3.9) node[near start,below]{$\mu_1$};
        \draw[dashed,->] (6.1, 3.2) -- (6.1,3.8) -- (9,3.8) node[midway,below]{$\mu_2$};
        \draw[ultra thin] (-2,4) --++(0,.2); \draw[ultra thin] (9,4) --++(0,.2);
        \draw[dashed,->] (-2,4.1) -- (9,4.1) node[midway,above]{$\mu_0$};
	\draw[ultra thin] (0,0) --++(0,-.2);\draw[ultra thin] (9,0) --++(0,-.2);
        \draw[dashed,->] (0,-0.1) -- +(9,0) node[near end,below]{$\lambda_0$};
	\draw[ultra thin] (-2,1.5) --++(0,-.2);\draw[ultra thin] (0,0.5)--++(-.2,0);
        \draw[dashed,->] (-2,1.4) -- (-0.1,1.4) -- (-0.1,0.5) node[midway,left]{$\alpha_0$};
        \node at (0,3) {\normalsize $S_1$};
        \node at (4,3) {\normalsize $S_2$};
        \draw (-2,2.4) -- +(11,0);
        \draw[ultra thin] (-1,2.4) -- +(0,0.2);
        \draw[ultra thin] (8,2.4) -- +(0,0.2);
        \draw[thick] (-1,2.4) -- +(9,0);
        \draw[->,dashed] (-1,2.5) -- +(9,0) node[very near end,above] {$\chi_2$};
        \node (p11) at (2,1.5) [point,fill=black,label=below:{$p_{1,1}$}] {};
        \node (p13) at (4,0) [point,fill=black,label=below:{$p_{1,3}$}] {};
        \node (p13t) at (3.3,0) [point,label=above:{$p^*$}] {};
        \node (p12) at (4,0.8) [point,fill=black,label=right:{$p_{1,2}$}] {};
        \draw[ultra thin] (2,1.5) -- +(0,-0.2);
        \draw[ultra thin] (3.8,0.8) -- +(0.4,0);
        \draw[dashed,->] (2, 1.4) -- (3.9,1.4) node[midway,below]{$\pmb{\alpha_1}$} -- (3.9,0.8);
        \draw[dashed,->] (4.1, 0.8) -- (4.1,0.1) -- (9,0.1) node[midway,above]{$\lambda_1$};
        \node at (2,0.6) {\normalsize $R_1$};
        \node at (7.5,1.5){\normalsize $S_2'$};		
        \draw[thick] (q11) -| (q12);
        \draw[thick] (p11) -| (p12);
        \draw[thick] (q21) -| (q22);
        \node at (2,2.4) [point,fill=black,label=below right:{$\pmb{o_1}$}] {};
        \node at (6,2.4) [point,fill=black,label=below right:{$\pmb{o_2}$}] {};
	\end{tikzpicture}
\end{center}
\caption{The arc $\beta_2$.}\label{subfig:arcs3}
\end{figure}%

Symmetrically, we construct the arc $\tilde{\alpha}_2$ in $b_{2,3} +
a_{2,1} + \cdots + a_{2,6} + a$, and points $p_{2,1}$, $p_{2,2}$,
$p_{2,3}$, together with the arcs $\alpha_2$ and $\lambda_2$. Again,
we know from~\eqref{eq:pcp:C} that $\alpha_2$ lies entirely in the
`lower' window, whence $\beta_2$ must cross the central chord,
$\chi_2$, at least once. Let $o_2$ be the first such point (measured
along $\chi_2$ from left to right); see Fig.~\ref{subfig:arcs3}.

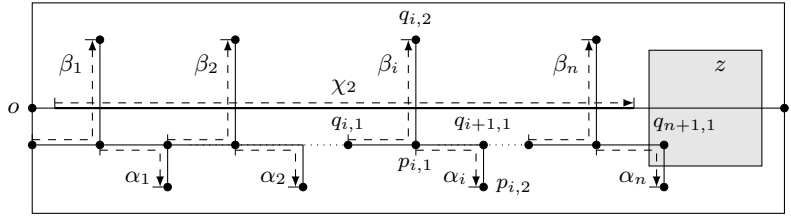
\begin{figure}[ht]
\begin{center}
	\begin{tikzpicture}[>=latex,point/.style={circle,draw=black,minimum size=1mm,inner sep=0pt},yscale=0.7]\small
        \filldraw[fill=gray!20] (8.2,0.9) rectangle +(1.5,2.2);
        \node at (9.15,2.8) {$z$};
        \draw (0,0) rectangle +(10,4);
        \draw (0,2) -- +(10,0);
        \node (o) at (0,2) [point,fill=black,label=left:{$o$}] {};
        \node (op) at (10,2) [point,fill=black] {};
        \draw[ultra thin] (0.3,2) -- +(0,0.2);
        \draw[ultra thin] (8,2) -- +(0,0.2);
        \draw[->,dashed] (0.3,2.1) -- +(7.7,0) node[midway,above] {$\chi_2$};
        \draw[thick] (0.3,2) -- +(7.7,0);
        \foreach \i/\x/\l/\f in  {1/0/0/2,2/1.8/0,i/4.2/1,n/6.6/0}
        {
			\ifnum \l=0
                \node (q\i 1) at (\x,1.3) [point,fill=black] {};
                \node (p\i 1) at ($(q\i 1)+(0.9,0)$) [point,fill=black] {};
                \node (q\i 2) at ($(q\i 1)+(0.9,2)$) [point,fill=black] {};
                \node (p\i 2) at ($(q\i 1)+(1.8,-0.8)$) [point,fill=black] {};
            \fi
	\ifnum \l=1
		\ifnum \f=2
			\node (q\i 1) at (\x,1.3) [point,fill=black,label=left:{\footnotesize $q_{\i,1}$}] {};
		\else
			\node (q\i 1) at (\x,1.3) [point,fill=black,label=above:{\footnotesize $q_{\i,1}$}] {};
		\fi
                \node (p\i 1) at ($(q\i 1)+(0.9,0)$) [point,fill=black,label=below:{\footnotesize $p_{\i,1}$}] {};
                \node (q\i 2) at ($(q\i 1)+(0.9,2)$) [point,fill=black,label=above:{\footnotesize $q_{\i,2}$}] {};
                \node (p\i 2) at ($(q\i 1)+(1.8,-0.8)$) [point,fill=black,label=right:{\footnotesize $p_{\i,2}$}] {};
            \fi
            \draw (q\i 1) -| (q\i 2);
            \draw (p\i 1) -| (p\i 2);
            \draw[->,dashed] ($(q\i 1)+(0,0.1)$) -| ($(q\i 2)+(-0.1,0)$) node[very near end,left] {$\beta_\i$};
            \draw[->,dashed] ($(p\i 1)+(0,-0.1)$) -| ($(p\i 2)+(-0.1,0)$) node[very near end,left] {$\alpha_\i$};
            \draw[ultra thin] (p\i 2) -- +(-0.2,0);
            \draw[ultra thin] (q\i 1) -- +(0,0.2);
            \draw[ultra thin] (p\i 1) -- +(0,-0.2);
            \draw[ultra thin] (q\i 2) -- +(-0.2,0);
        }
        \draw[dotted] (2.1,1.3) -- (qi1);
        \draw[dotted] (5.4,1.3) -- (qn1);
        \node (qip 1) at (6,1.3) [point,fill=black,label=above:{\footnotesize $q_{i+1,1}$}] {};
        \node (qip 1) at (8.4,1.3) [point,fill=black,label=above right:{\footnotesize\hspace*{-1em}$q_{n+1,1}$}] {};
    \end{tikzpicture}
\end{center}
\caption{The sequence of pairs of arcs $(\beta_i,\alpha_i)$.}
\label{fig:arcs2}
\end{figure}

This process continues, generating arcs $\beta_i \subseteq
a_{\md{i-1},3} + b_{\md{i},1} + \cdots + b_{\md{i},5}$ and $\alpha_i
\subseteq b_{\md{i},3} + a_{\md{i},1} + \cdots + a_{\md{i},5}$, as
long as $\alpha_i$ contains a point $\tilde{q}_{i+1,1} \in (-\intermediate{z})$. That
we eventually reach a value $i = n$ for which no such point exists
follows from~\eqref{eq:pcp:C}. For the conjuncts $\neg C(b_{i,j},
d_k)$, for $j \neq 5$, together entail $o_i \in b_{\md{i},5}$, for
every $i$ such that $\beta_i$ is defined; and these points cycle on
$\chi_2$ through the regions $b_{0,5}$, $b_{1,5}$ and $b_{2,5}$. If
there were infinitely many $\beta_i$, the $o_i$ would have an
accumulation point, lying in all three regions, contradicting, say,
$\neg C(b_{0,5},b_{1,5})$.  The resulting sequence of arcs and points
is shown, schematically, in Fig.~\ref{fig:arcs2}. It follows that
the final arc $\alpha_n$ contains a point $q_{n+1,1}\in \intermediate{z}$.

We
finish this stage in the construction by `re-packaging' the pairs of arcs
$(\beta_i,\alpha_i)$.  Specifically, for all $1 \leq i \leq n$, let
$\zeta_i$ be the initial segment of $\beta_i$ up to the point
$p_{i,1}$ followed by the initial segment of $\alpha_i$ up to the
point $q_{i+1,1}$; let $\eta_i$ be the final segment of $\beta_i$ from
the point $p_{i,1}$; and let $\kappa_i$ be the final segment of
$\alpha_i$ from the point $q_{i+1,1}$:
\begin{align*}
\zeta_i = \beta_i[q_{i,1},p_{i,1}]\alpha_i[p_{i,1},q_{i+1,1}],\quad
 \eta_i = \beta_i[p_{i,1},q_{i,2}],\quad
\kappa_i = \alpha_i[q_{i+1,1}, p_{i,2}]
\end{align*}
(see Fig.~\ref{fig:arcs2repackaged}).
\begin{figure}[ht]
\begin{center}
	\begin{tikzpicture}[>=latex,point/.style={circle,draw=black,minimum size=1mm,inner sep=0pt},yscale=0.7]\small
        \draw (0,0) rectangle +(10,4);
        \draw (0,2) -- +(10,0);
        \foreach \i/\x/\l/\f in  {1/0/0/2,2/1.8/0,i/4.2/1,n/6.6/0}
        {
			\ifnum \l=0
                \node (q\i 1) at (\x,1.3) [point,fill=black] {};
                \node (p\i 1) at ($(q\i 1)+(0.9,0)$) [point,fill=black] {};
                \node (q\i 2) at ($(q\i 1)+(0.9,2.2)$) [point,fill=black] {};
                \node (p\i 2) at ($(q\i 1)+(1.8,-0.9)$) [point,fill=black] {};
            \fi
	\ifnum \l=1
		\ifnum \f=2
			\node (q\i 1) at (\x,1.3) [point,fill=black,label=left:{\footnotesize $q_{\i,1}$}] {};
		\else
			\node (q\i 1) at (\x,1.3) [point,fill=black,label=above:{\footnotesize $q_{\i,1}$}] {};
		\fi
                \node (p\i 1) at ($(q\i 1)+(0.9,0)$) [point,fill=black,label=below:{}] {};
                \node (q\i 2) at ($(q\i 1)+(0.9,2.2)$) [point,fill=black,label=right:{\footnotesize $q_{\i,2}$}] {};
                \node (p\i 2) at ($(q\i 1)+(1.8,-0.9)$) [point,fill=black,label=left:{\footnotesize $p_{\i,2}$}] {};
            \fi
            \draw (q\i 1) -| (q\i 2);
            \draw (p\i 1) -| (p\i 2);
            \draw[->,dashed] ($(p\i 1)+(-0.1,0)$) -- ($(q\i 2)+(-0.1,0)$) node[very near end,left] {$\eta_\i$};
            \draw[->,dashed] ($(q\i 1)+(0,-0.1)$) -- node[midway, below] {$\zeta_\i$} ($(q\i 1)+(1.8,-0.1)$);
            \draw[->,dashed] ($(q\i 1)+(1.9,0)$) -- ($(p\i 2)+(0.1,0)$) node[near end,right] {$\kappa_\i$};
            \draw[ultra thin] ($(q\i 1)+(1.8,0)$) -- +(0.2,0);
            \draw[ultra thin] (p\i 2) -- +(0.2,0);
            \draw[ultra thin] (q\i 1) -- +(0,-0.2);
            \draw[ultra thin] (q\i 2) -- +(-0.2,0);
        }
        \draw[dotted] (2.1,1.3) -- (qi1);
        \draw[dotted] (5.4,1.3) -- (qn1);
        \node (qip 1) at (6,1.3) [point,fill=black,label=above:{\footnotesize $q_{i+1,1}$}] {};
        \node (qip 1) at (8.4,1.3) [point,fill=black] {}; 
    \end{tikzpicture}
\end{center}
\caption{The sequence of triples of arcs $(\zeta_i,\eta_i,\kappa_i)$
formed by `re-packaging' $(\beta_i,\alpha_i)$ from Fig.~\ref{fig:arcs2}.}
\label{fig:arcs2repackaged}
\end{figure}
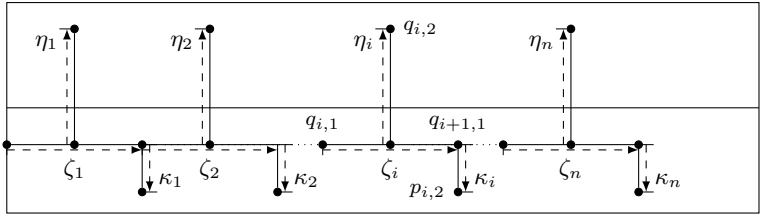
Defining, for $0 \leq i < 3$,
\begin{align}
\label{eq:r2:ai}
r_i & = a_{\md{i-1},3} + b_{i,1} + \cdots + b_{i,4} + a_{i,1} + \cdots + a_{i,4},\\
\label{eq:r2:bi}
b_i & = b_{i,2} + \cdots + b_{i,5},\\
\label{eq:r2:aai}
a_i & = a_{i,2} + \cdots + a_{i,5},
\end{align}
the constraints~\eqref{eq:pcp:C} guarantee that, for $1 \leq i \leq
n$,
\begin{equation*}
\zeta_i \subseteq r_{\md{i}},\qquad \eta_i \subseteq b_{\md{i}} \quad\text{and}\quad \kappa_i \subseteq a_{\md{i}}.
\end{equation*}
Observe that the arcs $\zeta_i$ are located entirely in the `lower
window,' and that each arc $\eta_i$ connects $\zeta_i$ to some point
$q_{i,2}$, which in turn is connected to $q^* \in \lambda_0$
by an arc in $b^*$.

\smallskip

\noindent
\textbf{Stage 3.}  We now repeat Stage~2 symmetrically, with the
`upper' and `lower' windows exchanged. Let $\tseq{a}'_{i,j}$,
$\tseq{b}'_{i,j}$ be 3-region variables (with indices in the same
ranges as for $\tseq{a}_{i,j}$, $\tseq{b}_{i,j}$). Let $\tseq{a}' =
\tseq{b}$, $\tseq{b}' = \tseq{a}$.
The formulas
\begin{align}
\tag{\ref{eq:aSeq1}$'$}
& (s'_3 \leq \intermediate{a}'_{0,3}),\\
\tag{\ref{eq:aSeq:b}$'$}
& \bigwedge_{i = 0}^2\stack_{\intermediate{z}}(\tseq{a}'_{\md{i-1},3}, \tseq{b}'_{i,1}, \dots, \tseq{b}'_{i,6}, \tseq{b}'),\\
\tag{\ref{eq:aSeq:a}$'$}
& \bigwedge_{i = 0}^2\stack(\tseq{b}'_{i,3}, \tseq{a}'_{i,1}, \dots, \tseq{a}'_{i,6}, \tseq{a}'),\\
%
\tag{\ref{eq:connectbd}$'$}
& c(b'_{1,5} + d_3)
\end{align}
then establish sequences of $n'$ triples of arcs $(\zeta'_i,
\eta'_i,\kappa'_i)$ satisfying
\begin{equation*}
\zeta_i' \subseteq r_{\md{i}}',\qquad
\eta_i' \subseteq b_{\md{i}}'\quad\text{and}\quad \kappa_i\subseteq a_{\md{i}}',
\end{equation*}
for $1 \leq i \leq n'$, where the $r'_i$, $b'_i$ and $a'_i$ are
defined as in \eqref{eq:r2:ai}--\eqref{eq:r2:aai} but with the primed
variables.  The arcs $\zeta'_i$ are located entirely in the `upper
window', and each arc $\eta'_i$ connects $\zeta'_i$ to a point
$q'_{i,2}$, which in turn is connected to a point ${q^*}'$ by an arc
in the region
${b^*}' = b'_{0,6} + b'_{1,6} + b'_{2,6} + b'$.

\medskip

\noindent
\textbf{Stage 4.}  Our next task is to write constraints to ensure
that $n = n'$, and that, furthermore, each $\eta_i$ (also each
$\eta'_i$) connects $\zeta_i$ to $\zeta'_i$, for $1\leq i\leq
n$. From~\eqref{eq:znots}, the only arc depicted in
Fig.~\ref{fig:arcs0} that $\ti{z}$ may intersect is
$\chi_3$. Recalling that $\zeta_n$ and $\zeta'_{n'}$ contain points
$q_{n+1,1}$ and $q'_{n'+1,1}$, respectively, both lying in
$\intermediate{z} \subseteq \ti{z}$, the constraint
\begin{equation}
c(\intermediate{z}) 
\end{equation}
ensures that $q_{n+1,1}$ and $q'_{n'+1,1}$ may be joined by an arc,
say $\zeta^*$, lying in $\ti{z}$, and also lying entirely in the upper
and lower windows, crossing the chord $\chi_1\chi_2\chi_3$ only in
$\chi_3$.  Without loss of generality, we may assume that $\zeta^*$
contacts each of $\zeta_n$ and $\zeta'_{n'}$ in just one
point. Bearing in mind that the formulas~\eqref{eq:pcp:C} force
$\eta_n\subseteq b_0 + b_1 + b_2$ and $\eta'_{n'} \subseteq b_0' +
b_1' + b_2'$ to cross the chord $\chi_1\chi_2\chi_3$ in its central
section, $\chi_2$, and bearing in mind~\eqref{eq:znots}, we see that
the following constraint ensures that $\zeta^*$ is as shown in
Fig.~\ref{fig:arcZeta}:
\begin{align}
\label{eq:zeta}
z\cdot(b^* + b_0 + b_1 + b_2 + 
{b^*}' + b'_0 + b'_1 + b'_2) = 0.
\end{align}
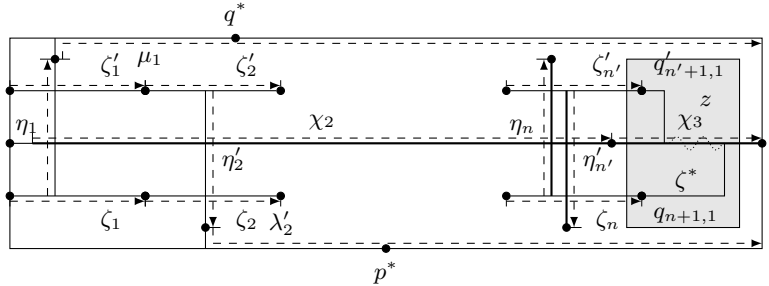
\begin{figure}[ht]
\begin{center}
	\begin{tikzpicture}[>=latex,point/.style={circle,draw=black,minimum size=1mm,inner sep=0pt},yscale=0.7]\small
        \filldraw[fill=gray!20] (8.2,0.4) rectangle +(1.5,3.2);
        \node at (9.25,2.8) {$z$};
        \draw (0,0) rectangle +(10,4);
        \draw (0,2) -- +(10,0);
        \node (o) at (0,2) [point,fill=black] {};
        \node (op) at (10,2) [point,fill=black] {};
        \node at (8,2) [point,fill=black] {};
        \draw[ultra thin] (0.3,2) -- +(0,0.2);
        \draw[ultra thin] (8,2) -- +(0,0.2);
        \draw[->,dashed] (0.3,2.1) -- +(7.7,0) node[midway,above] {$\chi_2$};
        \draw[thick] (0.3,2) -- +(9.7,0);
        \draw[->,dashed] (8.1,2.1) -- (10,2.1) node[midway,above] {$\chi_3$};
        \foreach \i/\x/\l in  {1/0/1,2/1.8/1,n/6.6/0}
        {
            \node (q\i 1) at (\x,1) [point,fill=black] {};
            \coordinate (p\i 1) at ($(q\i 1)+(0.6,0)$);
            \coordinate (q\i 2) at ($(q\i 1)+(0.6,2.6)$);
            \node (q\i 1n) at ($(q\i 1)+(1.8,0)$) [point,fill=black] {};
            \draw (q\i 1) -- (q\i 1n);
            \draw[ultra thin] (q\i 1n) -- +(0,-0.2);
            \draw[ultra thin] (q\i 1) -- +(0,-0.2);
            \node (q\i 1p) at (\x,3) [point,fill=black] {};
            \coordinate (p\i 1p) at ($(q\i 1p)+(0.8,0)$);
            \coordinate (q\i 2p) at ($(q\i 1p)+(0.8,-2.6)$);
            \node (q\i 1np) at ($(q\i 1p)+(1.8,0)$) [point,fill=black] {};
            \draw (q\i 1p) -- (q\i 1np);
            \draw[ultra thin] (q\i 1np) -- +(0,0.2);
            \draw[ultra thin] (q\i 1p) -- +(0,0.2);
            \ifnum \l=1
                \draw[->,dashed] ($(q\i 1)+(0,-0.1)$) -- ($(q\i 1n)+(0,-0.1)$) node[near end,below] {$\zeta_\i$};
                \draw[->,dashed] ($(q\i 1p)+(0,0.1)$) -- ($(q\i 1np)+(0,0.1)$) node[near end,above] {$\zeta'_\i$};
            \fi
        }
        \foreach \i in {n}
        {
            \draw[->,dashed] ($(q\i 1)+(0,-0.1)$) -- ($(q\i 1n)+(0,-0.1)$) node[near end,below] {$\zeta_n$};
            \draw[->,dashed] ($(q\i 1p)+(0,0.1)$) -- ($(q\i 1np)+(0,0.1)$) node[near end,above] {$\zeta'_{n'}$};
            \node at (q\i 2) [point,fill=black] {};
            \draw[thick] (p\i 1) -- (q\i 2);
            \draw[ultra thin] (q\i 2) -- +(-0.2,0);
            \draw[->,dashed] ($(p\i 1)+(-0.1,0)$) -- ($(q\i 2)+(-0.1,0)$) node [midway,left] {$\eta_n$};
            \node at (q\i 2p) [point,fill=black] {};
            \draw[thick] (p\i 1p) -- (q\i 2p);
            \draw[ultra thin] (q\i 2p) -- +(0.2,0);
            \draw[->,dashed] ($(p\i 1p)+(0.1,0)$) -- ($(q\i 2p)+(0.1,0)$) node [midway,right] {$\eta'_{n'}$};
        }
        \node at (qn1n) [point,fill=black,label=below right:{$q_{n+1,1}$}] {};
        \node at (qn1np) [point,fill=black,label=above right:{$q'_{n'+1,1}$}] {};
        \draw (qn1n) -| ++(1.1,1);
        \draw (qn1np) -| ++(0.3,-1);
		\draw[densely dotted,decorate,decoration=snake] ($(qn1n)+(1.1,1)$) -- ($(qn1np)+(0.3,-1)$);
        \node at (9,1.3) {$\zeta^*$};
        \foreach \i in  {1}
        {
            \node at (q\i 2) [point,fill=black] {};
            \draw (p\i 1) -- (q\i 2) -- ++(0,0.4);
            \draw[ultra thin] (q\i 2) -- +(-0.2,0);
            \draw[ultra thin] (q\i 2) -- +(0.2,0);
            \draw[->,dashed] ($(p\i 1)+(-0.1,0)$) -- ($(q\i 2)+(-0.1,0)$) node [midway,left] {$\eta_\i$};
            \draw[->,dashed] ($(q\i 2)+(0.1,0)$) -- ++(0,0.3) -- ++(9.3,0) node [very near start,below] {$\mu_\i$};
        }
        \foreach \i in  {2}
        {
            \node at (q\i 2p) [point,fill=black] {};
            \draw (p\i 1p) -- (q\i 2p) -- ++(0,-0.4);
            \draw[ultra thin] (q\i 2p) -- +(0.2,0);
            \draw[->,dashed] ($(p\i 1p)+(0.1,0)$) -- ($(q\i 2p)+(0.1,0)$) node [midway,right] {$\eta'_\i$};
            \draw[->,dashed] ($(q\i 2p)+(0.1,0)$) -- ++(0,-0.3) -- ++(7.3,0) node [very near start,above] {$\lambda'_\i$};
        }
    \node at (5,0) [point,fill=black,label=below:{$p^*$}] {};
    \node at (3,4) [point,fill=black,label=above:{$q^*$}] {};
    \end{tikzpicture}
\end{center}
\caption{The arc $\zeta^*$.}\label{fig:arcZeta}
\end{figure}

Now consider the arc $\eta_1$. Recalling that $\eta_1$ crosses
$\chi_2$ and connects $\zeta_1$ to some point $q_{1,2}$, which in turn
is connected to the point $q^*$ by an arc in $b^*$,
we see by inspection of Fig.~\ref{fig:arcZeta} that~\eqref{eq:zeta}
together with
\begin{align}
& \bigwedge_{i = 0}^2 \neg C(r'_i, b^*)
\end{align}
forces $\eta_1$ to cross one of the arcs $\zeta'_{j'}\subseteq
r'_{\md{j'}}$, for $1 \leq j' \leq n'$; and the constraints
\begin{align}
\label{eq:r2:aip:bj}
& \bigwedge_{i = 0}^2 \neg C(r'_i,\ \ b_{\md{i-1}} + b_{\md{i+1}})
\end{align}
ensure that $j' \equiv
1 $ modulo 3.  Now suppose $j' \geq 4$. We write the constraints
\begin{align}
\tag{\ref{eq:r2:aip:bj}$'$}\label{eq:abPrime}
& \bigwedge_{i = 0}^2 \neg C(b_i',\ \ r_{\md{i-1}} + r_{\md{i+1}}),\\
\label{eq:bbPrime}
& \bigwedge_{i = 0}^2 \neg C(b_i',\ \ b_{\md{i-1}} + b_{\md{i+1}}).
\end{align}
The arc $\eta'_2$
must connect $\zeta'_2$ to the point $q'_{2,2}$, which in turn is connected to the point $p^*$ on the bottom
edge of the lower window by an arc in ${b^*}'$, which is now impossible without $\eta'_2\subseteq b_2'$
crossing either $\zeta_1\subseteq r_1$ or $\eta_1\subseteq b_1$---both forbidden
by~\eqref{eq:abPrime}--\eqref{eq:bbPrime}.  Thus, $\eta_1$ intersects
$\zeta'_j$ if and only if $j = 1$. Symmetrically, $\eta'_1$ intersects
$\zeta_j$ if and only if $j = 1$.  And the reasoning can now be
repeated for $\eta_2, \eta_2', \eta_3, \eta_3', \dots$, leading
to the 1--1 correspondence depicted in
Fig.~\ref{fig:arcCorrespondence}.
In particular, we are guaranteed that $n = n'$.
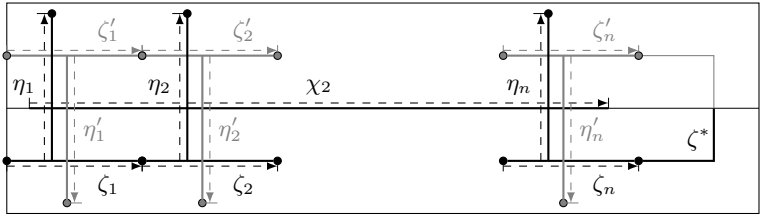
\begin{figure}[ht]
\begin{center}
	\begin{tikzpicture}[>=latex,point/.style={circle,draw=black,minimum size=1mm,inner sep=0pt},yscale=0.7]\small
        \draw (0,0) rectangle +(10,4);
        \draw[ultra thin] (0,2) -- +(10,0);
        \draw[ultra thin] (0.3,2) -- +(0,0.2);
        \draw[ultra thin] (8,2) -- +(0,0.2);
        \draw[->,dashed] (0.3,2.1) -- +(7.7,0) node[midway,above] {$\chi_2$};
        \draw[thick] (0.3,2) -- +(7.7,0);
%
        \foreach \i/\x/\l in  {1/0/0,2/1.8/0,n/6.6/0}
        {
            \node (q\i 1) at (\x,1) [point,fill=black] {};
            \coordinate (p\i 1) at ($(q\i 1)+(0.6,0)$);
            \coordinate (q\i 2) at ($(q\i 1)+(0.6,2.8)$);
            \node (q\i 1n) at ($(q\i 1)+(1.8,0)$) [point,fill=black] {};
            \draw[thick] (q\i 1) -- (q\i 1n);
            \draw[->,dashed] ($(q\i 1)+(0,-0.1)$) -- ($(q\i 1n)+(0,-0.1)$) node[near end,below] {$\zeta_\i$};
            \draw[ultra thin] (q\i 1n) -- +(0,-0.2);
            \draw[ultra thin] (q\i 1) -- +(0,-0.2);
            \node (q\i 1p) at (\x,3) [point,fill=gray] {};
            \coordinate (p\i 1p) at ($(q\i 1p)+(0.8,0)$);
            \coordinate (q\i 2p) at ($(q\i 1p)+(0.8,-2.8)$);
            \node (q\i 1np) at ($(q\i 1p)+(1.8,0)$) [point,fill=gray] {};
            \draw[thick,gray] (q\i 1p) -- (q\i 1np);
            \draw[gray,->,dashed] ($(q\i 1p)+(0,0.1)$) -- ($(q\i 1np)+(0,0.1)$) node[near end,above] {$\zeta'_\i$};
            \draw[gray,ultra thin] (q\i 1np) -- +(0,0.2);
            \draw[gray,ultra thin] (q\i 1p) -- +(0,0.2);
            \node at (q\i 2) [point,fill=black] {};
            \draw[thick] (p\i 1) -- (q\i 2);
            \draw[ultra thin] (q\i 2) -- +(-0.2,0);
            \draw[->,dashed] ($(p\i 1)+(-0.1,0)$) -- ($(q\i 2)+(-0.1,0)$) node [midway,left] {$\eta_\i$};
            \node at (q\i 2p) [point,fill=gray] {};
            \draw[thick,gray] (p\i 1p) -- (q\i 2p);
            \draw[gray,ultra thin] (q\i 2p) -- +(0.2,0);
            \draw[gray,->,dashed] ($(p\i 1p)+(0.1,0)$) -- ($(q\i 2p)+(0.1,0)$) node [midway,right] {$\eta'_\i$};
        }
        \draw[thick] (qn1n) -| ++(1,1);
        \draw[gray] (qn1np) -| ++(1,-1);
        \node at (9.2,1.4) {$\zeta^*$};
    \end{tikzpicture}
\end{center}
\caption{The 1--1 correspondence between the $\zeta_i$ and the
  $\zeta'_i$ established by the $\eta_i$ and the $\eta'_i$.}
\label{fig:arcCorrespondence}
\end{figure}

\medskip

\noindent
\textbf{Stage 5.}  Recall the given PCP-instance, $\pcpW = (\pcpw,
\pcpw')$ over alphabets $T$ and $U$. In the sequel, we use the
standard imagery of `tiles', where each tile $t \in T$ has an `upper
string', $\pcpw'(t) \in U^*$ and a `lower string', $\pcpw(t) \in
U^*$. Thus, the problem is to determine whether there is some
non-empty sequence of tiles such that the concatenated upper and lower
strings both spell out the same string in $U^*$. We shall label the
arcs $\zeta_1, \dots, \zeta_n$ so as to define a string $\tau \in T^*$
(with $|\tau| = m \leq n$); likewise we shall label the arcs
$\zeta'_1, \dots, \zeta'_n$ so as to define another string $\tau' \in
T^*$ (with $|\tau'| = m' \leq n$). Then the arcs $\eta_1,\dots,\eta_n$
will be labelled with the regions in $\vec{u}$, so to define a string
$\upsilon \in U^*$, with $|\upsilon| = n$. We shall then add conjuncts
to $\psi_\pcpW$ ensuring $\pcpw(\tau) = \pcpw'(\tau') = \upsilon$ and
$\tau = \tau'$, which will guarantee that $\pcpW$ is positive.

For all $1 \leq h \leq |T|$, $1 \leq \ell \leq |\pcpw(t_h)|$
and $0\leq i<3$, let $p_{h,\ell}$ be a fresh variable, and let
these variables be ordered in some way as the tuple $\vec{p}$.
As in the proof of Theorem~\ref{theo:cBCcn}, we think of $p_{h,\ell}$
as standing for the $\ell$th position in the string $\pcpw(t_h)$,
where $t_h\in T$. We use $\vec p$ to label the components of $r_i$, $0\leq i<3$,
but since the $r_i$ are not pairwise disjoint, we require a copy of the variables
$\vec p$ for each $i$. Hence, for all $1 \leq h \leq |T|$,
$1 \leq \ell \leq |\pcpw(t_h)|$ and $0\leq i<3$, let $p_{h,\ell}^i$ be a fresh variable,
and let $\vec p_i$ be an ordering of the variables with superscript $i$.
Consider the constraints
\begin{multline}
	\bigwedge_{i= 0}^2\bigl((r_i=\sum_{h = 1}^{|T|}\sum_{\ell= 1}^{|\pcpw(t_h)|} p_{h,\ell}^i)\ \wedge\
	\colourComp(r_i;\vec{p}_i)\bigr) \ \wedge \ 
	\bigwedge_{h = 1}^{|T|}\bigwedge_{\ell = 1}^{|\pcpw(t_h)|}(p_{h,\ell}=\sum_{i=0}^2p_{h,l}^i).
\end{multline}
The first conjunct ensures that each arc $\beta_{i,3} \subseteq
b_{\md{i},3}$ ($1 \leq i \leq n$) is included in exactly one of the
regions $\vec{p}_{\md{i}}$ and is disjoint from the rest of the regions in
$\vec p_{\md{i}}$ and all the regions in $\vec p_{\md{i-1}}$ and $\vec p_{\md{i+1}}$;
the second conjunct then ensures that $\zeta_i$ is contained in exactly one of the
regions $\vec p$, and that $\beta_{i,3}$ is disjoint from the rest of the regions in $\vec p$.
Note that the $\vec{p}$ do not actually form a partition, because
they cannot be made disjoint; nevertheless, we can think of the $\vec{p}$ as `labels' for arcs $\zeta_i$.
The regions in $\vec p_0$, $\vec p_1$ and $\vec p_2$ can now be forgotten.

Next, we organize the arcs $\zeta_i$ into (contiguous) blocks, $E_1,
\dots, E_m$ such that, in the $j$th block, $E_j$, the sequence of
labels reads $p_{h,1}, \dots, p_{h,|\pcpw(t_h)|}$, for some fixed $1
\leq h \leq |T|$. This amounts to insisting that: ({\em i}) the very
first arc, $\zeta_1$, must be labelled with $p_{h,1}$ for some $h$;
({\em ii}) if $\zeta_i$ ($i < n$) is labelled with $p_{h,\ell}$, where
$\ell < |\pcpw(t_h)|$, then the next arc, namely $\zeta_{i+1}$, must
be labelled with the next position in $\pcpw(t_h)$, namely
$p_{h,\ell+1}$; ({\em iii}) if $\zeta_i$ ($i < n$) is labelled with
the final position of $\pcpw(t_h)$, then the next arc must be labelled
with the initial position of some possibly different word
$\pcpw(t_{h'})$; and ({\em iv}) $\zeta_n$ must be labelled with the
final position of some word $\pcpw(t_h)$. To do this we simply
write:
\begin{align}
 &\bigwedge_{h=1}^{|T|}\bigwedge_{\ell = 2}^{|\pcpw(t_h)|} \neg C(p_{h,\ell}, \ s_3),\\
& \bigwedge_{i = 0}^2\bigwedge_{h = 1}^{|T|}\bigwedge_{\ell=1}^{|\pcpw(t_h)-1|}
\neg C(r_i \cdot p_{h,\ell},
       r_{\md{i+1}} \cdot (-r_{i}) \cdot (\sum_{\ell'\ne \ell+1} p_{h,\ell'} + \sum_{h' \ne h}\sum_{\ell'=1}^{|\pcpw(t_{h'})|}p_{h',\ell'})),
\\
 &\bigwedge_{i = 0}^2\bigwedge_{h = 1}^{|T|}\bigwedge_{h' = 1}^{|T|}\bigwedge_{\ell = 2}^{|\pcpw(t_{h'})|} \neg C(r_i \cdot p_{h,|\pcpw(t_h)|}, \ \ r_{\md{i+1}} \cdot (-r_{i}) \cdot p_{h',\ell}),\\
 &\bigwedge_{h = 1}^{|T|}\bigwedge_{\ell = 1}^{|\pcpw(t_h)|-1} \neg C(p_{h,\ell}, \ z).
\end{align}
Supposing the arcs of $j$th block $E_j$ to have labels reading
$p_{h,1}, \dots, p_{h,|\pcpw(t_h)|}$ (for some fixed $h$), then, we
write $h_j$ to denote the common subscript $h$.  The sequence of
indices $h_1,\dots,h_m$ corresponding to the successive blocks thus
defines a word $\tau = t_{h_1}\cdots t_{h_m}\in T^*$.

Using corresponding formulas, we label the arcs $\zeta'_i$ ($1 \leq i
\leq n$) with the tuple $\vec{p}\,'$ of variables $p'_{h,\ell}$, for
$1 \leq h \leq |T|$ and $1 \leq \ell\leq |\pcpw'(t_h)|$, so that, in
any satisfying assignment over $\RC(\R^2)$, every arc $\zeta'_i$ is
labeled with exactly one of the regions $\vec p\,'$ and $\beta'_{i,3}\subseteq b_{i,3}$
is disjoint from the rest of the regions in $\vec p\,'$. Further, we
can ensure that these labels are organized into (say) $m'$ contiguous
blocks, $E'_1, \dots, E'_{m'}$ such that in the $j$th block, $E_j'$,
the sequence of labels reads $p'_{h,1}, \dots, p'_{h,|\pcpw'(t_h)|}$,
for some fixed $h$. Again, writing $h'_j$ for the common value of $h$,
the sequence of indices $h'_1, \dots, h'_{m'}$ corresponding to the
successive blocks defines a word $\tau' = t_{h'_1}\cdots t_{h'_{m'}}
\in T^*$.

Now, the constraints
\begin{align}
&\partition(\vec{u}) \ \ \land \ \ \bigwedge_{i=0}^2
\colourComp(b_i; \ \vec{u})
\end{align}
ensure that, in any satisfying assignment over $\RC(\R^2)$, every arc
$\eta_i\subseteq b_{\md{i}}$, for $1 \leq i \leq n$, is included in (`labelled with')
exactly one of the regions in $\vec{u}$, so that the sequence of arcs
$\eta_1, \dots, \eta_n$ defines a string $\upsilon \in U^*$, with
$|\upsilon| = n$.

Securing $\pcpw(\tau) = \pcpw'(\tau') = \upsilon$ is easy. The
constraints
\begin{align}
& \bigwedge_{h = 1}^{|T|}\bigwedge_{\ell = 1}^{|\pcpw(t_h)|} \bigwedge_{\begin{subarray}{c}u_i \text{ is not the $\ell$th}\\\text{letter of $\pcpw(t_h)$}\end{subarray}} \neg C(u_i, \ p_{h,\ell})
\end{align}
ensure that, since $\eta_i$ intersects $\zeta_i$, for all $1 \leq i \leq
n$, the string $\upsilon \in U^*$ defined by the arcs $\eta_i$
must be identical to the string $\pcpw(t_{h_1}) \cdots
\pcpw(t_{h_m})$. But this is just to say that $\upsilon = \pcpw(\tau)$.
The equation $\upsilon = \pcpw'(\tau')$ is obtained similarly.

\medskip

\noindent
\textbf{Stage 6.} In the foregoing stages, we assembled conjuncts
of $\psi_\pcpW$ in such a way that,
given any satisfying assignment for $\psi_\pcpW$, we can construct
sequences of labelled arcs defining words $\upsilon \in U^*$ and
$\tau, \tau' \in T^*$ with $\pcpw(\tau) = \pcpw'(\tau') = \upsilon$, as described above.  In this stage, we add more
conjuncts to $\psi_\pcpW$ to enforce the equation $\tau =
\tau'$. This shows that, if $\psi_\pcpW$ is satisfiable over
$\RC(\R^2)$, then $\pcpW$ is positive.

In particular, it remains to show that $m = m'$ and that $h_j =
h'_j$, for all $1 \leq j \leq m$. To do so, we re-use the techniques
encountered in Stage~4.  We first introduce a new pair of variables,
$f_0$, $f_1$, which we refer to as `block colours', and with which we
label the arcs $\zeta_i$. Again, since the regions $r_i$
overlap, we additionally require regions $f_0^i$ and $f_1^i$,
for $0\leq i<3$. Consider the constraints:
\begin{equation}
\bigwedge_{i=0}^2 \bigl((r_i=f_0^i+f_1^i) \ \land\ \colourComp(r_i; \ f_0^i, f_1^i)\bigr) \quad \land\quad
\bigwedge_{k=0}^1\bigl(f_k=\sum_{i=0}^2f_k^i\bigr).
\end{equation}
It is readily checked that each $\zeta_i \subseteq r_{\md{i}}$ is
included in exactly one of the regions $f_0$ or $f_1$, and that $\beta_{i,3}$ is disjoint from the other.
(Again, however, $f_0, f_1$ do not form a partition, because they must overlap.) We force all arcs
in each block $E_j$ to have a uniform block colour, and we force the
block colours to alternate by writing:
\begin{align}
& \hspace{1cm}\bigwedge_{k = 0}^1\bigwedge_{h = 1}^{|T|}\bigwedge_{\ell = 1}^{|\pcpw(t_h)|-1}
\neg C(f_k \cdot p_{h,\ell}, \ \ f_{1-k} \cdot p_{h,\ell+1}), 
\\
& \hspace{1cm}\bigwedge_{k = 0}^1
\bigwedge_{h = 1}^{|T|}\bigwedge_{h' = 1}^{|T|}\bigwedge_{i = 0}^2
\neg C(f_k \cdot p_{h,|\pcpw(t_h)|}\cdot r_i,\ \  f_k \cdot p_{h',1}\cdot r_{\md{i+1}}\cdot(-r_i)).
\end{align}
Thus, we may speak unambiguously of the colour ($f_0$ or $f_1$) of a
block: if $E_1$ is coloured $f_0$, then $E_2$ will be coloured $f_1$,
$E_3$ coloured $f_0$, and so on.  Using variables
$f_0'$ and $f_1'$, we similarly establish a block structure $E'_1,
\dots, E'_{m'}$ on the arcs $\zeta'_i$.

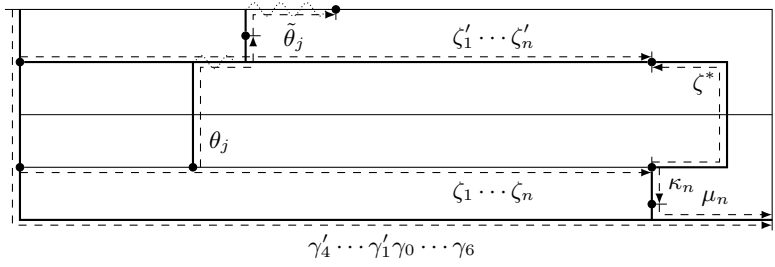
\begin{figure}[ht]
\begin{center}
	\begin{tikzpicture}[>=latex,point/.style={circle,draw=black,minimum size=1mm,inner sep=0pt},yscale=0.7]\small
        \draw (0,0) rectangle +(10,4);
        \draw[ultra thin] (0,2) -- +(10,0);
        \node (q11) at (0,1) [point,fill=black] {};
        \node (q11p) at (0,3) [point,fill=black] {};
        \node (qn1n) at (8.4,1) [point,fill=black] {};
        \node (qn1np) at (8.4,3) [point,fill=black] {};
        \node (pn2) at (8.4,0.3) [point,fill=black] {};
        \draw (q11) -- (qn1n);
        \draw[thick] (q11p) -- (qn1np);
        \draw[->,dashed] ($(q11)+(0,-0.1)$) -- ($(qn1n)+(0,-0.1)$) node[near end,below] {$\zeta_1\cdots\zeta_n$};
        \draw[->,dashed] ($(q11p)+(0,0.1)$) -- ($(qn1np)+(0,0.1)$) node[near end,above] {$\zeta_1'\cdots\zeta_n'$};
        \draw[thick] (qn1n) -- (8.4,0);
        \draw[->,dashed] ($(qn1n)+(0.1,0)$) -- ($(pn2)+(0.1,0)$) node[midway,right] {$\kappa_n$};
        \draw[ultra thin] (pn2) -- +(0.2,0);
        \draw[ultra thin] (qn1np) -- +(0,0.2);
        \draw[ultra thin] (qn1np) -- +(0,-0.2);
        \draw[ultra thin] (qn1n) -- +(0,0.2);
        \draw[->,dashed] ($(pn2)+(0.1,0)$) |- (10,0.1) node[near end,above] {$\mu_n$};
        \draw[thick] (0,4) |- (10,0);
        \draw[->,dashed] (-0.1,4) -- (-0.1,-0.1) -- (10,-0.1) node[midway,below] {$\gamma_4'\cdots\gamma_1'\gamma_0\cdots\gamma_6$};
        \draw[ultra thin] (0,4) -- +(-0.2,0);
        \draw[ultra thin] (10,0) -- +(0,-0.2);
        \draw[thick] (qn1n) -| ++(1,1);
        \draw[thick] (qn1np) -| ++(1,-1);
        \draw[dashed,<-] ($(qn1np)+(0,-0.1)$) -| ++(0.9,-0.9);
        \draw[dashed] ($(qn1n)+(0,0.1)$) -| ++(0.9,0.9);
        \node at (9.1,2.6) {$\zeta^*$};
        \node (ts) at (2.3,1) [point,fill=black] {};
        \node (tm) at (3,3.5) [point,fill=black] {};
        \draw[ultra thin] (tm) -- +(0.2,0);
        \node (te) at (4.2,4) [point,fill=black] {};
		\draw[densely dotted,decorate,decoration=snake] (2.3,3) -- (3,3);
		\draw[densely dotted,decorate,decoration=snake] ($(te)+(-1.2,0)$) -- (te);
        \draw[ultra thin] (te) -- +(0,-0.2);
        \draw[thick] (ts) -- (2.3,3);
        \draw[thick] (3,3) -- (tm);
        \draw[->,dashed] ($(ts)+(0.1,0)$) -- (2.4,2.9)  node [near start, right] {$\theta_j$} -- (3.1,2.9) -- ($(tm)+(0.1,0)$);
        \draw[thick] (tm) -- ++ (0,0.5);
        \draw[->,dashed] ($(tm)+(0.1,0)$) |- ($(te)+(0,-0.1)$) node [near end, below] {$\tilde{\theta}_j$};
    \end{tikzpicture}
\end{center}
\caption{Arc $\theta_j\tilde{\theta}_j$ intersecting $\zeta_1'\cdots\zeta_n'$.}
\label{fig:blockArcs}
\end{figure}

Now we match up the blocks in a 1--1 fashion, just as we matched up
the individual arcs in Stage 4. Let $\tseq{g}_0$, $\tseq{g}_1$,
$\tseq{g}'_0$ and $\tseq{g}'_1$ be new 3-region variables.  Recall
that every arc $\zeta_i$ contains some point of $b_{\md{i},3}$ (for
instance: $p_{i,1}$) and every such point is unambiguously labeled by
a region in $\vec p$ and a region in $(f_0,f_1)$. We wish to connect
any such arc that starts a block $E_j$ (i.e., any $\zeta_i$ labelled
by $p_{h,1}$ for some $h$) to the top edge of the upper window, with
the connecting arc depending on the block colour.  We can do this
using the constraints:
\begin{align}
\label{eq:r2:theta}
& 
\bigwedge_{k = 0}^1\bigl( (f_k \cdot (b_{0,3} + b_{1,3} + b_{2,3}) \leq \intermediate{g}_k)
                       \ \wedge \ \stack(\tseq{g}_k,\tseq{b}) 
                       \bigr).
\end{align}
Specifically, the first (actually: every) arc $\zeta_i$ in each block
$E_j$, for $1 \leq j \leq m$, is connected by an arc
$\theta_j\tilde{\theta}_j$ to some point on the upper edge of the
upper window, where $\theta_j \subseteq g_k$ and $\tilde{\theta}_j
\subseteq b$.  
Using corresponding formulas, we
%
ensure that the first arc in each block $E'_j$, for $1 \leq j \leq
m'$, is connected by an arc $\theta'_j\tilde{\theta}'_j$ to some point
on the bottom edge of the lower window, where $\theta'_j \subseteq
g'_k$ and $\tilde{\theta}'_j \subseteq b'$.

Recall from Stage~3 that
$q_{n+1,1}$ is connected by an arc $\kappa_n \subseteq a_0 + a_1 +
a_2$ to $p_{n,2}$,
which is in turn connected to the lower edge of the lower window by an
arc in lying in $a^*$. And recall from Stage~4 that $q_{n+1,1}$ is
connected by $\zeta^* \subseteq \ti{z}$ to $q'_{n+1,1}$. Thus, we see
from Fig.~\ref{fig:blockArcs} that the non-contact constraints
\begin{multline}
\label{eq:r2:gk}
    \neg C(g_0 + g_1,\ s'_4 +  \cdots  + s'_1 + s_0 + \cdots + s_5 + a^* + a_0 +a_1 + a_2 + z)
\end{multline}
ensure that each $\theta_j\tilde{\theta}_j$ ($1 \leq j \leq m$)
intersects one of the $\zeta'_i$ ($1 \leq i \leq n$). Indeed, since
$\tilde{\theta}_j \subseteq b$ cannot intersect any $\zeta'_i$, we
know that all such points of intersection lie on $\theta_j$. Using a corresponding formula, we
%
%
ensure that each $\theta'_j$ ($1 \leq j \leq m'$) intersects one of
the $\zeta_i$ ($1 \leq i \leq n$).

We now write the constraints
\begin{align}
& \bigwedge_{k = 0}^1\bigl(\neg C(g_k, f_{1-k}') \wedge  \neg C(g'_k, f_{1-k})\bigr).
\end{align}
Thus, any $\theta_j$ included in $g_k$ must join some arc $\zeta_i$ in
a block with colour $f_k$ to some arc $\zeta'_{i'}$ in a block
with colour $f_k'$; and symmetrically for the $\theta'_j$.
Adding
\begin{align}
& \neg C(g_0 + g'_0, \ \ g_1 + g'_1)
\end{align}
then ensures, via reasoning similar to that employed in
Stage~4, that $\theta_1$ connects the block $E_1$ to the block $E'_1$,
$\theta_2$ connects $E_2$ to $E'_2$, and so on; and similarly for the
$\theta'_j$ (as shown, schematically, in
Fig.~\ref{fig:blockCorrespondence}).  Thus, we have a 1--1
correspondence between the two sets of blocks, whence $m = m'$.
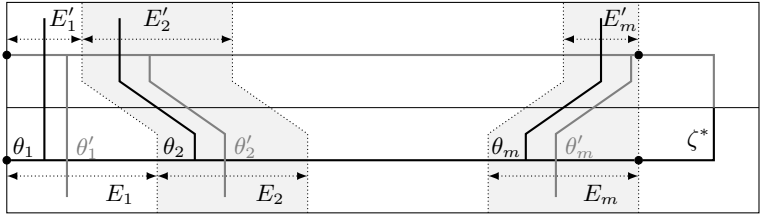
\begin{figure}[ht]
\begin{center}
	\begin{tikzpicture}[>=latex,point/.style={circle,draw=black,minimum size=1mm,inner sep=0pt},yscale=0.7]\small
        \fill[gray!10] (6.4,0) -- +(0,1.5) -- +(1,2.5) -- +(1,4) -- +(2,4) -- +(2,0) -- cycle;
        \fill[gray!10] (2,0) -- +(0,1.5) -- +(-1,2.5) -- +(-1,4) -- +(1,4) -- +(1,2.5) -- +(2,1.5) -- +(2,0) -- cycle;
        \draw (0,0) rectangle +(10,4);
        \draw[ultra thin] (0,2) -- +(10,0);
        \node (q11) at (0,1) [point,fill=black] {};
        \node (q11p) at (0,3) [point,fill=black] {};
        \node (qn1n) at (8.4,1) [point,fill=black] {};
        \node (qn1np) at (8.4,3) [point,fill=black] {};
        \draw[thick] (q11) -- (qn1n);
        \draw[thick,gray] (q11p) -- (qn1np);
        \draw[thick] (qn1n) -| ++(1,1);
        \draw[thick,gray] (qn1np) -| ++(1,-1);
        \node at (9.2,1.4) {$\zeta^*$};
        \draw[densely dotted,<->] (0,0.7) -- +(2,0) node [below,near end] {$E_1$};
        \draw[densely dotted,<->] (2,0.7) -- +(2,0) node [below,near end] {$E_2$};
        \draw[densely dotted,<->] (6.4,0.7) -- +(2,0) node [below,near end] {$E_m$};
        \draw[densely dotted,<->] (0,3.3) -- +(1,0) node [above,near end] {$E_1'$};
        \draw[densely dotted,<->] (1,3.3) -- +(2,0) node [above,midway] {$E_2'$};
        \draw[densely dotted,<->] (7.4,3.3) -- +(1,0) node [above,near end] {$E_m'$};
        \draw[thick] (0.5,1) -- (0.5,3.7) node[left,pos=0.09] {$\theta_1$};
        \draw[thick,gray] (0.8,3) -- (0.8,0.3) node[right,pos=0.65] {$\theta_1'$};
        \draw[thick] (2.5,1) -- (2.5,1.5) node[midway,left] {$\theta_2$} -- (1.5,2.5) -- (1.5,3.7);
        \draw[thick,gray] (1.9,3) -- (1.9,2.5) -- (2.9,1.5) -- (2.9,0.3) node[right,pos=0.2] {$\theta_2'$};
        \draw[thick] (6.9,1) -- (6.9,1.5) node[midway,left] {$\theta_m$\hspace*{-0.2em}} -- (7.9,2.5) -- (7.9,3.7);
        \draw[thick,gray] (8.3,3) -- (8.3,2.5) -- (7.3,1.5) -- (7.3,0.3) node[right,pos=0.2] {$\theta_m'$};
        \draw[densely dotted] (2,0) -- +(0,1.5) -- +(-1,2.5) -- +(-1,4);
        \draw[densely dotted] (4,0) -- +(0,1.5) -- +(-1,2.5) -- +(-1,4);
        \draw[densely dotted] (6.4,0) -- +(0,1.5) -- +(1,2.5) -- +(1,4);
        \draw[densely dotted] (8.4,0) -- +(0,4);
    \end{tikzpicture}
\end{center}
\caption{The 1--1 correspondence between the $E_j$ and the
  $E'_j$ established by the $\theta_j$ and the $\theta'_j$.}
\label{fig:blockCorrespondence}
\end{figure}

Finally, we regard elements of alphabet $T$ as fresh variables and
order them to form tuple $\vec{t}$. These variables are used for
labelling the components of $g_0$ and of $g_1$, and hence the arcs
$\theta_1, \dots, \theta_m$:
\begin{align}
& \partition(\vec{t}\,) \quad\land\quad \colourComp(g_0; \ \vec{t}\,)\quad\land\quad \colourComp(g_1; \ \vec{t}\,).
\end{align}
(Note that this time we can take the regions $\vec{t}$ to form a
partition.)  Adding the constraints
\begin{align}
& \bigwedge_{k = 0}^1 \bigwedge_{h = 1}^{|T|}\neg C(\sum_{\begin{subarray}{c}1 \leq h'\leq |T|\\h'\ne h\end{subarray}}(g_k \cdot t_{h'}), \ \ \sum_{\ell=1}^{|\pcpw(t_h)|} p_{h,\ell} + \sum_{\ell=1}^{|\pcpw'(t_h)|} p'_{h,\ell})
\end{align}
instantly ensures that the sequences of tile indices $h_1, \dots, h_m$
and $h'_1, \dots, h'_m$ are identical. In other words, $\tau = \tau'$.

This
completes the argument that, if $\psi_\pcpW$ has a satisfying
assignment over $\RC(\R^2)$, then $\pcpW$ is a positive instance of
the PCP. By extending the arrangement of Fig.~\ref{fig:concrete2} in
the obvious way, we see that, if $\pcpW$ is a positive instance of the
PCP, then $\psi_\pcpW$ has a satisfying assignment over $\RCP(\R^2)$,
and hence (trivially) a satisfying assignment over $\RC(\R^2)$.
\end{proof}

The case \cBCci{} is dealt with as for Corollary~\ref{cor:inftyCci}:
we simply replace all occurrences of $c$ in $\psi_\pcpW$ with
$\ic$. Denoting the resulting $\cBCci$-formula by
$\psi^\circ_\pcpW$, we see that the following are equivalent:
(\emph{i}) $\pcpW$ is positive; (\emph{ii}) $\psi^\circ_\pcpW$ is
satisfiable over $\RCP(\R^2)$; (\emph{iii}) $\psi^\circ_\pcpW$ is
satisfiable over $\RC(\R^2)$. Thus,
\begin{corollary}\label{cor:LRCR2:Cci}
The problems $\Sat(\cBCci,\RC(\R^2))$ and $\Sat(\cBCci,\RCP(\R^2))$ are r.e.-hard.
\end{corollary}

Employing the techniques of the proof of Theorem~\ref{cor:inftyBc}, we show that
\begin{theorem}\label{cor:LRCR2}
The problems $\Sat(\cBc,\RC(\R^2))$ and $\Sat(\cBc,\RCP(\R^2))$ are r.e.-hard.
\end{theorem}
\begin{proof}
Again, observe that all conjuncts of $\psi_\pcpW$ featuring the
predicate $C$ are \emph{negative} (remember that there are additional
such literals implicit in the use of 3-region variables, e.g.,
$(\intermediate{r} \ll r)$; but let us ignore these for the moment.)
Recall the formula $\noncontact(r,s)$ from the proof of
Theorem~\ref{cor:inftyBc} and consider the effect of replacing any
literal $\neg C(r,s)$ in $\psi_\pcpW$ by the corresponding instance of
$\noncontact(r+r',s+s')$, where $r'$ and $s'$ are fresh variables;
denote the resulting formula by $\psi$. It is easy to see that $\psi$
entails $\psi_\pcpW$; hence if $\psi$ is satisfiable, then $\pcpW$ is
a positive instance of the PCP.

We next show that, if $\pcpW$ is a positive instance of the PCP, then
$\psi$ is satisfiable over $\RCP(\R^2)$. For consider a tuple from
$\RCP(\R^2)$ satisfying $\psi_\pcpW$, and based on the arrangement of
Fig.~\ref{fig:concrete2}. Note that if $\tseq{r}$ and $\tseq{s}$ are
3-regions whose outer shells, $r$ and $s$ are not in contact (e.g.,
$\tseq{a}_{0,1}$ and $\tseq{a}_{0,3}$), then $r$ and $s$ have ({\em
  i}) finitely many components, and ({\em ii}) connected
complements. Hence, it is easy to find polygons $r'$ and $s'$
satisfying $\noncontact(r+r',s+s')$. Fig.~\ref{fig:connectingRsAndSs}
represents the situation schematically.  We may therefore assume that
all such literals $\neg C(r,s)$ have been eliminated from
$\psi_\pcpW$.
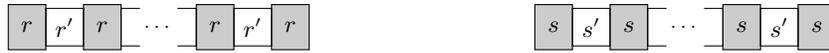
\begin{figure}[ht]
	\begin{center}
	\begin{tikzpicture}\small
        \foreach \l/\y in {r/0,s/7}
		{
            \foreach \x in {0,1,2.5,3.5}
            {
			     \draw[fill=black!20] (\x+\y,-0.3) rectangle ++(0.5,0.6) node[midway]{$\l$};
            }
            \foreach \x in {0.5,3}
            {
			     \draw (\x+\y,-0.25) rectangle ++(0.5,0.5) node[midway]{$\l'$};
            }
			\draw (1.75+\y,-0.25) --++ (-0.25,0) --++(0,0.5)--++(0.25,0);
			\node at (2+\y,0){\scriptsize$\dots$};
			\draw (2.25+\y,-0.25) --++ (0.25,0) --++(0,0.5)--++(-0.25,0);
		}		
	\end{tikzpicture}
	\end{center}
	\caption{Disjoint connected regions $r+r'$ and $s+s'$ for regions $r$ and $s$ with finitely many components and connected complements.}\label{fig:connectingRsAndSs}
\end{figure}

We are not quite done, however. We must show that we can replace the
\emph{implicit} non-contact constraints
$(\inner{r}\ll\intermediate{r})$ and $(\intermediate{r}\ll r)$ that
come with the use of each 3-region variable $\tseq{r}$ by suitable
\cBc-formulas.  Since the two conjuncts are identical in form, we
only show how to deal with $(\intermediate{r}\ll r)$, which, we recall,
is an alternative notation for
$\neg C(\intermediate{r}, -r)$.  Since the complement of $-r$
is in general not connected, a direct use of
$\noncontact(\intermediate{r}+r',(-r)+s')$ will result in an
unsatisfiable formula.  Instead, we represent $-r$ as the sum of two
regions $s_1$ and $s_2$ with connected complements, and then proceed
as before. In particular, we replace $(\intermediate{r} \ll r)$ by 
$((-r)=(s_1+s_2))  \land \noncontact(\intermediate r + r_1,s_1)
	\land \noncontact(\intermediate r + r_2,s_2)$.
%
For $i=1,2$, $\intermediate r+r_i$ is a connected region that is disjoint from
$s_i$. So, $\intermediate r$ is disjoint from $s_1$ and $s_2$, and hence
disjoint from their sum, $-r$. Fig.~\ref{fig:crenellate} shows regions
$r_1,s_1$ satisfying the above formula; the other pair, $r_2,s_2$ is the mirror image.
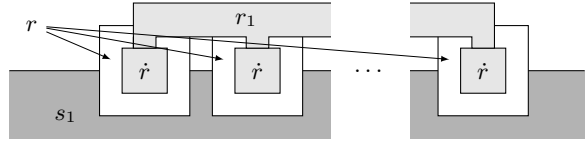
\begin{figure}[ht]
	\begin{center}
		\begin{tikzpicture}[scale=1.5,>=latex]\small
			\fill[gray!60] (-0.8,-0.6) rectangle (4.3,0);
			\draw (-0.8,0)--(4.3,0);
			\foreach \x in {0,1,3}
			{
				\filldraw[fill=white] ($(\x,0)+(0,-0.4)$) rectangle +(0.8,0.8);
				\filldraw[fill=gray!20] ($(\x,0)+(0.2,-.2)$) rectangle +(.4,.4) node[midway]{$\intermediate{r}$};
			}			
			\filldraw[fill=gray!20] (.3,.2)--++(0,.4)--++(1.8,0)--++(0,-.3)--(1.5,.3)--++(0,-.1)--++(-.2,0)--++(0,.1)--(.5,.3)--++(0,-.1)--++(-.2,0)--++(0,.1);
			\filldraw[fill=gray!20] (2.6,.6)--++(0.9,0)--++(0,-.4)--++(-.2,0)--++(0,.1)--(2.6,.3);
            \fill[white] (2.05,-0.61) rectangle +(0.7,1.22);
			\node at (-0.3,-.4) {$s_1$};
			\node at (1.3,.45) {$r_1$};
            \node at (2.4,0) {\dots};
            \node (r) at (-0.6,0.4) {$r$};
            \draw[ultra thin,->] (r) -- (0.1,0.1);
            \draw[ultra thin,->] (r) -- (1.1,0.1);
            \draw[ultra thin,->] (r) -- (3.1,0.1);
		\end{tikzpicture}
	\end{center}
	\caption{Disjoint connected regions $\intermediate{r}+r_1$ and $s_1\leq (-r)$ for $\intermediate{r}$ right inside $r$.}
	\label{fig:crenellate}
\end{figure}

Let $\psi^c_\pcpW$ be the result of replacing in $\psi_\pcpW$ all the (explicit
or implicit) conjuncts containing the predicate $C$, as just described. We have
thus shown that if $\psi^c_\pcpW$ is satisfiable over $\RC(\R^2)$ then
$\pcpW$ is positive, and, conversely, if $\pcpW$ is positive then $\psi^c_\pcpW$ is
satisfiable over $\RCP(\R^2)$. This completes the proof.
\end{proof}

\begin{theorem}\label{cor:LRCR2:Bci}
The problems $\Sat(\cBci,\RC(\R^2))$ and $\Sat(\cBci,\RCP(\R^2))$ are r.e.-hard.
\end{theorem}
\begin{proof}
We begin with the \cBCci-formula $\psi^\circ_\pcpW$ constructed in the
proof of r.e.-hardness result in Corollary~\ref{cor:LRCR2:Cci}. We
proceed by eliminating occurrences of $C$. However, we cannot directly
use the same Lemma~\ref{lma:Cci2BciStar} as in the proof of
Theorem~\ref{theo:inftyBci} because the regions in question may not
necessarily be bounded.  For instance, consider the formula
$(\intermediate{s}_0\ll s_0)$, which is an alternative notation for
$\neg C(\intermediate{s}_0,-s_0)$: although the region
$\intermediate{s}_0$ in Fig.~\ref{fig:concrete1} is evidently bounded,
$-s_0$ is not. We proceed as follows.  Say that a region $r$ is
\emph{quasi-bounded} if either $r$ itself or its complement, $-r$, is
bounded. Since all the polygons in the tuple satisfying
$\psi^\circ_\pcpW$ are quasi-bounded, we can eliminate all occurrences
of $C$ from $\psi^\circ_\pcpW$ using the following
fact~\cite[p.~137]{Newman64}:
\begin{lemma}
\label{lma:Newman}
Let $F$, $G$ be disjoint, closed subsets of $\R^2$ such that
both $\R^2\setminus F$ and $\R^2 \setminus G$ are connected. Then
$\R^2\setminus (F \cup G)$ is connected.
\end{lemma}

So, suppose we have a conjunct $\neg C(r,s)$ in $\psi^\circ_\pcpW$. We
consider the following formula:
\begin{equation*}
\chi(r,s,\vec{v}) \  =  \ (r = r_1 + r_2) \land (s = s_1 + s_2) \land
  \bigwedge_{1 \leq i,j \leq 2}
    \bigl(\noncontacti(\vec{v}_{ij}) \land (r_i \leq v^1_{ij}) \land (s_j \leq v^2_{ij})\bigr),
\end{equation*}
where $\vec{v}$ is a vector of variables containing $r_1,r_2,s_1,s_2$
and the $v^1_{ij},\dots,v^5_{ij}$, for $1 \leq i,j\leq 2$, and
$\noncontacti(v^1,\dots,v^5)$ is the formula defined before
Lemma~\ref{lma:Cci2BciStar}. By
Lemma~\ref{lma:Cci2BciStar}~(\emph{i}), $\chi(r,s,\vec{v})$ entails
$\neg C(r,s)$ over $\RC(\R^2)$. We also show that, conversely, if $a$
and $b$ are disjoint quasi-bounded polygons then there exists a tuple
of polygons $\vec{e}$ such that $(a, b, \vec{e})$ satisfies
$\chi(r,s,\vec{v})$.  Indeed, it is routine to show that, for each
quasi-bounded region $a$, there exist a pair of regular closed polygons $a_1$ and
$a_2$ such that $a = a_1 + a_2$ and both $\R^2 \setminus a_1$ and
$\R^2 \setminus a_2$ are connected. Let $b_1$ and $b_2$ be chosen
analogously for $b$. Then, for all $1 \leq i,j \leq 2$, we have $a_i
\cap b_j = \emptyset$ and, by Lemma~\ref{lma:Newman}, $\R^2 \setminus
(a_i + b_j)$ is connected. Thus, there exists a piecewise-linear
Jordan curve in $\R^2 \setminus(a_i+b_j)$ separating $a_i$ and $b_j$.
By Lemma~\ref{lma:Cci2BciStar}~(\emph{ii}), let $\vec{e}_{ij}$ be a
tuple of polygons satisfying $\noncontacti(\vec{v}_{ij})$ and such
that $a_i \leq e^1_{ij}$ and $b_j \leq e^2_{ij}$. It should be clear
that the tuple of $a_1,a_2,b_1,b_2$ and the $\vec{e}_{ij}$, for $1
\leq i,j\leq 2$, is as required.

By replacing all occurrences of $C$ in $\psi^\circ_\pcpW$ as
described above, we obtain a \cBci-formula, say
$\psi^*_\pcpW$, such that, if $\psi^*_\pcpW$ is
satisfiable over $\RC(\R^2)$, then $\pcpW$ is a positive instance of
PCP; and, conversely, if $\pcpW$ is a positive instance of PCP, then
$\psi^*_\pcpW$ is satisfiable over $\RCP(\R^2)$.
\end{proof}

\section{The language \cBci{} in dimensions greater\\ than 2}
\label{sec:3d}
In this section, we consider the complexity of satisfying
\cBci-formulas by regular closed polyhedra and regular
closed sets in three-dimensional Euclidean space. We proceed by
analysing the connections between geometrical and graph-theoretic
interpretations of \cBci.

A topological space $T$ in which the intersection of {\em any} family
of open sets is open is called an \emph{Aleksandrov space}.  Every
quasi-order $(W,R)$, that is, a transitive and reflexive relation $R$
on $W$, can be regarded as an Aleksandrov space by taking $X \subseteq
W$ to be open just in case $x \in X$ and $xRy$ imply $y\in X$. (Hence,
$X$ is closed just in case $x \in X$ and $yRx$ implies $y \in X$.) It
can be shown~\cite{Bourbaki66} that every Aleksandrov space is the
homeomorphic image of one constructed in this way.  In the sequel, we
shall silently treat any quasi-order $(W,R)$ as a topological space.

It turns out that, to satisfy all satisfiable \cBc- and
\cBci-formulas, Aleksandrov spaces of rather primitive structure are
enough. Call a quasi-order $(W_0\cup W_1,R)$ a \emph{quasi-saw} if
$W_0$ and $W_1$ are disjoint and $R$ is the reflexive closure of a
relation $R' \subseteq W_1 \times W_0$ with domain $W_1$. The points
in $W_i$ are said to be \emph{of depth} $i$ in $(W_0\cup W_1,R)$; see
Fig.~\ref{fig:qsaw}.
\begin{figure}[ht]
\begin{center}
\begin{tikzpicture}[>=latex,point/.style={circle,draw=black,minimum size=1mm,inner sep=0pt},yscale=0.7]\small
\foreach \x in {0,1,2,3,4,5}
{%
    \node (x \x) at (1.5*\x,1) [point,fill=black,label=above:{$x_\x$}] {};
}
\foreach \x in {0,1,2,3,4}
{
    \node (z \x) at (1.5*\x + 0.75,0) [point,fill=black,label=below:{$z_\x$}] {};
}
\node (x 6) at (8.25,1) [point,fill=black,label=above:{$x_6$}] {};
\node (z 5) at (8.25,0) [point,fill=black,label=below:{$z_5$}] {};
\foreach \z/\x in {0/0,0/1,1/1,1/2,2/1,2/3,2/4,3/3,4/4,4/5}
{
    \draw[->] (z \z) -- (x \x);
}
\draw[->] (z 5) -- (x 6) node[midway,right] {\scriptsize $R$};
\node at (10,1) {$W_0$ = depth 0};
\node at (10,0) {$W_1$ = depth 1};
\end{tikzpicture}
\end{center}
\caption{Quasi-saw.}
\label{fig:qsaw}
\end{figure}
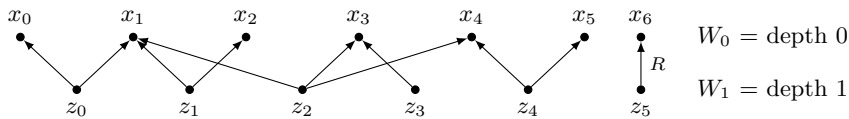

Every regular closed set $X$ in a quasi-saw $(W_0 \cup W_1,R)$ is
uniquely defined by its points of depth 0:
\begin{align*}
\text{for each } z\in W_1,\quad z\in X \quad\text{iff}\quad \text{there is } x\in W_0 \cap X \text{ such that } zRx.
\end{align*}
A \emph{quasi-saw model} is a model based on a quasi-saw (with
variables interpreted by regular closed sets).  The proof of the
following lemma follows from~\cite[Lemmas 4.1 and 4.2]{iscloes:kp-hwz10} (see also~\cite{Wolter&Z00ecai}). But the critical observation can, in essence,
be found already in~\cite{McKinsey&Tarski44} and~\cite{Kripke63}: for
every formula $\varphi$ and
every topological model $\mathfrak{I}$, there exist a finite Aleksandrov
model $\mathfrak{A}$ and a continuous
function $f\colon \mathfrak{I}
\to \mathfrak{A}$ such that $\tau^\mathfrak{A} = f(\tau^\mathfrak{I})$ for every term $\tau$ in $\phi$.
\begin{lemma}\label{lemma:quasi-saw}
Let $\varphi$ be a \cBc- or \cBci-formula. If $\varphi$ is satisfiable
over $\RegC$ then it can be satisfied in a finite quasi-saw model.
\end{lemma}

We begin by briefly discussing the results of~\cite{KPZ10kr} for the
\emph{polyhedral} case. Denote by $\ConRC$ the class of all frames over
connected topological spaces with regular closed regions.
For a $\cBci$-formula $\phi$, let $\phi^{\bullet}$ be the result of
replacing every occurrence of $\ic$ in $\phi$ with $c$. Evidently, the
mapping $\phi \mapsto \phi^{\bullet}$ is a bijection from \cBci{} to
\cBc.
\begin{theorem}\label{theo:BciRCRP3}
For all $n \geq 3$, the mapping $\phi \mapsto \phi^\bullet$
constitutes a reduction of $\Sat(\cBci,\RCP(\R^n))$ to
$\Sat(\cBc,\ConRC)$. Hence, the problems $\Sat(\cBci,\RCP(\R^n))$
coincide, and are all \ExpTime-complete.
\end{theorem}
\begin{proof}
A \emph{connected partition} in $\RCP(\R^n)$ is a tuple
$X_1,\dots,X_k$ of non-empty polyhedra having connected and pairwise
disjoint interiors, which sum to the entire space $\R^n$. The
\emph{neighbourhood graph} $(V,E)$ of this partition has vertices $V =
\{X_1, \dots, X_k\}$ and edges
\begin{equation*}
E = \{ (X_i, X_j) \mid i \ne j
\text{ and } \ti{(X_i + X_j)} \text{ is connected}\};
\end{equation*}
see Fig.~\ref{fig:conn-part}.
\begin{figure}[ht]
\begin{center}
\begin{tikzpicture}[point/.style={circle,draw=black,minimum size=1mm,inner sep=0pt},scale=0.2mm]
\filldraw[fill=gray!20] (0,-2,-3) -- (0,2,-3) -- (0,2,2) -- (0,-2,2) --   cycle;
\filldraw[fill=gray!60,fill opacity=0.5] (0,0,0) -- (3,0,0) -- (3,-2,-3) -- (0,-2,-3) --   cycle;
\filldraw[fill=gray!10,fill opacity=0.5] (0,0,-3) -- (0,0,2) -- (3,0,2) -- (3,0,-3) --   cycle;
\filldraw[fill=white,fill opacity=0.5] (0,0,0) -- (3,0,0) -- (3,2,0) -- (0,2,0) --   cycle;
\begin{scope}[dashed,draw=white]
\clip (0,0,0) -- (3,0,0) -- (3,2,0) -- (0,2,0) --   cycle;
\draw (0,-2,-3) -- (0,2,-3);
\draw (0,0,2) -- (0,0,-3) -- (3,0,-3);
\end{scope}
\begin{scope}[dashed,draw=white]
\clip (0,0,-3) -- (0,0,2) -- (3,0,2) -- (3,0,-3) --   cycle;
\draw (0,-2,-3) -- (0,2,-3);
\draw (0,-2,-3) -- (0,0,0);
\end{scope}
\filldraw[dashed,fill=gray,fill opacity=0.5] (0,0,-1.5) -- (2,0,0) -- (0,1.5,0)  --   cycle;
\draw (2,0,0) -- (0,1.5,0);
\node (1) at (-1.5,0,0) {\small $X_1$};
\node (2) at (1.8,-1,1.5) {\small $X_2$};
\node (3) at (3,-1,-2.5) {\small $X_3$};
\node (4) at (1.8,1.5,-2.5) {\small $X_4$};
\node (5) at (2.5,.5,2) {\small $X_5$};
\node (6) at (0.3,0.3,-0.3) {\small $X_6$};
\begin{scope}[xshift=110mm]
\node [label=left:{\small $X_1$}] (v1) at (-1.5,0,0) [point] {};
\node [label=right:{\small $X_2$}](v2) at (1.8,-1.5,1.5)[point] {};
\node [label=right:{\small $X_3$}](v3) at (1.8,-1.5,-2.5)[point] {};
\node [label=right:{\small $X_4$}](v4) at (1.8,1.5,-2.5)[point] {};
\node [label=right:{\small $X_5$}](v5) at (1.8,1.5,1.5)[point] {};
\node [label=above:{\small $X_6$}] (v6) at (0,0.8,0)[point] {};
\draw (v1) -- (v2);
\draw (v2) -- (v3);
v\draw (v1) -- (v3);
\draw (v3) -- (v4);
\draw (v2) -- (v5);
\draw (v1) to [bend right, looseness=0.3] (v5);
\draw (v1) to [bend left, looseness=1.5] (v4);
\draw (v3) to [bend left, looseness=0.3] (v6);
\draw (v4) -- (v5);
\draw (v1) -- (v6);
\draw (v4) -- (v6);
\draw (v5) -- (v6);
\end{scope}
\end{tikzpicture}
\end{center}
\caption{A connected partition in $\RCP(\R^3)$ and its neighbourhood graph.}\label{fig:conn-part}
\end{figure}
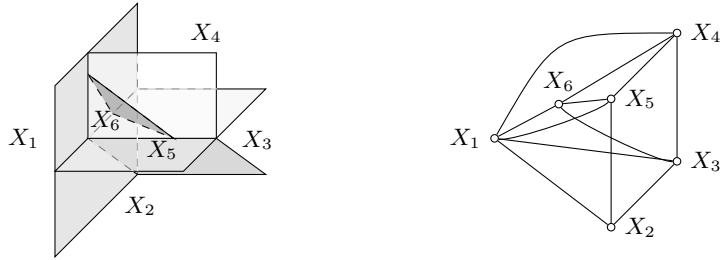
Clearly, every connected partition in
$\RCP(\R^n)$ has a connected neighbourhood graph; and conversely, one
can show that \emph{every} connected graph is the neighbourhood graph
of some connected partition in $\RCP(\R^n)$. Furthermore, every
neighbourhood graph $(V,E)$ gives rise to a quasi-saw $(W_0 \cup W_1,
R)$, where $W_0 = V$, $W_1 = \{z_{x,y} \mid (x,y) \in E\}$, and $R$ is
the reflexive closure of $\{(z_{x,y}, x), (z_{x,y}, y) \mid (x,y) \in
E \}$. Note that in this quasi-saw every point of depth 1 has
precisely two $R$-successors. Such quasi-saws are called
\emph{2-quasi-saws}. Conversely, every connected
2-quasi-saw $(W_0\cup W_1,R)$ can be represented as the neighbourhood
graph $(W_0,E)$ of some connected partition, where
\begin{align*}
E = \{ (x,y) \mid x\ne y \text{ and there is } z\in W_1 \text{ with } zRx \text{ and } zRy\}.
\end{align*}
{From} this, we see that a \cBci-formula $\varphi$ is satisfiable over
$\RCP(\R^n)$ if and only if $\varphi$ is satisfiable over a connected
2-quasi-saw. But, over 2-quasi-saws, connectedness coincides with
interior-connectedness.  Thus, $\varphi$ is satisfiable over
$\RCP(\R^n)$ if and only if $\phi^{\bullet}$ is satisfiable over a
connected 2-quasi-saw. The problem $\Sat(\cBc,\ConRC)$ is known to be
\ExpTime-complete~\cite{iscloes:kp-hwz10}.
\end{proof}

Having shown that the problem $\Sat(\cBci,\RCP(\R^3))$ is
\ExpTime-complete, we now turn our attention to the satisfiability of
\cBci-formulas over the complete Boolean algebra $\RC(\R^3)$, where
the picture changes drastically: for instance, the
\cBci-formula~\eqref{eq:wiggly} is not satisfiable over 2-quasi-saws,
but has a quasi-saw model as in Fig.~\ref{fig:broom}.
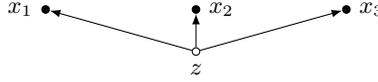
\begin{figure}[ht]
\centering\begin{tikzpicture}[>=latex, point/.style={circle,draw=black,minimum size=1mm,inner sep=0pt},yscale=0.7]\small
\node (r1) at (2,0.8)  [point,fill=black,label=left:{$x_1$}] {};
\node (r2) at (4,0.8)  [point,fill=black,label=right:{$x_2$}] {};
\node (r3) at (6,0.8)  [point,fill=black,label=right:{$x_3$}] {};
\node (r) at (4,0)  [point,label=below:{$z$}] {};
\draw[->] (r) -- (r1);
\draw[->] (r) -- (r2);
\draw[->] (r) -- (r3);
\end{tikzpicture}
\caption{A quasi-saw model $\mathfrak{A}$ of~\eqref{eq:wiggly}: $r_i^\mathfrak{A} = \{x_i,z\}$.}
\label{fig:broom}
\end{figure}
In fact, it is shown in~\cite{KPZ10kr} that every \cBci-formula
$\phi$ satisfiable over $\ConRC$ can be satisfied in a connected
quasi-saw model of size bounded by a polynomial function of $|\phi|$,
and thus the problem $\Sat(\cBci,\ConRC)$ is \NP-complete. The
following theorem says, in essence, that such polynomial models also
give rise to `small' models over regular closed subsets of $\R^n$, for
$n \geq 3$:
\begin{theorem}\label{theo:BciRCR3}
The problems $\Sat(\cBci,\RC(\R^n))$, for all $n \geq 3$, coincide with
$\Sat(\cBci,\ConRC)$ and are all \NP-complete.
\end{theorem}
\begin{proof}
We need only establish the special case $n=3$; the general result
follows by cylindrification. So, suppose $\varphi$ is satisfied in a
model $\mathfrak{A}$ over a connected quasi-saw $(W_0\cup W_1,R)$ of
size bounded by a polynomial function of $|\phi|$. Let $W_i$ be the
set of points of depth $i = 0,1$.  Without loss of generality we may
assume that there is a point $z_0\in W_1$ with $z_0 R x$ for all $x\in
W_0$ (adding such a point cannot change the truth-values of
subformulas of $\varphi$ of the form $(\tau_1 = \tau_2)$ or
$\ic(\tau)$).  We show now how $\mathfrak{A}$ can be embedded into a
model over $\RC(\R^3)$. In the remainder of this proof, we repeatedly
rely on the fact that, if $r$ and $s$ are interior-connected, regular
closed subsets of some topological space, with $r \cdot s \neq 0$,
then $r+s$ is also interior-connected.

We select \emph{open} balls $D_z$ for $z\in W_1\setminus\{z_0\}$
such that their closures are pairwise non-intersecting, and define
$D_{z_0} = \R^3 \setminus \bigcup_{z\in W_1\setminus\{z_0\}}
\tc{D}_z$.  Thus, each $D_z$ is connected and open, and the open set
$D  = \bigcup_{z \in W_1} D_z$ is dense. Then we take pairwise disjoint sets
$B^1_x$ for $x\in W_0$, each homeomorphic to a \emph{closed} ball, and
arranged so that, for all $x \in W_0$ and $z \in W_1$,  $D_z \nsubseteq B^1_x$ and $B^1_x \cap D_z \neq \emptyset$ if and only if
$zRx$.

We describe a construction in which the regular closed sets $B^1_x$
are expanded to sets $B_x$ so as to exhaust the entire space, $\R^3$.
First, let $q_1, q_2, \dots$ be an enumeration of all the points in
$D$ with \emph{rational} coordinates.  Consider any piecewise-linear
Jordan arc $\alpha$ such that the endpoints of each linear segment of
$\alpha$ have rational coordinates: call such an $\alpha$ {\em
  rational piecewise-linear}; and let $\alpha_1, \alpha_2, \dots$ be
an enumeration of all the rational piecewise-linear arcs with both
endpoints in the open set $D$.  We define, for all $k \geq 1$, a
collection $\{B^k_x \mid x \in W_0\}$ of interior-connected, pairwise
disjoint, regular closed sets in $\R^3$.  The case $k=1$ has already
been dealt with. Suppose, then, for $k \geq 1$, the sets $B^k_x$ have
been defined; we construct the sets $B^{k+1}_x$ in two steps:
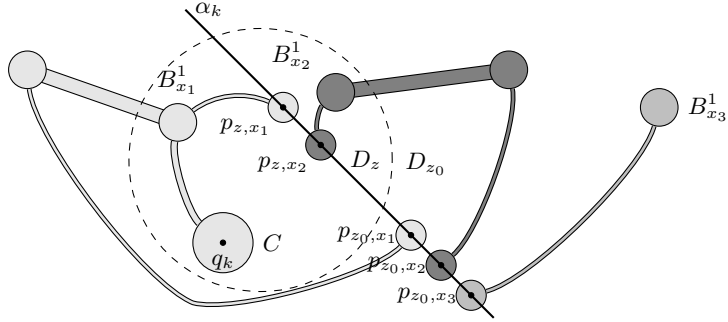
\begin{figure}[ht]
\begin{center}
\begin{tikzpicture}[
clball/.style={circle,draw=black,minimum size=3mm,inner sep=0pt},
opball/.style={circle,dashed,draw=black,minimum size=30mm,inner sep=0pt}
]
\draw[double=gray!20,double distance=5pt] (-0.4,-1.5) to (1.6,-2.2);
\draw[double=gray,double distance=5pt] (6,-1.5) to (3.7,-1.8);
%
\node (x10) at (-0.4,-1.5)[clball,minimum size=5mm,fill=gray!20] {};
\node (x20) at (6,-1.5)[clball,minimum size=5mm,fill=gray] {};
%
\node [label=above:{\small $B_{x_1}^1$}] (x1) at (1.6,-2.2)[clball,minimum size=5mm,fill=gray!20] {};
\node [label=above left:{\small $B_{x_2}^1$}] (x2) at (3.7,-1.8)[clball,minimum size=5mm,fill=gray] {};
\node [label=right:{\small $B_{x_3}^1$}] (x3) at (8,-2)[clball,minimum size=5mm,fill=gray!50] {};
%
\node (z1) at (2.7,-2.7)[opball,minimum size=35mm] {};
\node at (4.1,-2.7) {\small $D_z$};
\node at (4.9,-2.75) {\small $D_{z_0}$};
\node (xc1) at (2.2,-3.8)[clball,fill=gray!20,minimum size=8mm,label=right:{\small $C$}] {};
\node [label=below:{\small $q_k$}](q) at (xc1)[circle,inner sep=0pt,minimum size=2,draw=black,fill=black] {};
\draw[double=gray!20,double distance=2pt] (xc1) to [bend left, looseness=0.5] (x1);
%
\node (xcp1) at (3,-2)[clball,fill=gray!20,minimum size=4mm] {};
\node [label=below left:{\small $p_{z,x_1}$}] at (xcp1)[circle,inner sep=0pt,minimum size=2,draw=black,fill=black] {};
\draw[double=gray!20,double distance=1pt] (xcp1) to [bend right, looseness=0.9] (x1);
\node (xcp2) at (3.5,-2.5)[clball,fill=gray,minimum size=4mm] {};
\node [label=below left:{\small $p_{z,x_2}$}] at (xcp2)[circle,inner sep=0pt,minimum size=2,draw=black,fill=black] {};
\draw[double=gray,double distance=1pt] (xcp2) to [bend left, looseness=0.5] (x2);
\node (xcp10) at (4.7,-3.7)[clball,fill=gray!20,minimum size=4mm] {};
\node [label=left:{\small $p_{z_0,x_1}$}] at (xcp10)[circle,inner sep=0pt,minimum size=2,draw=black,fill=black] {};
\draw[double=gray!20,double distance=1pt] (xcp10) to [bend left, looseness=0.5] (1.8,-4.6) to [bend left, looseness=0.5] (x10);
\node (xcp20) at (5.1,-4.1)[clball,fill=gray,minimum size=4mm] {};
\node [label=left:{\small $p_{z_0,x_2}$}] at (xcp20)[circle,inner sep=0pt,minimum size=2,draw=black,fill=black] {};
\draw[double=gray,double distance=1pt] (xcp20) to [bend right, looseness=0.5] (x20);
\node (xcp30) at (5.5,-4.5)[clball,fill=gray!50,minimum size=4mm] {};
\node [label=left:{\small $p_{z_0,x_3}$}] at (xcp30)[circle,inner sep=0pt,minimum size=2,draw=black,fill=black] {};
\draw[double=gray!50,double distance=1pt] (xcp30) to [bend right, looseness=0.5] (x3);
\draw[thick] (1.7,-0.7) -- ++(4.1,-4.1) node[at start,right] {\small $\alpha_k$};
\end{tikzpicture}
\end{center}
\caption{Filling $D_z$ and $D_{z_0}$ with the sets $B_{x_i}$, for
  $z R x_1$, $z R x_2$ and $z_0 R x_i$, $i = 1,2,3$. 
  }
\label{fig:apollonian}
\end{figure}
\begin{enumerate}
\item If $q_k \in B^k_x$ for some $x \in W_0$,
  let $\hat{B}^k_{x'} = B^k_{x'}$, for every $x' \in W_0$.
  Otherwise, $q_k \in D_z$ for some $z \in W_1$. Pick some
  $x\in W_0$ with $zRx$ and let $C\subseteq D_z$ be a regular closed
  interior-connected set containing $q_k$ and a point in $\ti{(B^1_x)}
  \cap D_z$ in its interior (e.g., a closed ball centred on $q_k$ and
  a regular closed `rod' connecting it to $B^1_x$, as depicted in
  Fig.~\ref{fig:apollonian}). Let $\hat{B}^k_x = B^k_x + C$, and let
  $\hat{B}^k_{x'} = B^k_{x'}$ for all other $x' \in W_0$.  The sets
  $\hat{B}^k_{x'}$, for $x' \in W_0$, are interior-connected, and
  $C$ can be chosen so that the $\hat{B}^k_{x'}$ are
  pairwise disjoint.
\item Let $\hat{B}^k = \bigcup_{x\in W_0} \hat{B}^k_x$. For
  each $z \in W_1$ such that $\alpha_k \cap D_z$ is not contained in
  $\hat{B}^k$ and for each $x \in W_0$ such that $zRx$, choose a
  distinct point $p_{z,x} \in \alpha_k\cap D_z$, not lying in
  $\hat{B}^k$.  If $p_{z,x}$ is defined, let $C_{z,x} \subseteq D_z$
  be a regular closed interior-connected set containing $p_{z,x}$ and
  a point in $\ti{(B^k_x)}\cap D_z$ in its interior, see
  Fig.~\ref{fig:apollonian}; otherwise, let $C_{z,x}=\emptyset$.  Set
  $B^{k+1}_x = \hat{B}^k_x + \sum_{z \in W_1} C_{z,x}$, for all $x \in
  W_0$. The sets $B^{k+1}_x$ are interior-connected and, clearly, the
  $C_{z,x}$ can be chosen such that the $B^{k+1}_x$ are also pairwise
  disjoint.
\end{enumerate}
This completes the definition of the sets $\{B^{k}_x \mid x \in
  W_0\}$ for all $k \geq 1$.

Since $\RC(\R^3)$ is a complete Boolean algebra, define $B_x  =  \sum_{k =1}^\infty B^k_x$, for each $x
\in W_0$.
We show that the sets $B_x$ are interior-connected and form a
partition (i.e., their pairwise products are empty, and they sum to
$\R^3$).  Indeed, for distinct $x, y \in W_0$, we certainly have, for
all $k, \ell \geq 1$, $B^k_x \cdot B^\ell_y = 0$, whence, by the
distributivity law, $B_x \cdot B_y = 0$. And since, for all $k \geq
1$, $B^k_x$ is interior-connected and includes the non-empty,
interior-connected set $B^1_x$, the set $\bigcup_{k=1}^\infty
\ti{(B^k_x)}$ is connected. But $\ti{B}_x$ lies in between
$\bigcup_{k=1}^\infty \ti{(B^k_x)}$ and its closure, and hence is also
connected. Finally, by Step~1 of the above construction, every
rational point of the set $D$ lies in some $B_x$, so that $\sum_{x \in
  W_0} B_x \supseteq D$, whence $\sum_{x \in W_0} B_x = 1$.  This
completes the definition of the partition $\{B_x \mid x \in W_0\}$.

Now define a function $f\colon \RC(W,R) \rightarrow \RC(\R^3)$ by
\begin{equation*}
f(X) =  \sum_{x \in X \cap W_0} B_x.
\end{equation*}
Let $X \in \RC(W,R)$, and let $z$ be a point of $W_1$. We claim that $z
\in \ti{X}$ implies $D_z \subseteq f(X)$; further, if $zRx$, then $D_z
+ B_x$ is interior-connected. Indeed, if $z \in \ti{X}$, with $zRx$,
then $x \in X$. And since, by Step 1 of the above construction, every
rational point of $D_z$ lies in $B_x$ for some such $x$, it follows
that $D_z \subseteq \sum \{B_x \mid x \in X \cap W_0\} = f(X)$. The second
statement follows easily from the choice of the sets $B^1_x$ and the
fact that the sets $B_x$ are interior-connected.

We are now ready to show that $f$ is a Boolean algebra homomorphism,
and that $X \in \RC(W,R)$ is interior-connected if and only if $f(X)$
is interior-connected.  Trivially, $f(X+Y) = f(X)+f(Y)$; and since the
$B_x$ form a partition, $f(-X) = \sum_{x \in W_0 \setminus X} B_x =
-f(X)$. Now suppose $X \in \RC(W,R)$ is interior-connected, and let
$p, q$ be points in $\ti{f(X)}$. Then there exist points $p', q'$ in
the same components of $\ti{f(X)}$ as $p$, $q$, respectively, such
that, for some $k \geq 1$ and $x, y \in X \cap W_0$, we have $p' \in
B^k_x$ and $q' \in B^k_y$.  Since $X$ is interior-connected, we can
find sequences of points $x_0, \dots, x_m$ in $X \cap W_0$, and $z_1,
\dots, z_m \in \ti{X} \cap W_1$ such that $x = x_0$, $y = x_m$ and
$z_iRx_{i-1}$ and $z_iRx_{i}$, for all $1 \leq i \leq m$. But we have
shown above that the sets $D_{z_i}$ are subsets of $f(X)$, and that
the sets $D_{z_i} + B_{x_i}$ and $D_{z_i} + B_{x_{i-1}}$ are
interior-connected.  Hence, $p'$ and $q'$ lie in the same component of
$\ti{f(X)}$, whence $p$ and $q$ do as well. That is: $f(X)$ is
interior-connected, as required. Finally, suppose $X \in \RC(W,R)$ is
not interior-connected, so that we may find elements $x, y \in W_0$
lying in different components of $\ti{X}$. We show that $f(X)$ is not
interior-connected. For suppose otherwise.  Then there exists a
rational piecewise-linear arc $\alpha$ with endpoints in the sets $D
\cap \ti{(B^1_x)}$ and $D \cap \ti{(B^1_y)}$, and lying entirely in
$\ti{f(X)}$. But $\alpha$ occurs as some $\alpha_k$ in our
enumeration. It follows that there will be a first point $q'$ of
$\alpha_k$ lying in a set $\hat{B}^k_{y'}$ such that $x$ and $y'$ lie
in different components of $\ti{X}$; and there will be a last point
$p'$ of $\alpha_k$, occurring strictly before $q'$ and lying in a set
$\hat{B}^k_{x'}$. Obviously, $x'$ and $y'$ lie in different components
of $\ti{X}$.  Let $\alpha'$ be the interior segment of $\alpha_k$
between $p'$ and $q'$ (i.e.~without the end-points); thus $\alpha'$
does not intersect $\hat{B}^k$.  By construction of the sets $D_z$ ($z
\in W_1$), either $\alpha'$ lies entirely in $D_z$ for some $z \in
W_1 \setminus \{z_0\}$, or $\alpha'$ intersects $D_{z_0}$.  In the
former case, since $zRx'$ and $zRy'$, but $x'$ and $y'$ lie in
different components of $\ti{X}$, it follows that $z \notin
\ti{X}$. Thus, there exists $x'' \in W_0$ such that $zRx''$ and $x''
\notin X$. But then Step~2 in the above construction ensures that
$\alpha'\subseteq \alpha_k$ contains points of $B^{k+1}_{x''}\subseteq
B_{x''}$, contradicting the supposition that $\alpha_k \subseteq
\ti{f(X)}$.  On the other hand, if $\alpha' \cap D_{z_0} \neq
\emptyset$, then, since $X$ is not interior-connected, it follows that
$X \neq 1$, and so there certainly exists $x'' \in W_0$ such that
$z_0Rx''$ and $x'' \notin X$. By the same reasoning as before,
$\alpha'$ contains points of $B^{k+1}_{x''}\subseteq B_{x''}$,
contradicting the supposition that $\alpha_k \subseteq \ti{f(X)}$.

Now simply define an interpretation $\mathfrak{I}$ over $\RC(\R^3)$ by
setting $r^{\mathfrak{I}} = f(r^{\mathfrak{A}})$. It immediately
follows from the previous paragraph that $\phi$ is true in
$\mathfrak{I}$.
\end{proof}

This resolves the decidability and complexity of the problems
{\small $\Sat(\cBci,\RC(\R^n))$} and $\Sat(\cBci,\RCP(\R^n))$ for all $n \geq
3$.  Recall from \SECT\ref{sec:undecidability} that
$\Sat(\cL,\RCP(\R^n))$ is undecidable for all $n \geq 2$, where $\cL$
is any of \cBc, \cBCc{} or \cBCci; and recall from \SECT\ref{sec:2d}
that $\Sat(\cBci,\RCP(\R^2))$ is undecidable, and that
$\Sat(\cL,\RC(\R^2))$ is also undecidable, where $\cL$ is any of
\cBc, \cBci, \cBCc{} or \cBCci. At the time of writing, it is
not known whether any of the problems $\Sat(\cBc,\RC(\R^n))$,
$\Sat(\cBCc,\RC(\R^n))$ or $\Sat(\cBCci,\RC(\R^n))$, for $n \geq 3$,
is decidable. The best currently available lower bound can be found
in~\cite{iscloes:kp-hwz10}, where all three problems are shown to be
\ExpTime-hard.

\section*{Acknowledgements}
Roman Kontchakov and Michael Zakharyaschev acknowledge the support of
the EPSRC, grant ref.~EP/E034942/1. Yavor Nenov acknowledges the
support of the EPSRC, DTA account~FA01413.  Ian Pratt-Hartmann
acknowledges the support of the EPSRC, grant ref.~EP/E035248/1, and
expresses his thanks to the Department of Mathematics and Computer
Science, University of Wroc{\l}aw, and the Transregional Collaborative
Research Center SFB/TR 8 ``Spatial Cognition'', University of Bremen,
for their hospitality during the writing of this paper.

\bigskip
\begin{flushleft}
\begin{minipage}{5cm}
\begin{tabbing}
Yavor Nenov, Ian Pratt-Hartmann\\
School of Computer Science\\ 
University of Manchester\\
Manchester, M13 9PL, United Kingdom
\end{tabbing}
\end{minipage}

\bigskip

\begin{minipage}{5cm}
\begin{tabbing}
Roman Kontchakov, Michael Zakharyaschev\\
Department of Computer Science and Information Systems\\
Birkbeck, University of London\\ Malet Street, London, WC1E 7HX, United Kingdom
\end{tabbing}
\end{minipage}
\end{flushleft}

\bibliographystyle{plain}
\bibliography{utl}
\end{document}